\documentclass[journal]{IEEEtran}
\usepackage{amsmath,amssymb,amsthm, bm, bbm, xcolor, graphicx, physics,float,multirow,booktabs,array}
\usepackage[mathcal]{euscript}
\usepackage[utf8]{inputenc}
\usepackage{cite}
\usepackage{aligned-overset}
\newtheorem{theorem}{Theorem}
\newtheorem{lemma}{Lemma}
\newtheorem{assumption}{Assumption}

\newtheorem{Corollary }{Corollary}

\newtheorem{theoremgrp}{Theorem}
\newcounter{subtheorem}
\counterwithin{theoremgrp}{subtheorem}

\newcommand{\theoremgroup}{\refstepcounter{subtheorem}\refstepcounter{theorem}}

\allowdisplaybreaks

\def\boxit#1{%
	\smash{\fboxsep=0pt\llap{\rlap{\fbox{\strut\makebox[#1]{}}}~}}\ignorespaces
}

\DeclareMathOperator*{\argmax}{argmax}
\DeclareMathOperator*{\argmin}{argmin}

\begin{document}	
	\title{Optimal Aggregation Strategies for Social Learning over Graphs}
	\author{Ping~Hu,
		Virginia~Bordignon,
		Stefan~Vlaski,
		and~Ali~H.~Sayed
		\thanks{This work was supported in part by the Swiss National Science Foundation under Grant 205121-184999. A short conference version of this paper appears in \cite{Ping2022:Optimal}.}
		\thanks{Ping Hu, Virginia Bordignon, and Ali H. Sayed are with the School of Engineering, \'Ecole Polytechnique F\'ed\'erale de Lausanne (EPFL), CH1015 Lausanne, Switzerland (e-mail: ping.hu@epfl.ch; virginia.bordignon@epfl.ch; ali.sayed@epfl.ch).}
		\thanks{Stefan Vlaski is with the Department of Electrical and Electronic Engineering, Imperial College London, SW7 2BT London, U.K. (e-mail: s.vlaski@imperial.ac.uk).}}
	
	\maketitle
	
	\begin{abstract}
		Adaptive social learning is a useful tool for studying distributed decision-making problems over graphs. 
		This paper investigates the effect of combination policies on the performance of adaptive social learning strategies. Using large-deviation analysis, it first derives a bound on the steady-state error probability and characterizes the optimal selection for the Perron eigenvectors of the combination policies. It subsequently studies the effect of the combination policy on the transient behavior of the learning strategy by estimating the adaptation time in the low signal-to-noise ratio regime. In the process, it is discovered that, interestingly, the influence of the combination policy on the transient behavior is insignificant, and thus it is more critical to employ policies that enhance the steady-state performance. The theoretical conclusions are illustrated by means of computer simulations. 
	\end{abstract}
	
	
	\section{Introduction}
	\label{sec:intro}
	Social learning is a distributed inference process over graphs where agents  work collaboratively to identify the true state of nature from a set of admissible hypotheses \cite{Molavi:2018kq}. In each step, agents update their {\em beliefs} locally using streaming private observations, and then combine their beliefs with information received from neighbors using a \emph{combination policy}. There exist several useful variations of social learning algorithms, including those based on linear updates \cite{Jadbabaie:2012ii,Zhao:2013wa,Salami:2017jf}, log-linear updates \cite{Lalitha:2018ej,Nedic:2017bg,Hare2021General,Matta2020Interplay,shi2020distributed} and the min-rule \cite{Mitra:2020fi}. All these variants provide asymptotic learning guarantees under stationary environments where the underlying state is fixed. One useful feature of these learning procedures is the unanimity of the learning rules \cite{Molavi:2018kq}, which ensures that the effective weights assigned to each piece of independent observation are of the same order of magnitude. Consequently, the information from historical observations are stored in a uniform way, which means that more evidence in favor of the underlying state is collected over time. The cumulative evidence for a particular state can, however, hinder learning in face of a changing true state. The agents' stubbornness towards state changes during the learning process, which was observed in \cite{Bordignon:2020wx}, makes it imperative to develop new algorithmic variants for social learning in non-stationary environments.
	
	Motivated by this observation, the work \cite{Bordignon:2020wx} proposed an \emph{adaptive} social learning (ASL) algorithm, which, different from previous non-adaptive implementations \cite{Jadbabaie:2012ii,Zhao:2013wa,Salami:2017jf,Lalitha:2018ej,Nedic:2017bg,Hare2021General,Matta2020Interplay,shi2020distributed,Mitra:2020fi}, introduced a step-size parameter $\delta$ to control the amount of weighting given to recent observations in relation to past observations. Under this weighting mechanism, the agents become more sensitive to information contained in recent observations, and better equipped to track drifts in the statistical properties of the data. In particular, it was shown in \cite{Bordignon:2020wx} that the parameter $\delta$ controls a fundamental trade-off between the steady-state learning ability of an algorithm and its adaptation ability. It was found that in the slow adaptation  regime (i.e., with small $\delta$), the steady-state error probability decays exponentially with ${1}/{\delta}$. Moreover, the decaying rate (also called \emph{error exponent}) was observed to be affected by the \emph{eigenvector centrality} of the agents, which is a function of the graph topology and the combination policy employed by the social learning strategy. In this work, we would like to investigate more deeply the role of the combination policy on the behavior of adaptive social learning methods, as well as clarify optimal choices for the policy for faster transient behavior and lower steady-state error probability. 
	
	{\emph{Related works:}}  In the field of social learning, the effect of combination policies on the learning performance of some {\em non-adaptive} algorithms has been considered in previous studies \cite{jadbabaie2013information,Shahrampour:2016kk}. However, since the error probability converges to zero almost surely in the non-adaptive scenario, the role of the combination policy was only examined in the transient phase (i.e., only its effect on the speed of learning is studied). The main conclusion from these works is that an agent with better signal structure (in the sense of uniform informativeness \cite{jadbabaie2013information}) should be placed in a more centralized position in the network. In contrast to the non-adaptive scenario, the steady-state error probability is non-zero and dependent on the combination policy in the \emph{adaptive} social learning scenario. Therefore, both the transient and steady-state behavior need to be examined under different choices for the combination policy. 
	
	Another line of investigation relevant to our work is the field of distributed detection over multi-agent networks (see \cite{Viswanathan1997:Distirbuted,Blum1997:Distributed,Chamberland2003:Decentralized,matta2018estimation} for a brief review). Different from the social learning problem where agents may receive signals generated from distinct likelihoods, the distributed detection framework often assumes that the agent's observations are independent and identically distributed (i.i.d.). Two classical distributed detection strategies in the literature are: the consensus-based strategy that uses a \emph{decaying} step-size \cite{Bajovic2011:Distributed,Bajovic2012:Large,Bajovic2016:Distributed} and the diffusion-based  strategy that employs a \emph{constant} step-size \cite{Matta:2016kh,Matta:2016cb,Marano2021:Decision}. It was demonstrated in \cite{Matta:2016kh} that the diffusion-based detection strategy achieves a better adaptation ability than the consensus-based counterpart. The learning performance of both strategies, namely, the learning speed and the steady-state error probability, were shown to be dependent on the combination policy employed by the network \cite{Bajovic2011:Distributed,Bajovic2012:Large,Matta:2016kh,Matta:2016cb,Bajovic2016:Distributed}.
	
	{\emph{Contributions:}} Using techniques introduced in \cite{Bordignon:2020wx}, our first contribution is to extend the learning performance of the ASL strategy to a more general scenario \cite{Nedic:2017bg,Hare2021General,Matta2020Interplay,shi2020distributed} where the distribution of local observations received by an agent, namely the \emph{signal model}, may differ from the distributions known to this agent conditional on any possible hypothesis, namely the \emph{likelihood model}. Under this general scenario, we first derive a bound on the error exponent for the steady-state learning performance, and then construct the optimal centrality vector (i.e., Perron eigenvector of the combination matrix \cite{sayed2014adaptation}) that attains the upper bound of the error exponent (if it is achievable for the given social learning task). Our results for the i.i.d. case answer some of the questions posed earlier in \cite{Matta:2016cb} regarding the optimal choice of the combination policy. We also examine the effect of combination policies on the transient performance measured by the adaptation time. We show that in the low signal-to-noise ratio (SNR) regime, combination policies play a minor role in influencing the adaptation time. This indicates that, if the hypotheses are hard to distinguish, then it is sufficient to rely on a combination policy with better steady-state learning performance. 	
	
	\section{Problem setting}
	\label{sec:backgrounds}
	\subsection{Background}
	\textbf{Network model:} We consider a collection of $N$ agents, denoted by $\mathcal{N}=\{1,2,\dots,N\}$, working collectively to agree on a hypothesis that best explains the streaming and dispersed observations received by the group.  The communication network among agents is modeled as a directed graph, which is assumed to be strongly connected. An example of a strongly connected network is shown in Fig. \ref{fig: strong-connected network}.
	\begin{assumption}[\textbf{Strong connectivity of network}\footnote{In this paper, strong connectivity refers to a property of strongly connected networks. We require the existence of at least one self-loop, which ensures that the combination matrix will be primitive. This condition may not be used in some other studies on graph theory, such as \cite{newman2018networks,lewis2011network}. However, in the context of learning theory, this condition is not restrictive and is automatically satisfied in most cases of interest since agents naturally place some level of confidence on their own data.}]\label{asmp: assumption 1}
		The underlying graph of the network is strongly connected. That is, there exist paths between any two distinct agents in both directions, and at least one agent has a self-loop \cite{sayed2014adaptation}. \qed
	\end{assumption}
	\begin{figure}
		\centering
		\includegraphics[width=.7\linewidth]{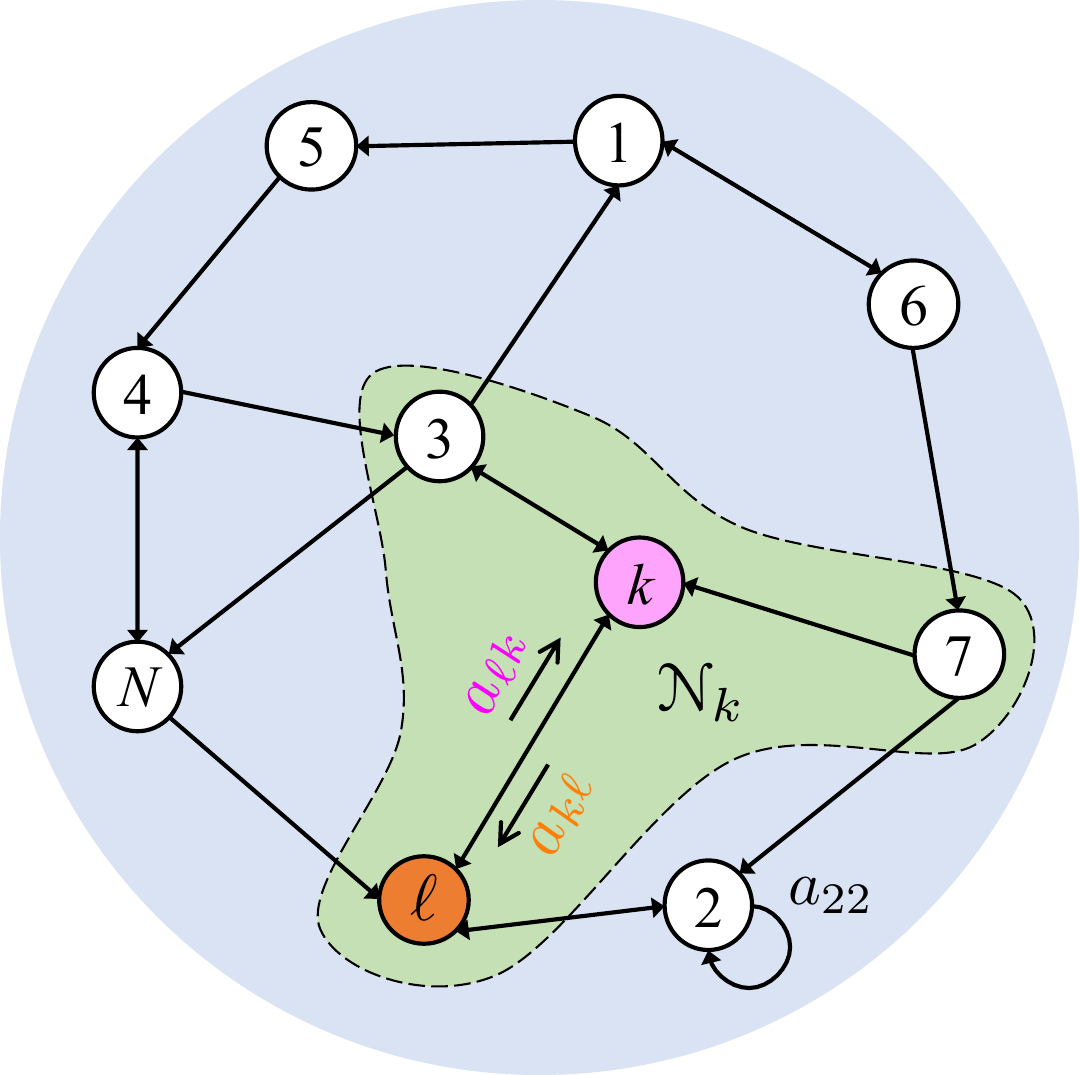}
		\caption{A strongly connected network consisting of $N$ agents.  The neighborhood of agent $k$ is marked by the area highlighted in green.}
		\label{fig: strong-connected network}
	\end{figure}
	\noindent The combination policy among agents is described by the matrix $A=[a_{\ell k}]$, where $a_{\ell k}$ is the weight that agent $k$ places on the information
	received from the neighboring agent $\ell$. We assume that the agents adopt a left-stochastic combination policy. Let $\mathcal{N}_k$ be the set of neighbors of agent $k$, then the combination matrix $A$ satisfies
	\begin{equation}
		A^\top\mathbbm{1}=\mathbbm{1}, \quad a_{\ell k}>0, \;\;\forall \ell\in\mathcal{N}_k, 
	\end{equation} 
	and $a_{\ell k}=0$ for $\ell\notin\mathcal{N}_k$, where $\mathbbm{1}$ denotes the $N$-dimensional vector of all ones. The strong connectivity of the graph ensures that the combination matrix $A$ is primitive. According to the Perron-Frobenius theorem \cite{sayed2014adaptation,sayed2014adaptive}, matrix $A$ has a single eigenvalue at $1$ and all other eigenvalues will be strictly inside the unit circle. Therefore, the second largest-magnitude eigenvalue of $A$ is strictly smaller than $1$. Moreover, the Perron eigenvector $\pi$ of matrix $A$ can be normalized to have strictly positive entries. That is,
	\begin{equation}\label{eq: Perron eigenvector}
		A\pi=\pi,\quad \mathbbm{1}^\top \pi=1,\quad\pi_k>0, \;\; \forall k\in\mathcal{N}.
	\end{equation}  

	\textbf{Observation model:} At each time instant $i$, each agent $k$ receives a private signal $\bm{\xi}_{k,i}$ belonging to a certain space $\mathcal{X}_k$. Note that we are utilizing boldface notation to emphasize that $ \boldsymbol{\xi}_{k, i} $ is random. The private signals of every agent, which are assumed to be statistically independent over time and space given a fixed true state of nature, are realizations of a random variable following an \emph{unknown} distribution $f_k$:
	\begin{equation}\label{eq: general signal model}
		\bm{\xi}_{k,i}\sim f_k,\quad \forall k\in\mathcal{N}.
	\end{equation}
	The joint observation profile at time instant $i$ generated by the network is denoted by $\bm{\xi}_i=(\bm{\xi}_{1,i},\bm{\xi}_{2,i},\dots,\bm{\xi}_{N,i})$, which is an i.i.d. sequence on the space $\mathcal{X}=\mathcal{X}_1\times\mathcal{X}_2\times\dots\times\mathcal{X}_N$ and distributed as $f=\prod_{k=1}^{N}f_k$ under the independence assumption of the private signals. Moreover, each agent $k$ has a family of local likelihood models $\{L_k(\cdot|\theta)\}$ parameterized by the hypothesis $\theta\in\Theta=\{\theta_1,\theta_2,\dots,\theta_{H}\}$. Among the given $H$ hypotheses, there is one true state of nature $\theta^\star\in\Theta$, referred to as the \emph{global truth} for the network. Without loss of generality, we assume $\theta^\star=\theta_1$. In addition, we assume that $\{L_k(\cdot|\theta)\}$ has the same support as $f_k$, namely, $\mathcal{X}_k$. The likelihood of a signal ${\xi}$ received by agent $k$ conditioned on hypothesis $\theta$ is denoted by
	\begin{equation}\label{eq: likelihood of an observation}
		L_k({\xi}|\theta),\quad \forall {\xi}\in\mathcal{X}_k.
	\end{equation}    
	Depending on whether the signal space $\mathcal{X}_k$ is continuous or discrete, the signal model $f_k$ and the likelihood models $\{L_k(\cdot|\theta)\}$ can be respectively probability density functions (pdfs) or probability mass functions (pmfs). 
	
	We consider a general learning scenario where the signal model $f_k$ may not match exactly any of the $H$ local likelihood models $L_k(\cdot|\theta)$. This is more general than many existing works (e.g., \cite{Jadbabaie:2012ii,Zhao:2013wa,Salami:2017jf,Lalitha:2018ej,Mitra:2020fi,Bordignon:2020wx,Shahrampour:2016kk,jadbabaie2013information}), where $f_k$ is taken as the likelihood model of the true hypothesis.
	That is, the signal $\bm{\xi}_{k,i}$ is a sample drawn according to the likelihood model $L_k(\cdot|\theta_1)$:
	\begin{equation}\label{eq: accurate likelihood model}
		\bm{\xi}_{k,i}\sim L_k(\cdot|\theta_1),\quad \forall k\in\mathcal{N}.
	\end{equation}
	Since in this case all agents possess knowledge of the true signal model, we refer to scenario \eqref{eq: accurate likelihood model} as the \emph{accurate signal model} scenario. In contrast, scenario \eqref{eq: general signal model} is referred to as the \emph{general signal model} scenario.
	The Kullback-Leiber (KL) divergence \cite{cover1999elements} between $f_k$ and $L_k(\cdot|\theta)$, denoted by $D_{\textnormal{KL}}(f_k\|L_k(\cdot|\theta))$, is a useful measure of the ``distance'' between relevant distributions. Without loss of generality, we assume the following regularity condition on KL divergences \cite{Jadbabaie:2012ii,Zhao:2013wa,Salami:2017jf,Lalitha:2018ej,Nedic:2017bg,Hare2021General,Matta2020Interplay,shi2020distributed,Mitra:2020fi,Bordignon:2020wx,Shahrampour:2016kk,jadbabaie2013information}. 
	\begin{assumption}[\textbf{Finiteness of KL divergences}]\label{asmp: assumption 2}
		For each hypothesis $\theta\in\Theta$ and for each agent $k\in\mathcal{N}$, $D_{\textnormal{KL}}(f_k\|L_k(\cdot|\theta))$ is finite.  \qed
	\end{assumption}
	\noindent Two hypotheses $\theta_m$ and $\theta_n$ are said to be \emph{observationally equivalent from the perspective of agent $k$} if $D_{\textnormal{KL}}(f_k\|L_k(\cdot|\theta_m))=D_{\textnormal{KL}}(f_k\|L_k(\cdot|\theta_n))$. The optimal hypothesis set for agent $k$ (also called the \emph{local truth}) is defined as the collection of hypotheses with the minimum KL divergence:
	\begin{equation}\label{eq: set of local truth}
		\Theta_k^\star\triangleq\argmin_{\theta\in\Theta} D_{\textnormal{KL}}(f_k\|L_k(\cdot|\theta)).
	\end{equation}
	Due  to the non-negativeness of KL divergences, it is clear that in the accurate signal model scenario, the global truth $\theta_1$ is also a local truth for all agents, i.e., $\theta_1\in\Theta_k^\star,\forall k\in\mathcal{N}$. However, in the general signal model scenario, some agents may fail to recognize the global truth due to the model discrepancy between $f_k$ and $L_k(\cdot|\theta_1)$. We denote the sets of agents whose local truth agrees or collides with the global truth as $\mathcal{N}_A$ and $\overline{\mathcal{N}}_A$ respectively, such that $\mathcal{N}=\mathcal{N}_A\cup\overline{\mathcal{N}}_A$. The goal of all agents is to learn the global truth by cooperation with their neighbors. We note that we are not going to discuss the multi-task decision-making problem where different groups of agents in the network try to identify different hypotheses. The interested readers are referred to \cite{Marano2021:Decision}. 
	
	The optimal hypothesis set for the group $\mathcal{N}_A$ is defined as the common hypotheses shared by all agents $k\in\mathcal{N}_A$:
	\begin{equation}
		\Theta^\star=\bigcap_{k\in\mathcal{N}_A}{\Theta_k^\star}.
	\end{equation}
	We impose the following identifiability assumption of the  global truth on the group $\mathcal{N}_A$.
	
	\begin{assumption}[\textbf{Identifiability}]\label{asmp: assumption 3}
		Hypothesis $\theta_1$ is the unique optimal hypothesis for the group $\mathcal{N}_A$, i.e., $\Theta^\star=\left\{\theta_1\right\}$. \qed
	\end{assumption}
	\noindent For the accurate signal model \eqref{eq: likelihood of an observation}--\eqref{eq: accurate likelihood model}, we have $\mathcal{N}_A=\mathcal{N}$, and thus  Assumption \ref{asmp: assumption 3} becomes the standard global identifiability assumption, namely,  $\{\theta_1\}=\cap_{k=1}^N\Theta_k^\star$, as considered in \cite{Jadbabaie:2012ii,Zhao:2013wa,Salami:2017jf,Lalitha:2018ej,Mitra:2020fi,Bordignon:2020wx,Shahrampour:2016kk,jadbabaie2013information}.
	\subsection{Adaptive social learning strategy}
	We first introduce the basic framework of social learning. 
		
	\emph{Motivating example \cite{Nedic:2017bg}:}  Consider a distributed source localization problem, where a network of $N$ agents receives noisy measurements of the distance to the source. Specifically, the private signal $\bm{\xi}_{k,i}$ received by agent $k$ at time instant $i$ is expressed as
		\begin{equation}\label{eq: private signal in the example}
			\bm{\xi}_{k,i}={\sf d}_{k,i}+\bm{n}_{k,i}
		\end{equation}
	where ${\sf d}_{k,i}$ denotes the distance between agent $k$ and the target measured at time instant $i$, and $\bm{n}_{k,i}$ is some zero-mean Gaussian noise. In the stationary environment where the target is assumed to be static, ${\sf d}_{k,i}$ is a constant for each agent $k$. The geographic region is partitioned into a collection of $H$ disjoint areas, and the hypothesis space $\Theta$ include all these possible locations.  Each agent $k$ constructs likelihood functions $\{L_k(\cdot|\theta)\}$ based on its sensor model. In principle, each agent could estimate its distance to the target from its local observations. However, their information is not enough to arrive at the coordinates for the target, since each agent can only conclude that the target lies on a circle of radius ${\sf d}_{k,i}$ around it. To achieve the goal of source localization, the agents would need to cooperate with each other.
	
	In social learning solutions, each agent $k$ holds a local belief vector $\bm{\mu}_{k,i}$ that represents a pmf over the set of hypotheses $\Theta$. Each component $\bm{\mu}_{k,i}(\theta)$ indicates the confidence of agent $k$ at time instant $i$ that $\theta$ is the true state of nature. In the context of source localization, $\bm{\mu}_{k,i}(\theta)$ denotes the agent $k$'s estimate of the probability that the target is located in area $\theta$. The belief vector is updated through the continuous flow of information in the network through an adaptation step and a combination step. More specifically, at each time instant $i$, agent $k$ receives a new signal $\bm{\xi}_{k,i}$ and uses it to compute an intermediate belief vector $\bm{\psi}_{k,i}$. In the \emph{non-adaptive }learning scenario, the Bayes rule is employed, which computes $\bm{\psi}_{k,i}$ according to
	\begin{equation}\label{eq: Bayes update}
		\bm{\psi}_{k,i}(\theta)=\dfrac{\bm{\mu}_{k,i-1}(\theta) L_k(\bm{\xi}_{k,i}|\theta)}{\sum_{\theta^\prime\in\Theta}\bm{\mu}_{k,i-1}(\theta^\prime) L_k(\bm{\xi}_{k,i}|\theta^\prime)}
	\end{equation}
 	for each $\theta\in\Theta$. Eq. \eqref{eq: Bayes update} describes how the agents refine their estimates of the target's location with the latest measurement in the source localization example. After this adaptation step, agent $k$ aggregates the intermediate beliefs of its neighbors following a certain pooling protocol in order to update the local belief vector $\bm{\mu}_{k,i}$. Different protocols for belief aggregation lead to various social learning algorithms in the literature \cite{Jadbabaie:2012ii,Zhao:2013wa,Salami:2017jf,Lalitha:2018ej,Nedic:2017bg,Hare2021General,Matta2020Interplay,shi2020distributed,Mitra:2020fi}. Examples of useful pooling rules appear in \cite{koliander2022fusion}. Next, we describe one fusion rule in the context of non-stationary environments.
 
 First, we note that the ASL algorithm introduced in \cite{Bordignon:2020wx} is an important variant of social learning developed for non-stationary conditions. Within the previous distributed source localization example, the target might move to another location at some time instant $i_0$ and consequently, the distance ${\sf d}_{k,i}$ in \eqref{eq: private signal in the example} becomes a different value after $i\geq i_0$. It is necessary for the network to quickly track the new location of the target in practice. Adaptation is a desirable feature of distributed learning strategies \cite{sayed2014adaptive,sayed2014adaptation}. To improve the adaptation ability, a modified adaptive update (compared with \eqref{eq: Bayes update}) is proposed in the ASL algorithm, where the weights assigned to the past and new information for constructing the intermediate belief are controlled through a small positive parameter $\delta$:
\begin{equation}\label{eq: ASL0}
	\bm{\psi}_{k,i}(\theta)=\dfrac{\bm{\mu}_{k,i-1}^{1-\delta}(\theta) L_k^\delta(\bm{\xi}_{k,i}|\theta)}{\sum_{\theta^\prime\in\Theta}\bm{\mu}_{k,i-1}^{1-\delta}(\theta^\prime) L_k^\delta(\bm{\xi}_{k,i}|\theta^\prime)}.
\end{equation}
The motivation for this adaptation formula has been elaborated from different perspectives in \cite{Bordignon:2020wx}. In particular, we can get \eqref{eq: ASL0} as the solution to the following optimization problem:
\begin{equation}\label{eq: ASL as optimization solution}
	\bm{\psi}_{k,i}=\argmin_{p\in\Delta_\theta}\Big\{(1-\delta)D_{\textnormal{KL}}(p\|\bm{\mu}_{k,i-1})+\delta D_{\textnormal{KL}}(p\|\bm{\mu}_{k,i}^{\sf lik})\Big\}
\end{equation}
where $\Delta_\theta$ is the set of all pmfs on the hypothesis set $\Theta$, and $\bm{\mu}_{k,i}^{\sf lik}$ is the likelihood pmf involving the new observation $\bm{\xi}_{k,i}$:
\begin{equation}
	\bm{\mu}_{k,i}^{\sf lik}(\theta)=\frac{L_k(\bm{\xi}_{k,i}|\theta)}{\sum_{\theta^\prime\in\Theta}L_k(\bm{\xi}_{k,i}|\theta^\prime)},\quad\forall\theta\in\Theta.
\end{equation} The first and second KL divergences in \eqref{eq: ASL as optimization solution} describe respectively the consistency between the intermediate belief and the past information (captured by $\bm{\mu}_{k,i-1}$), and the consistency between the intermediate belief and the new information (represented by $\bm{\mu}_{k,i}^{\sf lik}$). The trade-off between these two consistency costs is adjusted through the parameter $\delta$. Following \eqref{eq: ASL0}, the belief aggregation step employs the log-linear rule, which generates the local belief vector $\bm{\mu}_{k,i}$ as follows: 
\begin{equation}\label{eq: ASL 1}
	\bm{\mu}_{k,i}(\theta)=\dfrac{\exp{\sum_{\ell\in\mathcal{N}_k}a_{\ell k}\log{\bm{\psi}_{\ell,i}(\theta)}}}{\sum_{\theta^\prime\in\Theta}\exp{\sum_{\ell\in\mathcal{N}_k}a_{\ell k}\log{\bm{\psi}_{\ell,i}(\theta^\prime)}}}
\end{equation} 
	for each $\theta\in\Theta$. This pooling rule \eqref{eq: ASL 1} can be obtained as the solution to the following optimization problem:
	\begin{equation}
		\bm{\mu}_{k,i}=\argmin_{p\in\Delta_\theta}\Bigg\{\sum_{\ell=1}^N a_{\ell k}D_{\textnormal{KL}}(p\|\bm{\psi}_{\ell,i})\Bigg\}.
\end{equation}
Since its inception, the ASL algorithm has been applied to different tasks, such as discovering influencers in opinion formation over online social networks \cite{shumovskaia2023discovering} and solving image classification problems involving heterogeneous classifiers \cite{bordignon2021learning}. 
	
	To avoid trivial cases, we assume that for all agents, the initial belief on each hypothesis $\theta\in\Theta$ is non-zero, i.e., $\bm{\mu}_{k,0}(\theta)>0,\forall k\in\mathcal{N}$. This is because if $\bm{\mu}_{\ell,0}(\theta)=0$ for some $\ell\in\mathcal{N}$ and $\theta\in\Theta$, then from \eqref{eq: ASL0} and \eqref{eq: ASL 1}, we have $\bm{\mu}_{k,N-1}(\theta)=0$ for all agents due to Assumption \ref{asmp: assumption 1}. Hence, hypothesis $\theta$ will end up being excluded from the social learning process. 
	\begin{assumption}[\textbf{Positive initial belief}]\label{asmp: assumption 4}
		For each hypothesis $\theta\in\Theta$ and for each agent $k\in\mathcal{N}$, the initial belief $\bm{\mu}_{k,0}(\theta)$ is positive. \qed
	\end{assumption}
	\noindent Under Assumptions \ref{asmp: assumption 1}--\ref{asmp: assumption 4}, we discuss the steady-state learning performance of the ASL strategy for the general signal model \eqref{eq: general signal model}--\eqref{eq: likelihood of an observation}. 
	\subsection{Steady-state learning performance}
	We start by introducing the log-belief ratios  $\bm{\nu}_{k,i}(\theta)$ and $\bm{\lambda}_{k,i}(\theta)$ for all $k\in\mathcal{N}$ and $\theta\in\Theta$:
	\begin{equation}
		\bm{\nu}_{k,i}(\theta)\triangleq \log\frac{\bm{\psi}_{k,i}(\theta_1)}{\bm{\psi}_{k,i}(\theta)},\quad \bm{\lambda}_{k,i}(\theta)\triangleq\log{\frac{
				\bm{\mu}_{k,i}(\theta_1)}{\bm{\mu}_{k,i}(\theta)}}.
	\end{equation}
	 Each agent makes a decision about the underlying state based on its belief. One natural option is to select the hypothesis that maximizes the belief. For each agent $k\in\mathcal{N}$, the instantaneous error probability of social learning at time instant $i$ is defined as
		\begin{equation}	\label{eq: instantaneous error probability definition}
		p_{k,i}\triangleq \mathbb{P}\Big[\argmax_{\theta\in\Theta}\bm{\mu}_{k,i}(\theta)\neq \theta_1 \Big]
		= \mathbb{P}\Big[\exists\theta\neq\theta_1:\bm{\lambda}_{k,i}(\theta)\leq 0 \Big].
	\end{equation} 
We also introduce the log-likelihood ratio $\bm{x}_{k,i}(\theta)$ for all $k\in\mathcal{N}$ and $\theta\in\Theta$:
\begin{equation}
	\bm{x}_{k,i}(\theta)
	\triangleq\log{\frac{L_k(\bm{\xi}_{k,i}|\theta_1)}{L_k(\bm{\xi}_{k,i}|\theta)}}.
\end{equation}
Using the variables $\bm{\nu}_{k,i}(\theta)$, $\bm{\lambda}_{k,i}(\theta)$ and $\bm{x}_{k,i}(\theta)$ in the logarithmic domain, the ASL algorithm represented by \eqref{eq: ASL0} and \eqref{eq: ASL 1} can be rewritten as a two-step linear recursion:
\begin{equation}\label{eq: ASL}
	\begin{cases}
		\bm{\nu}_{k,i}(\theta)=(1-\delta)\bm{\lambda}_{k,i-1}(\theta)+\delta \bm{x}_{k,i}(\theta)\\[3pt]
		\bm{\lambda}_{k,i}(\theta)=\sum\limits_{\ell\in\mathcal{N}_k}a_{\ell k}\bm{\nu}_{\ell,i}(\theta)
	\end{cases}
\end{equation}
which has the form of a standard \emph{diffusion} learning rule \cite{sayed2014adaptation,sayed2014adaptive}. 
Iterating the recursion in \eqref{eq: ASL}, we obtain
	\begin{align}
		\nonumber
		\bm{\lambda}_{k,i}(\theta)=&\; (1-\delta)^i\sum_{\ell=1}^{N}[A^i]_{\ell k}\bm{\lambda}_{\ell,0}(\theta)\\
		\label{eq: ASL expanded recursion}
		&+\delta\sum_{m=0}^{i-1}\sum_{\ell=1}^N(1-\delta)^m[A^{m+1}]_{\ell k}\bm{x}_{\ell,i-m}(\theta).
	\end{align}
	The first term in \eqref{eq: ASL expanded recursion} involves the initial belief vector and it decays as $i$ grows. In order to evaluate the steady-state learning performance, it suffices to focus on the second term:
	\begin{equation}\label{eq: lambda_hat}
		\widehat{\bm{\lambda}}_{k,i}(\theta)\triangleq \delta\sum_{m=0}^{i-1}\sum_{\ell=1}^N(1-\delta)^m[A^{m+1}]_{\ell k}\bm{x}_{\ell,i-m}(\theta).
\end{equation}
Given that the random variables $\bm{x}_{k,i}(\theta)$ are i.i.d. across time, and that our analysis concerns only the distribution of partial sums associated with these terms, it is useful to introduce the following random variable:
\begin{equation}\label{eq: lambda_tilde}
		\widetilde{\bm{\lambda}}_{k,i}(\theta)\triangleq \delta\sum_{m=0}^{i-1}\sum_{\ell=1}^N(1-\delta)^m[A^{m+1}]_{\ell k}\bm{x}_{\ell,m+1}(\theta)
\end{equation} 
where the summation in \eqref{eq: lambda_hat} is taken in reversed order. By repeating the same arguments used in \cite{Matta:2016kh,Matta:2016cb,Bordignon:2020wx}, we can show that $\widehat{\bm{\lambda}}_{k,i}(\theta)$ and $\widetilde{\bm{\lambda}}_{k,i}(\theta)$ share the same distribution. Formally, 
\begin{equation}\label{eq: equality in distribution}
	\widehat{\bm{\lambda}}_{k,i}(\theta)\overset{\textnormal{d}}{=}\widetilde{\bm{\lambda}}_{k,i}(\theta)
\end{equation}
where the symbol $\overset{\textnormal{d}}{=}$ denotes equality in distribution. Furthermore, the random variable $\widetilde{\bm{\lambda}}_{k,i}(\theta)$ converges almost surely as $i\to\infty$ under Assumptions \ref{asmp: assumption 1}--\ref{asmp: assumption 2}. Therefore, there exists a steady-state random variable termed as  \emph{steady-state log-belief ratio}: 
\begin{equation}\label{eq: limiting random variable}
	\widetilde{\bm{\lambda}}_{k,\infty}(\theta)\triangleq \delta\sum_{m=0}^{\infty}\sum_{\ell=1}^N(1-\delta)^m[A^{m+1}]_{\ell k}\bm{x}_{\ell,m+1}(\theta)
\end{equation}
such that 
\begin{equation}\label{eq: convergence to a limiting random variable}
	\widehat{\bm{\lambda}}_{k,i}(\theta)\overset{i\to\infty}{\rightsquigarrow}\widetilde{\bm{\lambda}}_{k,\infty}(\theta)
\end{equation}
where the symbol $\rightsquigarrow$ means convergence in distribution. In view of \eqref{eq: instantaneous error probability definition}, the relation in \eqref{eq: convergence to a limiting random variable} allows us to define the steady-state error probability for each agent $k$ as:
\begin{equation}\label{eq: steady-state error probability}
	p_k\triangleq\mathbb{P}\big[\exists\theta\neq\theta_{1}: \widetilde{\bm{\lambda}}_{k,\infty}(\theta)\leq 0\big].
\end{equation}
Using similar analytical tools to the ones employed in \cite{Bordignon:2020wx}, we can prove that for small $\delta$, the steady-state error probability $p_k$ obeys a Large Deviations Principle (LDP) \cite{dembo1998applications,den2008large} for some error exponent that is related to the combination policy. For the benefit of the reader, we recall here that a process $\bm{y}_\delta$ is said to obey an LDP if the following limit exists \cite{dembo1998applications,den2008large}:
\begin{equation}\label{eq: LDP}
	\underset{\delta\to 0}{\lim\;}\delta \ln\mathbb{P}\big[\bm{y}_\delta\in\Gamma\big]=-\underset{\gamma\in\Gamma}{\inf}I(\gamma)\triangleq -I_\Gamma
\end{equation}  
for some $I(\gamma)$ that is called the \emph{rate function}, where $\Gamma$ is an arbitrary set. Equivalently,
\begin{equation}
	\mathbb{P}\left[\bm{y}_\delta\in\Gamma\right]\doteq e^{-(1/\delta)I_\Gamma}
\end{equation}
where the symbol $\doteq$ denotes equality to the \emph{leading} order in the exponent as $\delta$ goes to zero. For the social learning problem, the leading-order exponent $I_\Gamma$ corresponding to the error probability is also referred to as \emph{error exponent}.

For our subsequent analysis on the error exponent, we also need to introduce the network average log-likelihood ratio for all $\theta\in\Theta$:
	\begin{equation}\label{eq: x_ave}
		\bm{x}_{{\sf ave},i}(\pi,\theta)\triangleq\sum_{k=1}^N{\pi_k\bm{x}_{k,i}(\theta)}.
	\end{equation}  
	We note that the combination policy, which is the main subject of interest in this work, plays a key role in weighting the network average quantity defined above. Let $\mathbb{E}$ and $\mathbb{P}$ denote the expectation and probability operators relative to the joint signal model $f$, respectively. The expectation of $\bm{x}_{k,i}(\theta)$ is given by 
	\begin{equation}\label{eq: expectation}
		d_k(\theta)\triangleq\mathbb{E} \left[{\bm{x}}_{k,i}(\theta)\right]
		=D_{\textnormal{KL}}({f_k}\|L_k(\cdot|\theta))-D_{\textnormal{KL}}({f_k}\|L_k(\cdot|\theta_1)).
	\end{equation}
	Assumption \ref{asmp: assumption 2} ensures that $d_k(\theta)$ is finite, i.e., $\abs{d_k(\theta)}<\infty$ for all $k\in\mathcal{N}$ and for all $\theta\in\Theta$. By definition \eqref{eq: set of local truth}, we know that if $d_k(\theta)\geq0$ for all $\theta\neq\theta_1$, then the global truth  $\theta_1$ is also a local truth for agent $k$, i.e., $\theta_1\in\Theta_k^\star$.  Otherwise, it conflicts with the local truth at agent $k$, i.e., $\theta_1\notin\Theta_k^\star$. From \eqref{eq: x_ave}, the expectation of $\bm{x}_{{\sf ave},i}(\pi,\theta)$ can be written as 
	\begin{equation}\label{eq: mave}
		{\sf m_{ave}}(\pi,\theta)\triangleq\mathbb{E}\left[\bm{x}_{{\sf ave},i}(\pi,\theta)\right]=\sum_{k=1}^N{\pi_k d_k(\theta)}.
	\end{equation}
	Let $\Lambda_k(t;\theta)$ and $\Lambda_{\sf ave}(t;\pi,\theta)$ denote the Logarithmic Moment Generating Functions (LMGFs) of variables $\bm{x}_{k,i}(\theta)$ and $\bm{x}_{{\sf ave},i}(\pi,\theta)$, respectively:
	\begin{equation}\label{eq: LMGF}
		\Lambda_k(t;\theta)\triangleq\log\mathbb{E}\Big[e^{t\bm{x}_{k,i}(\theta)}\Big],
	\end{equation}
	\begin{equation}\label{eq: LMGF_ave}
		\Lambda_{\sf ave}(t;\pi,\theta)\triangleq\log\mathbb{E}\Big[e^{t\bm{x}_{{\sf ave},i}(\pi,\theta)}\Big]=\sum_{k=1}^N\Lambda_k(\pi_k t;\theta).
	\end{equation}
	We note that since the random variables $\bm{x}_{k,i}(\theta)$ are i.i.d. across time, $\Lambda_k(t;\theta)$ and $\Lambda_{\sf ave}(t;\pi,\theta)$ are independent of time. For this reason, the subscript $i$ pertaining to the time instant is not required. One fundamental property of LMGFs states that $\Lambda_{k}(t;\theta)$ is an alternative representation for the probability distribution of $\bm{x}_{k,i}(\theta)$. Hence, it captures the informativeness of agent $k$ on learning the global truth $\theta_1$. 
	
	 In the following theorem, we extend two important results on the steady-state learning performance in the small-$\delta$ regime to the general signal model; these results were previously established, albeit only for the accurate signal model in \cite{Bordignon:2020wx}.
	 
	\begin{theorem}[\textbf{Steady-state learning performance}\footnote{The asymptotic normality of steady-state log-beliefs in the small-$\delta$ regime provided in \cite{Bordignon:2020wx} can also be established for the general signal model.}] \label{theorem: learning performance for the general model}
		Under Assumptions \ref{asmp: assumption 1}--\ref{asmp: assumption 4} for the general signal model \eqref{eq: general signal model}--\eqref{eq: likelihood of an observation}, we consider a combination policy with Perron eigenvector $\pi$. In the small-$\delta$ regime, the steady-state log-belief ratio $\widetilde{\bm{\lambda}}_{k,\infty}(\theta)$
		 converges to ${\sf m_{ave}}(\pi,\theta)$ in probability as $\delta$ approaches 0:
		\begin{equation}
			\widetilde{\bm{\lambda}}_{k,\infty}(\theta)\xrightarrow{\delta\to 0}{\sf m_{ave}}(\pi,\theta) \quad \textnormal{in probability}
		\end{equation}
		for all $\theta\in\Theta$. Therefore, if the Perron eigenvector $\pi$ satisfies 
		\begin{equation}\label{eq: consistency condition}
			{\sf m_{ave}}(\pi,\theta)>0,\quad \forall \theta\neq\theta_1,
		\end{equation}
		we have:
		\begin{enumerate}
			\item[i)] \textbf{Consistency of learning\footnote{We note that the assumption on the independence of local observations over space is not necessary for the consistency of learning \cite{Bordignon:2020wx}.}:} 
			The steady-state error probability $p_k$ converges to 0 as $\delta$ goes to 0 by definition \eqref{eq: steady-state error probability}. This means that all agents learn the global truth successfully.
			\item[ii)] \textbf{Error exponent:} Assume $\Lambda_{k}(t;\theta)<\infty,\forall t\in\mathbb{R}$ for all $k\in\mathcal{N}$ and $\theta\neq\theta_1$. The steady-state error probability $p_k$ obeys an LDP with rate $\frac{1}{\delta}$ and error exponent $\Phi(\pi)$, i.e.,
			\begin{equation}\label{eq: LDP in theorem 1}
				p_k\doteq e^{-\Phi(\pi)/\delta},
			\end{equation}
			where
			\begin{equation}\label{eq: error exponent expression}
				\Phi(\pi)\triangleq\min_{\theta\neq\theta_1}\Phi(\pi,\theta).
			\end{equation}
			The \textnormal{$\theta$-related error exponent} $\Phi(\pi,\theta)$ is given by
			\begin{equation}\label{eq: rate function}
				\Phi(\pi,\theta)\triangleq-\inf_{t\in\mathbb{R}}\phi(t;\pi,\theta)
			\end{equation}
			with
			\begin{equation}\label{eq: phi}
				\phi(t;\pi,\theta)\triangleq\int_{0}^t\frac{\Lambda_{\sf ave}(\tau;\pi,\theta)}{\tau} d\tau.
			\end{equation}
		\end{enumerate}
	\end{theorem} 
	\begin{proof}
		The performance analysis on the ASL strategy for the accurate signal model \eqref{eq: likelihood of an observation}--\eqref{eq: accurate likelihood model} has been established in \cite{Bordignon:2020wx}. 
		From the analysis there, the steady-state learning performance is determined by the statistical properties of the log-likelihood ratio involving the true hypothesis and an alternative one. We note that there are some differences in notation between \cite{Bordignon:2020wx} and this work. For instance, the $H$ hypotheses are numbered from $0$ to $H-1$, and the true state is denoted by $\theta_0$ in \cite{Bordignon:2020wx}. Using the notation introduced in this paper, the key condition to the proof in \cite{Bordignon:2020wx} is that the log-likelihood ratio $\bm{x}_{k,i}(\theta)$ has a finite mean for all $k\in\mathcal{N}$ and $\theta\in\Theta$ (see Lemma 1 in \cite{Bordignon:2020wx}). That is,
			\begin{equation}\label{eq: finite mean}
				\mathbb{E}\big[\bm{x}_{k,i}(\theta)\big]\triangleq\mathbb{E}_{L_k(\cdot|\theta_1)}\big[\bm{x}_{k,i}(\theta)\big]<\infty,
			\end{equation}
		 which is ensured by the assumption of finite KL divergences (i.e., Assumption 1 in \cite{Bordignon:2020wx}). 
			The only difference in the analysis for a general signal model \eqref{eq: general signal model}--\eqref{eq: likelihood of an observation} is that, we need to examine the statistical properties of $\bm{x}_{k,i}(\theta)$ conditioned on a general model $f_k$ which might be different from the likelihood model $L_k(\cdot|\theta_1)$. However, under Assumption \ref{asmp: assumption 2}, we have $\mathbb{E}[\bm{x}_{k,i}(\theta)]<\infty$ as shown in \eqref{eq: expectation}. Therefore, the finite-mean condition \eqref{eq: finite mean} continues to hold for the general signal model:
			\begin{equation}
				\mathbb{E}\big[\bm{x}_{k,i}(\theta)\big]\triangleq\mathbb{E}_{ f_k}\big[\bm{x}_{k,i}(\theta)\big]<\infty.
			\end{equation}
			Consequently, Theorem \ref{theorem: learning performance for the general model} is established by repeating the proof of the steady-state learning performance developed in \cite{Bordignon:2020wx} by substituting the expectation w.r.t. $L_k(\cdot|\theta_1)$ by that w.r.t. $f_k$.
	\end{proof}
	\section{Maximizing error exponent}
	\label{sec:error exponent}
	In this section, we discuss the effect of combination policies on the steady-state learning accuracy of the ASL strategy. From Theorem \ref{theorem: learning performance for the general model}, the error exponent $\Phi(\pi)$ plays a crucial role in the steady-state error probability. According to \eqref{eq: error exponent expression}--\eqref{eq: phi}, $\Phi(\pi)$ is influenced by the Perron eigenvector of the combination policy through the LMGF $\Lambda_{\sf ave}(t;\pi,\theta)$ defined in \eqref{eq: LMGF_ave}. A Perron eigenvector that delivers a larger error exponent is beneficial for reducing the steady-state error probability in the slow adaptation regime. To find the best Perron eigenvectors that provide the largest error exponent for the given learning task, we formulate the  following  optimization problem:
	\begin{align}\label{eq: objective function}
		&\max_{\pi}\;\Phi(\pi)\\ \label{eq: constraint 1}
		\text{s.t. }& \mathbbm{1}^\top\pi =1, \;\pi_k>0, \quad \forall k\in\mathcal{N},\\\label{eq: constraint 2}
		&{\sf m_{ave}}(\pi,\theta)>0,\qquad\;\forall \theta\neq\theta_1.
	\end{align} 
	\noindent Here, constraints \eqref{eq: constraint 1} and \eqref{eq: constraint 2} are imposed to guarantee the strong connectivity of the network and the successful learning of the global truth according to Theorem \ref{theorem: learning performance for the general model}.
	We denote the set of all feasible solutions to the optimization problem \eqref{eq: objective function}--\eqref{eq: constraint 2} by $\Pi$:
	\begin{equation}\label{eq:set of  feasible solution}
		\Pi=\left\{\pi: \pi \text{ satisfies } \eqref{eq: constraint 1}\text{ and }\eqref{eq: constraint 2}\right\}.
	\end{equation}
	It is clear from \eqref{eq: error exponent expression}--\eqref{eq: phi} that the design of $\pi$ relates to the individual LMGFs $\Lambda_{k}(t;\theta)$, which measure the ability of every agent $k$ to learn the global truth $\theta_1$, namely, its level of \emph{informativeness}. Before solving the optimization problem above, we first provide some useful definitions and preliminary results for the subsequent analysis. 
	
	\subsection{Preliminary definitions}
	\label{subsec: preliminary definitions}
	We classify the $N$ agents into different groups according to their informativeness. Agents in different groups play different roles in the learning performance.
	For each wrong hypothesis $\theta\neq\theta_1$, we denote the sets $\mathcal{N}^{{U}}(\theta)$, $\mathcal{N}^{{I}}(\theta)$, and  $\mathcal{N}^{{C}}(\theta)$ as the collections of \emph{uninformative} agents, \emph{informative} agents and \emph{conflicting} agents with respect to $\theta$, respectively:
	\begin{align} \label{eq: uninformative agent}
		\mathcal{N}^{{U}}(\theta)\triangleq&\left\{k\in\mathcal{N}: \Lambda_k(t;\theta)\equiv0\right\},\\
		\label{eq: informative agent}
		\mathcal{N}^{{I}}(\theta)\triangleq&\left\{k\in\mathcal{N}: \Lambda_k(t;\theta)\not\equiv0,d_{k}(\theta)>0\right\},\\
		\label{eq: conflicting agent}
		\mathcal{N}^{{C}}(\theta)\triangleq&\left\{k\in\mathcal{N}: \Lambda_k(t;\theta)\not\equiv0,d_{k}(\theta)\leq 0\right\},
	\end{align}
	where the symbol $\equiv$ denotes that $\Lambda_{k}(t;\theta)=0$ for all $t\in\mathbb{R}$, and the symbol $\not\equiv$ means that $\Lambda_{k}(t;\theta)\neq 0$ for some $t\in\mathbb{R}$. According to the definitions above, an agent $k$ is $\theta$-uninformative if the likelihoods conditioned on hypotheses $\theta_1$ and $\theta$ are the same for all local observations (i.e., $L_k(\cdot|\theta) = L_k(\cdot|\theta_1)$ almost everywhere), and is $\theta$-informative if hypothesis $\theta_1$ is more consistent with its local observations than hypothesis $\theta$. It is $\theta$-conflicting if its information associated to hypothesis $\theta$ is detrimental for learning the global truth $\theta_1$. This point will be self-evident later when we discuss the bounds of error exponents. Moreover, it is clear that if $k\in\overline{\mathcal{N}}_A$ (i.e., $\theta_1\not\in\Theta_k^\star$), then agent $k$ must be a conflicting agent for some hypothesis $\theta$. Another observation is that definition \eqref{eq: uninformative agent} of a $\theta$-uninformative agent $k$ requires a more stringent condition (i.e., $\Lambda_{k}(t;\theta)\equiv 0$) than the observational equivalence of $\theta_1$ and $\theta$ from the perspective of agent $k$ (i.e., $d_k(\theta)=0$). In particular, due to the non-negativeness of KL divergence, $\Lambda_{k}(t;\theta)\equiv 0$ implies that $d_k(\theta)=0$ under the accurate signal model \eqref{eq: likelihood of an observation}--\eqref{eq: accurate likelihood model}. We provide the following example to illustrate the aforementioned sets of agents.
	
	\emph{Example:} Consider a network of 4 sensor agents tasked with a binary detection problem. In practice, the likelihood models are usually constructed from a finite number of samples collected under the corresponding hypothesis. Assume that the signal model is $f_k={\sf N}(0,1)$ for all agents, where ${\sf N}(a,b)$ denotes the Gaussian pdf  with mean $a$ and variance $b$. Due to the limited number of samples, the likelihood models might be inaccurate and thus differ from the underlying signal model. Let us consider a group of unit-variance Gaussian likelihood models: $L_k(\cdot|\theta_1)={\sf N}(0.1,1)$ for all $k$, and 
		\begin{align}
			&L_1(\cdot|\theta_2)={\sf N}(-0.1,1), &&L_2(\cdot|\theta_2)={\sf N}(0.2,1),\\
			&L_3(\cdot|\theta_2)={\sf N}(0,1),  &&L_4(\cdot|\theta_2)={\sf N}(0.1,1).
		\end{align}
		Under the ASL protocol, the expectations $d_k(\theta_2)$ are given by $d_1(\theta_2)=d_4(\theta_2)=0$, $d_2(\theta_2)=0.03$ and $d_3(\theta_2)=-0.01$.
	The LMGFs $\Lambda_k(t;\theta_2)$ for all $k\in\mathcal{N}$ are presented in Fig. \ref{fig: LMGF}. According to definitions \eqref{eq: uninformative agent}--\eqref{eq: conflicting agent}, we have $\mathcal{N}^{U}(\theta_2)=\{4\}$, $\mathcal{N}^{I}(\theta_2)=\{2\}$ and $\mathcal{N}^{C}(\theta_2)=\{1,3\}$. Therefore, the learning performance of this network will be affected by the eigenvector centrality of different types of agents. \qed
	\begin{figure}
		\centering
		\includegraphics[width=.7\linewidth]{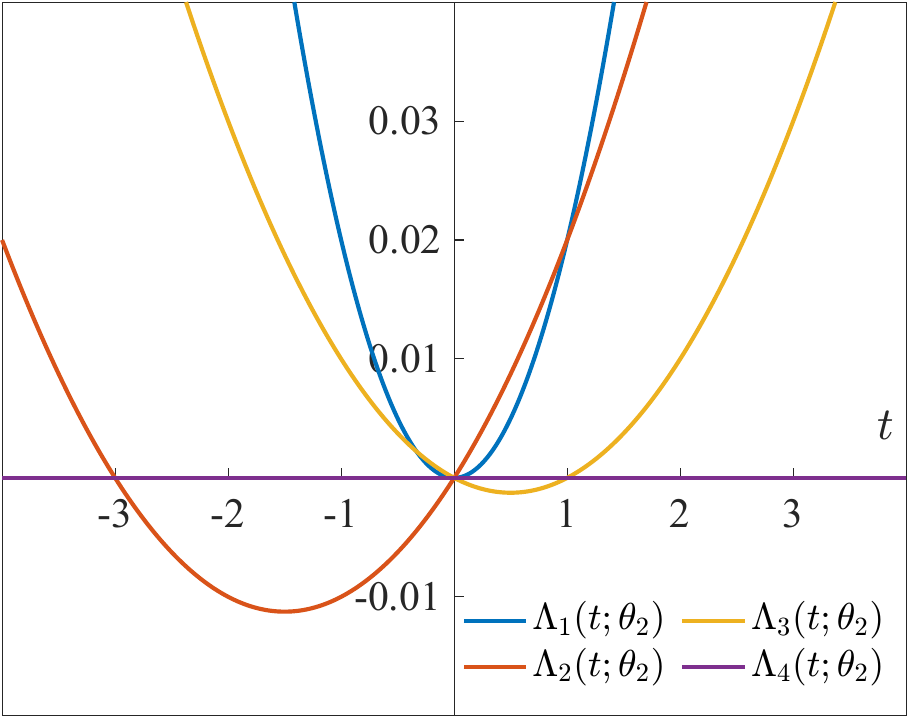}
		\caption{LMGFs $\Lambda_k(t;\theta_2)$ for $k\in\mathcal{N}$ in the example.}
		\label{fig: LMGF}
	\end{figure}
	
	Next, we introduce two important quantities related to the learning performance in the \emph{non-cooperative} scenario, where the ASL algorithm update \eqref{eq: ASL} at agent $k$ becomes
	\begin{equation}\label{eq: ASL in non-cooperative scenario}
		\bm{\lambda}^{\sf nc}_{k,i}(\theta)=\bm{\nu}^{\sf nc}_{k,i}(\theta)=(1-\delta)\bm{\lambda}^{\sf nc}_{k,i-1}(\theta)+\delta\bm{x}_{k,i}(\theta).
	\end{equation}Here and in the following, we use the superscript `$\sf nc$' for variables associated with the non-cooperative scenario. Similar to \eqref{eq: limiting random variable} in the social learning setting, we introduce the steady-state log-belief ratio for agent $k$:
	\begin{equation}
		\widetilde{\bm{\lambda}}_{k,\infty}^{\sf nc}(\theta)\triangleq\delta\sum_{m=0}^{\infty}(1-\delta)^m\bm{x}_{k,m+1}(\theta).
\end{equation} 
Accordingly, we define the $\theta$-related steady-state error probability as 
\begin{equation}
	p_k^{\sf nc}(\theta)\triangleq\mathbb{P}[\widetilde{\bm{\lambda}}_{k,\infty}^{\sf nc}(\theta)\leq 0].
\end{equation}Let $\Phi_k^{\sf nc}(\theta)$ denote the $\theta$-related error exponent for agent $k$ in this single-agent setup for each $\theta\neq\theta_1$:
	\begin{equation}\label{eq: error exponent for theta in non-coopeartive scenario}
		p_k^{\sf nc}(\theta)\doteq e^{-\Phi_{k}^{\sf nc}(\theta)/\delta}.
	\end{equation} 
 It is clear that $\Phi_{k}^{\sf nc}(\theta)\geq 0$ for all $k\in\mathcal{N}$. Similar to \eqref{eq: phi}, we define the function $\phi_k^{\sf nc}(t;\theta)$: 
	\begin{equation}\label{eq: rate function individually}		 
		\phi_k^{\sf nc}(t;\theta)\triangleq\int_0^t\frac{\Lambda_{k}(\tau;\theta)}{\tau}d\tau.
	\end{equation}
 	Let $t_k^{\sf nc}(\theta)\leq 0$ be the \emph{critical value} to attain $\Phi_k^{\sf nc}(\theta)$ for each agent $k\in\mathcal{N}$:
	\begin{equation}\label{eq: t_k for informative agent k}
		\Phi_k^{\sf nc}(\theta)=-\phi_k^{\sf nc}(t_k^{\sf nc}(\theta);\theta).
	\end{equation}
	We will show that $t_k^{\sf nc}(\theta)$ is essential to characterize the optimal solution of the error exponent maximization problem \eqref{eq: objective function}--\eqref{eq: constraint 2}. The value of $t_k^{\sf nc}(\theta)$ for agent $k$ in different groups $\mathcal{N}^I(\theta)$, $\mathcal{N}^U(\theta)$ and $\mathcal{N}^C(\theta)$ is derived as follows.
	
	For each $\theta$-informative agent $k\in\mathcal{N}^{{I}}(\theta)$,  $\Phi_k^{\sf nc}(\theta)$ can be expressed as
	\begin{equation}\label{eq: error exponent individually}
		\Phi_k^{\sf nc}(\theta)=-\inf_{t\in\mathbb{R}} \phi_k^{\sf nc}(t;\theta)
	\end{equation}
	by using the results from Theorem \ref{theorem: learning performance for the general model}. With the properties of $\Lambda_k(t;\theta)$ and $\phi_k^{\sf nc}(t;\theta)$ established in Lemma \ref{lemma:properties} (Appendix \ref{appendix: lemma}), we have $\Phi_k^{\sf nc}(\theta)>0$. 
	The existence and uniqueness of $t_k^{\sf nc}(\theta)$ for each $\theta$-informative agent  is proved in Appendix \ref{appendix: theorem 3}. If agent $k$ is $\theta$-uninformative, we have $\Lambda_{k}(t;\theta)\equiv 0$ from \eqref{eq: uninformative agent}. Hence, we obtain $\Phi_k^{\sf nc}(\theta)=0$ and $t_k^{\sf nc}(\theta)$ can be any non-positive value. We assume $t_k^{\sf nc}(\theta)=-C$ for simplicity, where $C$ is an arbitrary positive constant that can be dependent on $k$ and $\theta$. In addition, if agent $k$ is $\theta$-conflicting, then as $\delta$ approaches 0, it rejects hypothesis $\theta_1$ (if $d_k(\theta)<0$) or cannot distinguish between hypotheses $\theta_1$ and $\theta$ (if $d_k(\theta)=0$) with probability $1$ in steady state. From \eqref{eq: error exponent for theta in non-coopeartive scenario}, $\Phi_k^{\sf nc}(\theta)=0$ is obtained in this case. Moreover, according to Lemma \ref{lemma:properties} in Appendix \ref{appendix: lemma}, the following condition holds for any $\theta$-conflicting agent $k$:
	\begin{equation}\label{eq: phi_k conflicting agent}
		\phi_k^{\sf nc}(t;\theta)>0,\quad \forall t< 0,
	\end{equation}
	which yields $t_k^{\sf nc}(\theta)=0$ in \eqref{eq: t_k for informative agent k}. Based on the above analysis, we obtain
	\begin{equation}\label{eq: t_k in non-cooperative scenario}
		t_k^{\sf nc}(\theta):\begin{cases}
			<0, &\text{if } k\in\mathcal{N}^{I}(\theta),\\
			=-C,&\text{if } k\in\mathcal{N}^{U}(\theta),\\
			=0,&\text{if } k\in\mathcal{N}^{C}(\theta).
		\end{cases}
	\end{equation}
	Consider the previous example with LMGFs shown in Fig. \ref{fig: LMGF}. We have $t_2^{\sf nc}(\theta_2)=-3$ and $t_1^{\sf nc}(\theta_2)=t_3^{\sf nc}(\theta_2)=0$. Since agent $4$ is $\theta_2$-uninformative, $t_4^{\sf nc}(\theta_2)$ can take arbitrary value $-C$.
	
	\subsection{General results}
	\label{subsec: general results}
	Let $\Phi_{{\sum}}^{\sf nc}(\theta)$ denote the sum of individual $\Phi_k^{\sf nc}(\theta)$ in the non-cooperative scenario:
	\begin{align}
		\label{eq: sum of individual error exponent for theta}
		\Phi_{ \sum}^{\sf nc}(\theta)\triangleq\sum_{k=1}^N\Phi_k^{\sf nc}(\theta).
	\end{align}
	Then, we can derive the following bound for the error exponent $\Phi(\pi)$ for any feasible Perron eigenvector $\pi\in\Pi$.
	\begin{theorem}[\textbf{Benefit of cooperation}]\label{theorem: bound}
		For any Perron eigenvector $\pi\in\Pi$, the $\theta$-related error exponent $\Phi(\pi,\theta)$ defined in \eqref{eq: rate function} is bounded by
		\begin{equation}\label{eq: bound of each error exponent}
			\min_{k\in\mathcal{N}}\Phi^{\sf nc}_k(\theta)\leq\Phi(\pi,\theta)\leq\Phi_{\sum}^{\sf nc}(\theta),\quad\forall\theta\neq\theta_1.
		\end{equation}
		Correspondingly, the error exponent $\Phi(\pi)$ of the steady-state error probability is bounded by 
		\begin{equation}\label{eq: bound of error exponent}
			\min_{\theta\neq\theta_1}\min_{k\in\mathcal{N}}\Phi^{\sf nc}_k(\theta)\leq\Phi(\pi)\leq\min_{\theta\neq\theta_1}\Phi_{\sum}^{\sf nc}(\theta).
		\end{equation}
	\end{theorem}
	\begin{proof}
		See Appendix \ref{appendix: theorem 3}.
	\end{proof}
	\noindent Theorem \ref{theorem: bound} shows that the error exponent under adaptive social learning is no less than the worst error exponent in the single-agent setup. Therefore, the cooperation among agents is always beneficial for the agent that has the worst learning performance in the non-cooperative scenario. In addition, the best error exponent that can be achieved by the ASL strategy is given by the minimum aggregated quantity $\Phi_{\sum}^{\sf nc}(\theta)$ among all wrong hypotheses $\theta\neq\theta_1$. Since the centralized solution of the adaptive social learning problem is equivalent to a fully connected network, \eqref{eq: bound of error exponent} applies to the centralized case as well.  Furthermore, the upper bound in \eqref{eq: bound of error exponent} satisfies
	\begin{equation}\label{eq: greater than the sum of non-cooperative error exponent}
		\min_{\theta\neq\theta_1}{\Phi}^{\sf nc}_{\sum}(\theta)\triangleq\min_{\theta\neq\theta_1}\sum_{k=1}^N{\Phi}^{\sf nc}_k(\theta)\geq\sum_{k=1}^N\min_{\theta\neq\theta_1}{\Phi}^{\sf nc}_k(\theta),
	\end{equation}  
	which reveals that the cooperation among agents enables each agent in the network to obtain an error exponent that could be even larger than the sum of individual error exponents in the non-cooperative scenario. However, the achievability of the upper bound in \eqref{eq: bound of error exponent} is related to the specific setting of learning tasks. For instance, the  existence of some conflicting agents may lead to a smaller error exponent. We will describe next when this upper bound is achievable and how to reach this upper bound with proper combination policies.
	
	Let $\theta^\dagger$ be the wrong hypothesis\footnote{Here, we have assumed that the set $\Theta^{\sf nc}_{\min}\triangleq\{\theta:\argmin_{\theta\neq\theta_1}\Phi_{\sum}^{\sf nc}(\theta)\}$ is a singleton. If $\Theta^{\sf nc}_{\min}$ contains more than one element, we can repeat our analysis for each element in $\Theta^{\sf nc}_{\min}$.} corresponding to the upper bound in \eqref{eq: bound of error exponent}:
	\begin{equation}\label{eq: theta_dagger}
		\theta^\dagger\triangleq\argmin_{\theta\neq\theta_1}\Phi_{\sum}^{\sf nc}(\theta).
	\end{equation}
	From Theorem \ref{theorem: bound}, we have $\Phi(\pi)\leq \Phi_{\sum}^{\sf nc}(\theta^\dagger)$ for all Perron eigenvectors $\pi\in\Pi$. Therefore, any Perron eigenvector $\pi\in\Pi$ that gives $\Phi(\pi)= \Phi_{\sum}^{\sf nc}(\theta^\dagger)$ must be an optimal solution to the optimization problem in \eqref{eq: objective function}--\eqref{eq: constraint 2}. In view of constraint \eqref{eq: constraint 2}, ${\sf m_{ave}}(\pi,\theta^\dagger)>0$ for any $\pi\in\Pi$. Then, we have 
	\begin{equation}
		\Phi(\pi,\theta^\dagger)\overset{\eqref{eq: rate function}}{=}-\inf_{t\in\mathbb{R}}\phi(t;\pi,\theta^\dagger)=-\phi(t^\ast_\pi;\pi,\theta^\dagger)
	\end{equation}
	for some $t^\ast_\pi<0$ due to properties v) and vi) in Lemma 1. By definitions \eqref{eq: LMGF_ave}, \eqref{eq: phi}, and \eqref{eq: rate function individually}, the following inequality holds:
	\begin{align}
		\nonumber	\Phi(\pi,\theta^\dagger)
		&=-\sum_{k\in\mathcal{N}^C(\theta^\dagger)}{\phi_k^{\sf nc}(\pi_k t^\ast_\pi;\theta)}-\sum_{k\notin\mathcal{N}^C(\theta^\dagger)} \phi_k^{\sf nc}(\pi_k t^\ast_\pi;\theta)\\
		\label{eq: Phi for a feasible pi}
		\overset{\eqref{eq: phi_k conflicting agent}}&{\leq} -\sum_{k\notin\mathcal{N}^C(\theta^\dagger)} \phi_k^{\sf nc}(\pi_k t^\ast_\pi;\theta)\overset{\eqref{eq: error exponent individually}}{\leq} \Phi_{\sum}^{\sf nc}(\theta^\dagger).
	\end{align}	
	Due to the constraint $\pi_k>0,\forall k\in\mathcal{N}$ in \eqref{eq: constraint 1},  the first inequality in \eqref{eq: Phi for a feasible pi} becomes equality if and only if $\mathcal{N}^{C}(\theta^\dagger)=\emptyset$.
	Hence, the upper bound in \eqref{eq: bound of error exponent} cannot be attained if there are some $\theta^\dagger$-conflicting agents in the network. Eqs. \eqref{eq: phi_k conflicting agent} and \eqref{eq: Phi for a feasible pi} illustrate our definition \eqref{eq: conflicting agent} of the $\theta$-conflicting agents. In the following, we discuss the design of optimal Perron eigenvectors for the cases  $\mathcal{N}^C(\theta^\dagger)=\emptyset$ and $\mathcal{N}^C(\theta^\dagger)\neq\emptyset$, respectively.
	\subsubsection{Case 1: $\mathcal{N}^C(\theta^\dagger)=\emptyset$}In this case, the quantity $t_k^{\sf nc}(\theta^\dagger)$ defined in \eqref{eq: t_k in non-cooperative scenario} is negative for all $k\in\mathcal{N}$. With this property, we construct the following candidate Perron eigenvector $\pi^\dagger$:
	\begin{equation}\label{eq: candidate Perron eigenvector-case 1}
		\pi_k^\dagger=\frac{t_k^{\sf nc}(\theta^\dagger)}{\sum_{\ell=1}^N t_\ell^{\sf nc}(\theta^\dagger)},\quad \forall k\in\mathcal{N}.
	\end{equation}
	It is easy to see that $\mathbbm{1}^\top\pi^\dagger=1$ and $\pi^\dagger_k> 0, \forall k\in\mathcal{N}$. Under Perron eigenvector $\pi^\dagger$, the $\theta^\dagger$-related error exponent $\Phi(\pi^\dagger,\theta^\dagger)$  equals to the upper bound in \eqref{eq: bound of error exponent}:
	\begin{align}
		\nonumber
		\Phi_{\sum}^{\sf nc}(\theta^\dagger)\overset{\eqref{eq: bound of each error exponent}}{\geq}\Phi(\pi^\dagger,\theta^\dagger)\overset{\eqref{eq: rate function}}&{=}-\inf_{t\in\mathbb{R}}\phi(t;\pi^\dagger,\theta^\dagger)\\
		\nonumber
		&\geq -\phi\left({\sum\limits_{k=1}^Nt_k^{\sf nc}({\theta^\dagger})};\pi^\dagger,\theta^\dagger\right)\\
		\nonumber
		\overset{\text{(a)}}&{=}-\sum_{k=1}^N \phi_k^{\sf nc}({t_k^{\sf nc}(\theta^\dagger)};\theta^\dagger)\\
		\label{eq: pi_dagger for the upper bound-case 1}
		\overset{\eqref{eq: t_k for informative agent k}}&{=}\sum_{k=1}^N \Phi_k^{\sf nc}(\theta^\dagger)\triangleq\Phi_{\sum}^{\sf nc}(\theta^\dagger)
	\end{align}
	where in (a) we used the definitions given in \eqref{eq: LMGF_ave}, \eqref{eq: phi}, \eqref{eq: rate function individually}, and \eqref{eq: candidate Perron eigenvector-case 1}. From Theorem \ref{theorem: learning performance for the general model}, the error exponent $\Phi(\pi^\dagger)$ under Perron eigenvector $\pi^\dagger$ is determined by the minimum $\theta$-related error exponent $\Phi(\pi^\dagger,\theta)$. If $\pi^\dagger$ is feasible for the truth learning (i.e., $\pi^\dagger\in\Pi$) and satisfies $\Phi(\pi^\dagger,\theta^\dagger)\leq\Phi(\pi^\dagger,\theta)$, $\forall \theta\neq\theta_1$, then we have
	\begin{equation}
		\Phi(\pi^\dagger)\triangleq\min_{\theta\neq\theta_1}\Phi(\pi^\dagger,\theta)=\Phi(\pi^\dagger,\theta^\dagger)\overset{\eqref{eq: pi_dagger for the upper bound-case 1}}{=}\Phi_{\sum}^{\sf nc}(\theta^\dagger),
	\end{equation}
	which proves that $\pi^\dagger$ is an optimal solution. Let $\Pi_1$ denote the following set of Perron eigenvectors:
	\begin{equation}\label{eq: Pi_1-necessary conditions}
		\Pi_1=\Big\{ \pi\in\Pi:\Phi(\pi,\theta^\dagger)\leq\Phi(\pi,\theta),\forall \theta\neq\theta_1 \Big\}.
	\end{equation}
	In the following theorem, we formally establish the achievability of the upper bound in \eqref{eq: bound of error exponent} for the error exponent maximization problem in \eqref{eq: objective function}--\eqref{eq: constraint 2}, and characterize the optimal Perron eigenvectors corresponding to this upper bound.	 
	\theoremgroup
	\begin{theoremgrp}[\textbf{Optimal Perron eigenvector}]\label{theorem: optimal Perron eigenvector}
		Consider $\mathcal{N}^C(\theta^\dagger)=\emptyset$. Let $\Phi^\star$ be the maximum error exponent in the optimization problem \eqref{eq: objective function}--\eqref{eq: constraint 2} and $\Pi^\star$ be the set of optimal Perron eigenvectors. Define
		\begin{align}
			\label{eq: Pi_2-case 1}
			\Pi_2=&\Big\{ \pi:\pi_k=\alpha \pi_k^\dagger,\forall k\notin\mathcal{N}^{ U}(\theta^\dagger),\forall \alpha >0\Big\},\\
			\label{eq: optimal set-case 1}
			\Pi^\dagger=&\Pi_1\cap\Pi_2.
		\end{align}
		If $\Pi^\dagger\neq\emptyset$, then the upper bound of the error exponent can be achieved for the given learning task, i.e.,  $\Phi^\star=\Phi_{\sum}^{\sf nc}(\theta^\dagger)$, and the corresponding optimal set is given by $\Pi^\star=\Pi^\dagger$.
		Otherwise, we have $\Phi^\star<\Phi_{\sum}^{\sf nc}(\theta^\dagger)$.
	\end{theoremgrp}
	\begin{proof}
		see Appendix \ref{appendix: theorem 3a}.
	\end{proof}
	\noindent Theorem \ref{theorem: optimal Perron eigenvector} asserts that if the upper bound of the error exponent is achievable for the given learning task, the optimal Perron eigenvectors can be derived with $t_k^{\sf nc}(\theta^\dagger)$ following \eqref{eq: candidate Perron eigenvector-case 1}. Since $t_k^{\sf nc}(\theta^\dagger)$ is unique for all $k\in\mathcal{N}^I(\theta^\dagger)$, we have $\Pi_2=\{\pi^\dagger\}$ if $\mathcal{N}^{{U}}(\theta^\dagger)=\emptyset$. In this circumstance, $\Pi^\dagger$ is either a singleton, i.e., $\Pi^\dagger=\{\pi^\dagger\}$, or an empty set. The definition of $\Pi_2$ in \eqref{eq: Pi_2-case 1} reveals a basic feature of the optimal Perron eigenvectors: the centralities of $\theta^\dagger$-informative agents should be distributed in a proportional manner that depends on the values of $t_k^{\sf nc}(\theta^\dagger)$. For example, let us consider a learning task with $N=3$, $\mathcal{N}^{ I}(\theta^\dagger)=\{1,2\}$, $\mathcal{N}^{U}(\theta^\dagger)=\{3\}$ and $t_1^{\sf nc}(\theta^\dagger)=2t_2^{\sf nc}(\theta^\dagger)$. Assume $\Pi^\dagger\neq\emptyset$ for the given learning task, then the conditions $\pi^\star_1=2\pi^\star_2$ and $\pi^\star_3=1-3\pi_2^\star$ must be satisfied for any optimal solution $\pi^\star\in\Pi^\star$. This implies that the optimality of a Perron eigenvector requires keeping a balance among the information from the $\theta^\dagger$-informative agents.
	\subsubsection{Case 2: $\mathcal{N}^C(\theta^\dagger)\neq\emptyset$} From \eqref{eq: Phi for a feasible pi}, we already know that the upper bound of the error exponent cannot be attained in this case. Since $\phi(t;\pi,\theta)$ defined in \eqref{eq: phi} is a weighted quantity of the individual $\phi_{k}^{\sf nc}(t;\theta)$, in principle, we would set $\pi_k=0$ for all $k\in\mathcal{N}^C(\theta^\dagger)$ to improve the learning accuracy. However, a zero centrality of an agent means that the information from this agent cannot spread over the network, which violates the assumption of the strongly connected communication network (Assumption \ref{asmp: assumption 1}).\footnote{Consider the communication network after removing all $\theta^\dagger$-conflicting agents. If it is still strongly connected, then Theorems \ref{theorem: learning performance for the general model}--\ref{theorem: optimal Perron eigenvector} can be applied to this smaller network. Since $\Phi_k^{\sf nc}(\theta^\dagger)=0,\forall k\in\mathcal{N}^C(\theta^\dagger)$, the upper bound of the error exponent in this smaller network is still $\Phi_{ \sum}^{\sf nc}(\theta^\dagger)$.} Therefore, only combination policies that deliver an error exponent close to the upper bound can be pursued. For a small $\epsilon>0$, we say that a Perron eigenvector $\pi$ is $\epsilon$-optimal if the difference between the corresponding error exponent $\Phi(\pi)$  and the upper bound in \eqref{eq: bound of error exponent} is not larger than $\epsilon$:
	\begin{equation}\label{eq: 1-epsilon optimal}
		\Phi(\pi)\geq \Phi_{\sum}^{\sf nc}(\theta^\dagger)-\epsilon.
	\end{equation}  
	Next, we proceed to derive the $\epsilon$-optimal Perron eigenvectors. Let $\abs{\mathcal{N}^C(\theta^\dagger)}$ be the number of $\theta^\dagger$-conflicting agents, then for any given $\epsilon>0$, we define $t_\epsilon$ as 
	\begin{equation}\label{eq: t_epsilon}
		t_\epsilon=\inf\Bigg\{t<0:\phi_k^{\sf nc}(t;\theta^\dagger)\leq \frac{\epsilon}{\abs{\mathcal{N}^C(\theta^\dagger)}},\forall k\in\mathcal{N}^C(\theta^\dagger)\Bigg\}.
	\end{equation}		
	This yields
	\begin{equation}\label{eq: sum of error exponents related to conflicting agents}
		0<\sum_{k\in\mathcal{N}^C(\theta^\dagger)}\phi_k^{\sf nc}(t_\epsilon;\theta^\dagger)\leq\epsilon.
	\end{equation} 
	In view of \eqref{eq: phi_k conflicting agent}, $\phi_k^{\sf nc}(t;\theta^\dagger)>0,\forall t<0$ for all $\theta^\dagger$-conflicting agents. By definition \eqref{eq: rate function individually}, $t_\epsilon$ converges to 0 as $\epsilon$ approaches 0. Furthermore, we define 
	\begin{equation}\label{eq: new t_k}
		\widetilde{t}_k^{\sf nc}(\epsilon)=\begin{cases}
			t_k^{\sf nc}(\theta^\dagger), &\text{if }k\notin\mathcal{N}^{C}(\theta^\dagger),\\
			t_\epsilon, & \text{if }k\in\mathcal{N}^{C}(\theta^\dagger),
		\end{cases}
	\end{equation}
	then $\widetilde{t}_k^{\sf nc}(\epsilon)<0$ for all $k\in\mathcal{N}$. Similar to \eqref{eq: candidate Perron eigenvector-case 1}, we construct the following candidate Perron eigenvector $\pi^\dagger(\epsilon)$ using $\widetilde{t}_k^{\sf nc}(\epsilon)$:
	\begin{equation}\label{eq: candidate Perron eigenvector}
		\pi_k^\dagger(\epsilon)=\frac{\widetilde{t}_k^{\sf nc}(\epsilon)}{\sum_{\ell=1}^N \widetilde{t}_\ell^{\sf nc}(\epsilon)},\quad \forall k\in\mathcal{N}.
	\end{equation}
	Following the same analytical steps employed in the case $\mathcal{N}^C({\theta^\dagger})=\emptyset$, we can prove that the $\theta^\dagger$-related error exponent under Perron eigenvector $\pi^{\dagger}({\epsilon})$ satisfies
	\begin{equation}
		\Phi(\pi^\dagger(\epsilon),\theta^\dagger)\geq \Phi_{\sum}^{\sf nc}(\theta^\dagger)-\epsilon.
	\end{equation}
	Hence, $\pi^\dagger(\epsilon)$ is an $\epsilon$-optimal Perron eigenvector if it satisfies the conditions given by set $\Pi_1$ in \eqref{eq: Pi_1-necessary conditions}. Likewise, we can establish the following theorem for the $\epsilon$-optimal Perron eigenvectors.
	\begin{theoremgrp}[\textbf{$\epsilon$-optimal Perron eigenvector}]\label{theorem: epsilon-optimal Perron eigenvector}
		Consider $\mathcal{N}^C(\theta^\dagger)\neq\emptyset$. For a given small $\epsilon>0$, we define the sets:
		\begin{align}
			\label{eq: Pi_2-case 2}
			\Pi_2(\epsilon)&=\Big\{\pi: \pi_k=\alpha \pi_k^\dagger(\epsilon),\forall k\notin\mathcal{N}^{ U}(\theta^\dagger),\forall \alpha >0\Big\},\\
			\label{eq: optimal set-case 2}
			\Pi^\dagger_\epsilon&=\Pi_1\cap\Pi_2(\epsilon).
		\end{align}
		If $\Pi^\dagger_\epsilon \neq\emptyset$, then any Perron eigenvector $\pi\in\Pi_\epsilon^\dagger$ is $\epsilon$-optimal.
	\end{theoremgrp}
	\begin{proof}
		See Appendix \ref{appendix: theorem 3b}.
	\end{proof}
	\noindent Theorems \ref{theorem: optimal Perron eigenvector} and \ref{theorem: epsilon-optimal Perron eigenvector} describe the optimal Perron eigenvectors of combination policies that deliver error exponents which are, respectively, equal to the upper bound or close enough to it.  To illustrate our conclusions in this part, we will discuss next some interesting cases within the framework of adaptive social learning. 
	
	\subsection{Interesting cases}\label{subsec: interesting cases}
	In the following, we consider three learning cases: 
	\begin{enumerate}
		\item[1)] \emph{Distributed detection}: The local likelihood models are identical for all agents and the  local observations ${\bm{\xi}}_{k,i}$ are i.i.d., which is often assumed in the distributed detection problem \cite{Bajovic2011:Distributed,Bajovic2012:Large,Matta:2016kh,Matta:2016cb}. 
		\item[2)]\emph{Social learning with accurate signal model}: The local likelihood models are accurate, i.e., $f_k=L_k(\cdot|\theta_1)$ for all agents, which is often considered in the social learning literature (e.g.,  \cite{Jadbabaie:2012ii,Zhao:2013wa,Salami:2017jf,Lalitha:2018ej,Mitra:2020fi,Bordignon:2020wx,Shahrampour:2016kk,jadbabaie2013information}). 
		\item[3)]\emph{Social learning under Gaussian noises}: The shift-in-mean Gaussian model \cite{Matta:2016kh} with zero-mean Gaussian noises is considered to test the impact of noises on the optimal Perron eigenvector.
	\end{enumerate}
	
	\begin{Corollary }[\textbf{Distributed detection}]\label{Corollary  1} 
		If $\{f_k,L_k(\cdot|\theta)\}$ are identical for all agents, then the uniform Perron eigenvector $\frac{1}{N}\mathbbm{1}$ is the unique optimal solution to the error exponent optimization problem. The corresponding optimal error exponent has an $N$-fold improvement in comparison to that in the non-cooperative scenario, i.e.,  ${\Phi}^\star=\Phi_{\sum}^{\sf nc}(\theta^\dagger)=N{\Phi}_k^{\sf nc}(\theta^\dagger)$. 
	\end{Corollary }
	\begin{proof}
		Since all agents have the same signal and likelihood models, ${\sf m_{ave}}(\pi,\theta)=d_k(\theta),\forall\theta\in\Theta$ holds for any Perron eigenvector $\pi$. In view of Assumption \ref{asmp: assumption 3}, we have $d_k(\theta)>0$ for all $k\in\mathcal{N}$ and thus $\mathcal{N}^{ I}(\theta)=\mathcal{N},\forall\theta\neq\theta_1$. Moreover, it is clear that $\Phi_k^{\sf nc}(\theta)$ and $t_k^{\sf nc}(\theta)$ will be identical for all agents. Then, the candidate Perron eigenvector $\pi^\dagger$ is given by $\pi^\dagger=\frac{1}{N}\mathbbm{1}$ according to \eqref{eq: candidate Perron eigenvector-case 1}. Since $\mathcal{N}^{U}(\theta^\dagger)=\emptyset$, $\Pi _2$ is a singleton with $\Pi_2=\{\frac{1}{N}\mathbbm{1}\}$. Moreover, it is easy to derive $\Phi(\frac{1}{N}\mathbbm{1},\theta)=N\Phi_k^{\sf nc}(\theta)$ for all $\theta\neq\theta_1$. By the definition of $\theta^\dagger$, we know $\Phi(\frac{1}{N}\mathbbm{1},\theta^\dagger)<\Phi(\frac{1}{N}\mathbbm{1},\theta),\forall\theta\neq\theta_1$ and thus $\frac{1}{N}\mathbbm{1}$ belongs to set $\Pi_1$. This guarantees that the set $\Pi^\dagger$ in \eqref{eq: optimal set-case 1} is not empty with $\Pi^\dagger=\{\frac{1}{N}\mathbbm{1}\}$. The claim follows thereby. 
	\end{proof}
	We note that Corollary \ref{Corollary  1} actually answers the question posed in \cite{Matta:2016cb}, regarding the optimal combination policy. In the simulation part of \cite{Matta:2016cb}, the authors had provided an important intuitive answer for this question with a ``first-order analysis''. They claimed that doubly-stochastic combination matrices may be preferred, while this statement remains to be theoretically proved. Our results now establish formally that their intuition about this non-trivial problem was correct. 
	
	\begin{Corollary }[\textbf{Social learning with accurate signal model}]\label{Corollary  2}
		Assume the accurate signal model \eqref{eq: likelihood of an observation}--\eqref{eq: accurate likelihood model}, then the upper bound of the error exponent can be achieved by the uniform Perron eigenvector. That is, $\Phi^\star=\Phi_{\sum}^{\sf nc}(\theta^\dagger)$ and $\frac{1}{N}\mathbbm{1}\in\Pi^\star$.
	\end{Corollary }
	\begin{proof}
		When $f_k=L_k(\cdot|\theta_1)$ for all agents $k\in\mathcal{N}$, the expectation $d_k(\theta)$ becomes the KL divergence between $L_k(\cdot|\theta_1)$ and $L_k(\cdot|\theta)$. For any wrong hypothesis $\theta\neq\theta_1$, due to the non-negativeness of KL divergences, either $d_k(\theta)>0$ or $\Lambda_{k}(t;\theta)\equiv 0$ for all $k\in\mathcal{N}$. This yields $\mathcal{N}^{{C}}(\theta)=\emptyset$. Under Assumption \ref{asmp: assumption 3}, we obtain
		\begin{equation}\label{eq: mave >0 in Corollary 2}
			{\sf m_{ave}}(\pi,\theta)>0,\quad \forall\theta\neq\theta_1,
		\end{equation}
		for any Perron eigenvector $\pi$. In addition, from property ii) in Lemma \ref{lemma:properties}, we have 
		\begin{equation}\label{eq: corollary 2 t_k}
			t_k^{\sf nc}(\theta)=-1,\quad \forall k\in\mathcal{N}^{{I}}(\theta),
		\end{equation}
		for all $\theta\neq\theta_1$. Although different $\theta$-informative agents may possess distinct likelihood models, which would endow those agents with different learning abilities when working in a non-cooperative way, we still obtain an identical $t_k^{\sf nc}(\theta)$  for all $\theta$-informative agents. Hence, it is clear that the uniform Perron eigenvector $\frac{1}{N}\mathbbm{1}$ belongs to $\Pi_2$.
		
		An important consequence of \eqref{eq: corollary 2 t_k} is that for any hypothesis $\theta\neq\theta_1$, the uniform Perron eigenvector $\frac{1}{N}\mathbbm{1}$ corresponds to the upper bound of the $\theta$-related error exponent $\Phi(\pi,\theta)$ given in Theorem \ref{theorem: bound}. That is,
		\begin{equation}\label{eq: achieves all upper bounds in Corollary 2}
			\Phi(\frac{1}{N}\mathbbm{1},\theta)=\Phi_{\sum}^{\sf nc}(\theta),\quad \forall \theta\neq\theta_1.
		\end{equation}
		Then by the definition of $\theta^\dagger$ in \eqref{eq: theta_dagger}, we obtain $\Phi(\frac{1}{N}\mathbbm{1},\theta^\dagger)<\Phi(\frac{1}{N}\mathbbm{1},\theta)$ for all $\theta\neq\theta_1$. Moreover, since ${\sf m_{ave}}(\pi,\theta)>0$ in \eqref{eq: mave >0 in Corollary 2}, we know that $\frac{1}{N}\mathbbm{1}\in\Pi_1$. Therefore, $\Pi^\dagger$ is not empty and $\frac{1}{N}\mathbbm{1}\in\Pi^\dagger$ is an optimal solution to the error exponent maximization problem. 
	\end{proof}
	Since the uniform Perron eigenvector corresponds to a doubly-stochastic combination policy, we can conclude from Corollary \ref{Corollary  2} that any doubly-stochastic combination policy will be optimal for the social learning tasks with accurate signal model \eqref{eq: likelihood of an observation}--\eqref{eq: accurate likelihood model}.  Importantly, this result is in contrast to the analogous results in the context of distributed optimization \cite{sayed2014adaptation,sayed2014adaptive}, when the agents have access to data of varying quality. This is because, in the adaptive social learning problem, the agents want to learn an optimal decision from the received data, rather than a specific parameter.  From \eqref{eq: error exponent expression}, the performance of learning is determined by the distribution of $\bm{x}_{{\sf ave},i}(\theta)$, which captures the information of local observations across the network. When the signal model is accurate, the log-likelihood ratio $\bm{x}_{k,i}(\theta)$ provides the full information of each piece of independent observation $\bm{\xi}_{k,i}$ for the decision-making task. Therefore, the optimality of a uniform Perron eigenvector is expected. The conclusion here has been demonstrated in the conference version \cite{Ping2022:Optimal} of this paper. 
	\begin{Corollary }[\textbf{Social learning under Gaussian noises}]\label{Corollary  3}
		Consider the canonical shift-in-mean Gaussian problem \cite{Matta:2016kh} where the local likelihood models of each agent $k$ are given by a family of Gaussian distributions with different means:
		\begin{equation}\label{eq: mean-shift Gaussian model}
			L_k(\cdot|\theta)= {\sf N}({\sf m}_k(\theta),\sigma_k^2),\quad\forall \theta\in\Theta.
		\end{equation}
		Assume that the measurements $\bm{\xi}_{k,i}$ are corrupted by some noise $\bm{n}_{k,i}$ following a zero-mean Gaussian model ${\sf N}(0,\upsilon^2_k)$, and define the ratio between $\upsilon_k^2$ and $\sigma_k^2$ as the noise level $\varepsilon_k$ at agent $k$:
		\begin{equation}
			\varepsilon_k\triangleq\frac{\upsilon_k^2}{\sigma_k^2}.
		\end{equation}
		An optimal Perron eigenvector $\pi^\star\in\Pi^\star$ for this learning task \eqref{eq: mean-shift Gaussian model} is given by
		\begin{equation}\label{eq: optimal pi for noisy Gaussian model}
			\pi^\star_k=\frac{(1+\varepsilon_k)^{-1}}{\sum_{\ell=1}^N(1+\varepsilon_\ell)^{-1}},\quad\forall k\in\mathcal{N}.
		\end{equation}
		If all agents are $\theta^\dagger$-informative, then $\pi^\star$ is unique, i.e., $\Pi^\star=\{\pi^\star\}$. The corresponding maximum error exponent $\Phi^\star$ is given by
		\begin{equation}\label{eq: optimal Phi for noisy Gaussian model}
			{\Phi}^\star=\Phi_{\sum}^{\sf nc}(\theta^\dagger)\triangleq\min_{\theta\neq\theta_1}\sum_{k=1}^N\frac{\left({\sf m}_k(\theta_1)-{\sf m}_k(\theta)\right)^2}{4\sigma_k^2(1+\varepsilon_k)}.
		\end{equation}
	\end{Corollary }
	\begin{proof}
		See Appendix \ref{appendix: corollary 3}. 
	\end{proof}
	In this noisy environment, the log-likelihood ratio $\bm{x}_{k,i}(\theta)$ is calculated based on a perturbed signal. From \eqref{eq: optimal pi for noisy Gaussian model}, the optimal centrality of agents is determined by the quality of their observations. To obtain better steady-state learning performance, an agent with a lower noise level should be placed in a more centralized position such that it receives more effective attentions from other agents. Moreover, $\pi^\star=\frac{1}{N}\mathbbm{1}$ is obtained in the noiseless environment, which is consistent with Corollary \ref{Corollary  2}. An adverse impact of the noisy observations on the steady-state learning performance is captured by \eqref{eq: optimal Phi for noisy Gaussian model}.
	
	\subsection{Practical aspects on the design of combination policies}
	In this section, we discuss some practical issues related to designing the combination policy in real-world systems. As demonstrated in Sections \ref{subsec: preliminary definitions} and \ref{subsec: general results}, developing an optimal combination policy relies on first finding an optimal Perron eigenvector $\pi^\star$ and then constructing the combination matrix with the given Perron eigenvector $\pi^\star$. Therefore, one critical aspect concerns learning the optimal Perron eigenvector $\pi^\star$. 
	
	In Section \ref{subsec: interesting cases}, we discussed some interesting cases where the explicit expression for $\pi^\star$ is obtained. Nevertheless, a closed form for $\pi^\star$ is not available in general cases. According to Theorem \ref{theorem: optimal Perron eigenvector}, $\pi^\star$ is defined by the critical value $t_k^{\sf nc}(\theta)$ associated with each agent $k$. From \eqref{eq: t_k for informative agent k}, $t_k^{\sf nc}(\theta)$ is determined by solving an equation involving the LMGF $\Lambda_k(t;\theta)$. By definition \eqref{eq: LMGF}, $\Lambda_k(t;\theta)$ is characterized by the signal model $f_k$, which is unknown to the agents in practice. This raises the question of how to estimate $t_k^{\sf nc}(\theta)$ in social learning.
	\subsubsection{Estimation of $t_k^{\sf nc}(\theta)$}
	\label{sec: estimation of t_k}
			According to \eqref{eq: t_k in non-cooperative scenario}, the critical value $t_k^{\sf nc}(\theta)$ for $k\notin\mathcal{N}^I(\theta)$ can be obtained once we know the type (i.e., uninformative, informative, or conflicting) of agent $k$ for hypothesis $\theta$. It is more demanding to derive the critical value for a $\theta$-informative agent. According to definition \eqref{eq: definition of t_k for informative agents} given in Appendix \ref{appendix: theorem 3}, $t_k^{\sf nc}(\theta)$ is found by solving the following equation for each $\theta$-informative agent $k$:
		\begin{equation}\label{eq: direct estination: finding t_k}
			\Lambda_{k}(t;\theta)=0,\quad t< 0.
		\end{equation}
		Therefore, the evaluation of $t_k^{\sf nc}(\theta)$ depends on our approximation of the LMGF. Since $\Lambda_{k}(t;\theta)$ is the logarithm of the MGF, it is more convinient to estimate the MGF directly. Let $M_k(t;\theta)$ denote the MGF of variable $\bm{x}_{k,i}(\theta)$:
		\begin{equation}
			M_k(t;\theta)\triangleq\mathbb{E}[e^{t\bm{x}_{k,i}(\theta)}].
		\end{equation} 
		Eq. \eqref{eq: direct estination: finding t_k} is equivalent to solving (assume that the signal space $\mathcal{X}_k$ is continuous):
		\begin{equation}\label{eq: 2}
			M_k(t;\theta)=\int_{\xi\in\mathcal{X}_k} f_k(\xi)\left[\frac{L_k(\xi|\theta_1)}{L_k(\xi|\theta)}\right] ^td\xi=1.
		\end{equation}
		In the following,  we propose two methods for MGF approximation within different settings.  
		
		{\em i) Data-based MGF approximation:} Since the signal model $f_k$ is unknown, a direct approach for MGF approximation is to estimate $M_{k}(t;\theta)$ from empirical data. Consider a finite set of $S$ realizations $\{{x}_{k,m} \}_{m=1}^S $, the estimator of $M_{k}(t;\theta)$ is constructed as
		\begin{equation}\label{eq: direct estimator}
			\widehat{M}_{k}^{(S)}(t;\theta)\triangleq\frac{1}{S}\sum_{m=1}^{S}e^{t{x}_{k,m}(\theta)}.
		\end{equation}
		Due to the assumption of $\Lambda_{k}(t;\theta)<\infty$ in Theorem \ref{theorem: learning performance for the general model}, according to the law of large numbers, we know that $\widehat{M}_{k}^{(S)}(t;\theta)$ converges pointwise to $M_{k}(t;\theta)$ as the sample size $S$ increases. The convergence rate of this estimate is discussed in \cite{rohwer2015convergence}. In the context of social learning where the observations arrive in a streaming manner, the estimator in \eqref{eq: direct estimator} can be updated by the following recursion:
		\begin{equation}\label{eq: recurive estimation of MGF}
			\widehat{M}_{k}^{(i)}(t;\theta)=\frac{1}{i}e^{t\bm{x}_{k,i}(\theta)}+\frac{i-1}{i}\widehat{M}_{k}^{(i-1)}(t;\theta)
		\end{equation}
		where $\widehat{M}_{k}^{(i)}(t;\theta)$ denotes the agent $k$'s estimation of $M_{k}(t;\theta)$ at time instant $i$, i.e., after collecting $i$ observations. Therefore, in addition to performing the ASL protocol \eqref{eq: ASL}, the agents also run an MGF estimation \eqref{eq: recurive estimation of MGF} at each time instant. One issue in the MGF approximation is that $M_{k}(t;\theta)$ is a function of the continuous variable $t$. Consequently, we need to discretize the estimated quantity over variable $t$, which introduces an unavoidable discretization error. Since the ultimate goal is to obtain a good estimate for $t_k^{\sf nc}(\theta)$, the properties of $t_k^{\sf nc}(\theta)$ can be helpful for the discretization design. From Lemma \ref{lemma:properties} in Appendix \ref{appendix: lemma}, we know that $t_k^{\sf nc}(\theta)$ lies in the region where $M_{k}(t;\theta)$ is decreasing. This suggests that agent $k$ can choose a finer discretization of $t$ around the region where its estimator $\widehat{M}_{k}^{(i)}(t;\theta)$ decreases with $t$ and takes values close to 1. After enough observations, it can focus on this critical region. When the random variable $\bm{x}_{k,i}(\theta)$ is bounded, the estimator in \eqref{eq: recurive estimation of MGF} converges quickly for any $t$ \cite{rohwer2015convergence}. Even if $\bm{x}_{k,i}(\theta)$ is unbounded, this estimator can be useful for finding $t_k^{\sf nc}(\theta)$ as only a small region of $t$ needs to be considered.  For ease of reference, the approach being discussed where agents first evaluate their MGF using estimator \eqref{eq: recurive estimation of MGF} and then approximate $t_k^{\sf nc}(\theta)$ as the solution of \eqref{eq: 2} involving the estimated MGF, will be referred to as the \emph{direct estimation method}.
		
		{\em ii) Model-based MGF approximation:} If the statistical models (i.e., the signal model $f_k$ and likelihood models $\{L_k(\cdot|\theta)\}$) belong to the same exponential family, then $t_k^{\sf nc}(\theta)$ can be approximated by estimating the \emph{natural parameter} of the signal model $f_k$. Similar ideas were employed in \cite{nielsen2013information} for approximating numerically the critical value (i.e., Chernoff point) that determines the Chernoff information associated with the Bayesian decision rule in binary hypothesis testing problems. An exponential family $\mathcal{F}_E$ with natural parameter $\beta$ is a set of distributions of the form \cite{sayed2022inference}:
		\begin{equation}
			p_\beta(\xi) = \exp(\beta^\top T(\xi) - F(\beta) + K(\xi))
		\end{equation}
		for $\beta$ belonging to the natural parameter space 
		\begin{equation}
			\mathcal{B}=\left\{\beta\in\mathbb{R}^D\mid\int p_\beta(\xi)d\xi=1\right\}
		\end{equation}
		where $D$ is the dimension of $\mathcal{B}$, namely, the order of the family $\mathcal{F}_E$. The term $T(\xi)$ is a sufficient statistic, and the map $K(\xi):\mathcal{X}_k\mapsto \mathbb{R}$ is an auxiliary function. Function $F(\beta)$ characterizes the family and is known as a partition function or the log-normalizer in the literature. With $\int p_\beta(\xi)d\xi=1$, it follows that
		\begin{equation}\label{eq: F(theta)}
			F(\beta)=\log \int \exp(\beta^\top T(\xi) + K(\xi)) d\xi.
		\end{equation} 
		Suppose that $f_k$ and $\{L_k(\cdot|\theta)\}$ belong to the same exponential family with $f_k = p_{\beta^\ast}$ and
		\begin{equation}\label{eq: statistical models in exponential family}
			L_k(\cdot|\theta_h)=p_{\beta_h},\quad \forall h=1,2,\dots,H.
		\end{equation}
		Then, \eqref{eq: 2} can be rewritten as 
		\begin{align}
			\nonumber
			&\quad\;M_k(t;\theta)=\int_{\xi\in\mathcal{X}_k} p_{\beta^\ast}(\xi)\left[\frac{p_{\beta_1}(\xi)}{p_\beta(\xi)}\right] ^td\xi \\\nonumber
			& = \int_{\xi\in\mathcal{X}_k}  \exp({\beta^\ast}^\top T(\xi) - F(\beta^\ast) + K(\xi))\\\nonumber
			&\quad\times\Big[\exp((\beta_1-\beta)^\top T(\xi)- F(\beta_1) + F(\beta) )\Big] ^td\xi \\\nonumber
			& = \exp(-F(\beta^\ast)-tF(\beta_1)+tF(\beta))\\\nonumber
			&\quad\times\int_{\xi\in\mathcal{X}_k} \exp((\beta^\ast+ t(\beta_1-\beta))^\top T(\xi) + K(\xi)) d\xi \\\nonumber
			& \overset{\text{(a)}}{=} \exp(-F(\beta^\ast)-tF(\beta_1)+tF(\beta)+F(\beta^\ast + t(\beta_1-\beta))) \\ 
			& = 1
		\end{align}
		where in (a) we used the definition of $F(\beta)$ in \eqref{eq: F(theta)}. Therefore, $t_k^{\sf nc}(\theta)$ can be found by solving the following equation:
		\begin{equation}\label{eq: critical function for exponential family}
			F(\beta^\ast + t(\beta_1-\beta)) - F(\beta^\ast) - t(F(\beta_1) - F(\beta))=0.
		\end{equation}
		Since the likelihood models are available to the agents, $\beta$ and $\beta_1$ are known parameters in \eqref{eq: critical function for exponential family}. If we can estimate the natural parameter $\beta^\ast$ corresponding to the signal model $f_k$ with $\widehat{\beta}^\ast$, then $t_k^{\sf nc}(\theta)$ can be approximated by solving
		\begin{equation}\label{eq: approximated critical function for exponential family}
			F(\widehat{\beta}^\ast + t(\beta_1-\beta)) - F(\widehat{\beta}^\ast) - t(F(\beta_1) - F(\beta))=0.
		\end{equation}
		This method will be referred to as the \emph{indirect estimation method}, since we approximate $t_k^{\sf nc}(\theta)$ by an implicit function instead of evaluating the possible value pointwise. Compared with the direct estimation method, the indirect approach only requires to estimate the natural parameter of the signal model, which is more straightforward in practice. Next, we provide the example of Gaussian distributions to explain the quantities in \eqref{eq: approximated critical function for exponential family}. 
		
		\emph{Gaussian Example:} Consider the Gaussian distributions ${\sf N}({\sf m},\sigma^2)$ with the general formula:
		\begin{equation}\label{eq: statistical models with Gaussian models}
			f(\xi,{\sf m},\sigma^2) = \frac{1}{\sqrt{2\pi \sigma^2}}\exp\left(-\frac{(\xi-{\sf m})^2}{2\sigma^2}\right).
		\end{equation}
		It is easy to verify that Gaussian distributions belong to the exponential family with 
		\begin{equation}\label{eq: Gaussian exponential family}
			T(\xi)=\begin{bmatrix}
				\xi\\ \xi^2
			\end{bmatrix},
			\quad \beta =\begin{bmatrix}
				\beta(1)\\\beta(2)
			\end{bmatrix} =\begin{bmatrix}
				\frac{{\sf m}}{\sigma^2}\\ -\frac{1}{2\sigma^2}
			\end{bmatrix},
			\quad K(\xi)=0,
		\end{equation}
		and
		\begin{align}
			\nonumber
			F(\beta)&=\frac{{\sf m}^2}{2\sigma^2} +\log\sqrt{2\pi\sigma^2}\\
			\label{eq: Gaussian F}
			&=-\frac{\beta(1)^2}{4\beta(2)}-\frac{1}{2}\log(-\beta(2)) +\frac{1}{2}\log\pi.
		\end{align}
		For the social learning with statistical models \eqref{eq: statistical models in exponential family} described by Gaussian distributions, the critical equation \eqref{eq: critical function for exponential family} can be expressed explicitly by using \eqref{eq: Gaussian F}. Particularly, if the likelihood models share the same variance, i.e., $\beta_1(2)=\beta(2)$, as seen in the shift-in-mean Gaussian models \eqref{eq: mean-shift Gaussian model}, Eq. \eqref{eq: critical function for exponential family} is simplified as
		\begin{equation}
			\frac{\beta^\ast(1)^2-(\beta^\ast(1)+t(\beta_{1}(1)-\beta(1))^2}{4\beta^\ast(2)}+t\frac{\beta_{1}(1)^2-\beta(1)^2}{4\beta_{1}(2)}=0
		\end{equation}
		In this case, the expression for $t_k^{\sf nc}(\theta)$ admits a closed form:
		\begin{equation}\label{eq: t_k estimate for shift-in-mean Gaussian models}
			t_k^{\sf nc}(\theta)=\frac{\beta^\ast(2)(\beta_{1}(1)+\beta(1))-2\beta_{1}(2)\beta^\ast(1)}{\beta_{1}(2)(\beta_{1}(1)-\beta(1))}.
		\end{equation}
		Therefore, once we have obtained an estimate $\widehat{\beta}^\ast$ of the natural parameter $\beta^\ast$, an approximate $t_k^{\sf nc}(\theta)$ can be derived explicitly.
	
	Given a prescribed Perron eigenvector, the next question is how to construct the combination policy for the given Perron eigenvector. In the following, we comment on some useful results pertaining to this problem.
	\subsubsection{Construction of combination policies}
	\label{sec: construction of combination policy}
	
	First, our objective is to design the combination weights of a left-stochastic matrix $A$ that generates a specified Perron eigenvector $\pi$ for the given communication network. Therefore, there are two constraints in the design: i) the required Perron eigenvector and ii) the fixed directed graph. If we relax the second constraint and assume that the structure of the communication network can be freely designed, then this combination policy construction problem can be cast into a special case of the canonical partially described inverse eigenvalue problem (PDIEP), which we describe next.  
	
	The PDIEP is one kind of the general IEP that involves the reconstruction of a matrix for the given spectral data, i.e., the partial or complete information of eigenvalues or eigenvectors \cite{chu2005inverse}. In PDIEPs, the spectral constraint consists of only one or few eigenpairs (i.e., the pair of an eigenvalue and the corresponding eigenvector). In the traditional approach to solving PDIEPs, both analytical and numerical methods have conventionally been tailored to address specific \emph{structured} matrices, including Jacobi, Toeplitz, or quadratic pencils. From \cite{chu2005inverse}, there is a lack of general systematic studies for PDIEPs in the literature. Further investigation on PDIEPs for most other matrix structures are still needed.
	
	There are two fundamental questions associated with any PDIEP (and indeed, any IEP): the theory of solvability and the practice of computability. The solvability concerns determining the condition under which a PDIEP has a solution. Provided that the given spectral data is feasible, the computability involves developing numerical methods to construct a desirable matrix. According to the conclusion drawn in \cite{chu2005inverse}, both questions are difficult and challenging, and complete answers are yet to find. Some important attempts have been made in \cite{borges1995some,chu1994symmetric,hu2008systematic} to numerically solve PDIEPs for the particular structured matrices, such as Jacobi, Toeplitz, and quadratic pencil. 
	
	Returning to our problem of combination policy construction, we know that constraint i) is a special case of the spectral constraint for PDIEPs. This suggests that without constraint ii), one solution to our problem could be first designing a suitable matrix structure within the solvability theory of PDIEPs, and then resorting to the numerical methods proposed in \cite{borges1995some,chu1994symmetric,hu2008systematic} to find a desirable combination policy. However, if the topology of the communication network is predetermined and constraint ii) must be satisfied, we cannot directly apply the results from  \cite{borges1995some,chu1994symmetric,hu2008systematic}. This is because constraint ii) defines a generic and specific matrix structure for the PDIEP, which remains an open question in the literature according to \cite{chu2005inverse}. 
	
	Overall, as we described above, constructing a combination policy with the prescribed Perron eigenvector for a given directed graph is a challenging task. Nonetheless, one exception to this is when the graph is undirected and the agents all have self-loops. In this scenario, we can employ a particular rule to generate a desired combination policy $A$ for any given Perron eigenvector $\pi$, which complies with the predetermined network topology \cite{vlaski2021regularized}:
	\begin{equation}\label{eq: construct combination policy}
		a_{\ell k}=\begin{cases}
			0, &\text{if } \ell\notin \mathcal{N}_k,\\
			\pi_\ell, &\text{if } \ell\in\mathcal{N}_k\setminus\{k\},\\
			1-\sum_{m\in\mathcal{N}_k\setminus\{k\}} a_{mk}, &\text{if } \ell=k.
		\end{cases}
	\end{equation}
	It is worth noting that the existence of a self-loop means that the agent will use its local observations for the belief updating in \eqref{eq: ASL 1}, which is a common assumption in distributed learning over graphs. Therefore, if possible, we can always assume the network structure to be undirected and utilize rule \eqref{eq: construct combination policy} for the combination policy design in the implementation of adaptive social learning. 
	
	\section{Minimizing adaptation time}
	\label{sec:adaptation time}
	Section \ref{sec:error exponent} derived the optimal Perron eigenvectors for combination policies that minimize the \emph{steady-state} error probability. In this section, we investigate the effect of combination policies on the adaptation time of social learning (i.e., on the \emph{transient} learning performance). The adaptation time is defined as the critical time instant $i$ after which the instantaneous error probability is decaying with an error exponent $(1-\omega)\Phi(\pi)$ for some small $\omega>0$ \cite{Bordignon:2020wx}:
	\begin{equation}\label{eq: inst error pb ub}
		p_{k,i}\leq e^{-\frac{1}{\delta}[(1-\omega)\Phi(\pi)+\cal{O}(\delta)]}, \quad \forall k\in\mathcal{N},
	\end{equation}
	where the notation $\cal{O}(\delta)$ signifies that the ratio $\cal{O}(\delta)/\delta$ stays bounded as $\delta \to 0$. To avoid confusion, we note that in expression \eqref{eq: inst error pb ub}, the parameter $\omega$ is free and can be designed by the user. Basically, $\omega$ describes the user's perception of the transient period, i.e., the moment from which the learning process has entered into the steady-state region. A smaller $\omega$ requires that the instantaneous error probability of each agent is dominated by a larger error exponent when the the steady-state region is reached, which entails a larger adaptation time. Since the error exponent $\Phi(\pi)$ is associated with the slow adaptation regime, we note that the following discussion on the adaptation time are also within this regime. In order to avoid some redundancy, this dependence will not be emphasized in the remainder of this part.
	
	Due to the term $\mathcal{O}(\delta)$ in definition \eqref{eq: inst error pb ub},  there exist different approximations for the adaptation time that satisfy \eqref{eq: inst error pb ub}. One approximation for the adaptation time,  denoted by $\sf{T_{ASL}}(\pi,\omega)$, is provided in \cite{Bordignon:2020wx}. Consider the unfavorable case such as the uniform initial belief condition, the expression of $\sf{T_{ASL}}(\pi,\omega)$ is given by
	\begin{equation}\label{eq: T_ASL}
		{\sf T_{ASL}}(\pi,\omega)\triangleq\frac{1}{\log(1-\delta)^{-1}}\log\frac{{\sf K}_1(\pi)}{\omega\Phi(\pi)}
\end{equation}
for all $\omega<\frac{{\sf K}_1(\pi)}{\Phi(\pi)}$, where 
\begin{equation}\label{eq: T_ASL K1}
	{\sf K}_1(\pi) \triangleq \max_{\theta\neq\theta_1}\abs{t_\theta^\star(\pi)}{\sf m_{ave}(\pi,\theta)}
\end{equation}
with $t_\theta^\star(\pi)$ defined by the forthcoming \eqref{eq: t_theta(pi)}. For our purpose of comparing different combination policies in this work, an approximation that decouples the influence of the combination policy from other factors would be preferred. However, this is generally unattainable due to the difficulty in calculating the instantaneous error probability of each agent and the intricate relation between the error exponent and the Perron eigenvector embedded in the LMGF $\Lambda_{\sf ave}(t;\pi,\theta)$. In the following, we examine the learning tasks in the low SNR regime where the error probabilities need not be too small \cite{Matta:2016kh} and $\Lambda_{\sf ave}(t;\pi,\theta)$ can be approximated by a second-order polynomial for $t\in[t^\star_\theta(\pi),0]$. That is,
	\begin{equation}\label{eq: second-order polynomial}
		\Lambda_{\sf ave}(t;\pi,\theta)=\sum_{n=1}^\infty\frac{\kappa_n(\pi,\theta)t^n}{n!}\approx\kappa_1(\pi,\theta) t+\frac{\kappa_2(\pi,\theta)}{2}t^2
	\end{equation}
	with $\kappa_1(\pi,\theta)={\sf m_{ave}}(\pi,\theta)$ and $\kappa_2(\pi,\theta)={\sf c_{ave}}(\pi,\theta)$. Here, ${\sf c_{ave}}(\pi,\theta)$ is  the variance of $\bm{x}_{{\sf ave},i}(\pi,\theta)$: 
	\begin{equation}
		{\sf c_{ave}}(\pi,\theta)\triangleq \mathbb{E}\left[(\bm{x}_{{\sf ave},i}(\pi,\theta)-{\sf m_{ave}}(\pi,\theta))^2\right]=\sum_{k=1}^N{\pi^2_k \rho_k(\theta)},
	\end{equation}
	where $\rho_{k}(\theta)$ denotes the variance of $\bm{x}_{k,i}(\theta)$: 
	\begin{equation}
		\rho_{k}(\theta)\triangleq \mathbb{E}\left[(\bm{x}_{k,i}(\theta)-d_k(\theta))^2\right].
	\end{equation}
	
	\noindent Since the parabolic approximation of an LMGF is actually a Gaussian approximation,  the approximation in \eqref{eq: second-order polynomial} becomes an equality if, and only if, the log-likelihood ratio $\bm{x}_{{\sf ave},i}(\pi,\theta)$ follows a Gaussian distribution (e.g., in the canonical shift-in-mean Gaussian problems discussed in Section \ref{subsec: interesting cases}).  For non-Gaussian cases, \eqref{eq: second-order polynomial} will be a valid approximation only for learning tasks in the low SNR regime \cite{Matta:2016kh}. 
	The exact definition of the low SNR regime depends on the specific learning setup, but it generally includes the scenarios where the hypotheses are close to each other, i.e., when the learning task is difficult. For instance, this regime is related to detecting weak signals in the framework of \emph{locally optimum detection} \cite{braca2009asymptotic,kassam2012signal}. In the low SNR regime, we can derive an explicit approximation result for the adaptation time of the ASL strategy.
	\begin{theorem}[\textbf{Adaptation time for the low SNR regime}]\label{theorem: adaptation time}
		Suppose the uniform initial belief condition $\bm{\lambda}_{k,0}(\theta)=0, \forall k\in\mathcal{N},\forall\theta\in\Theta$, and the low SNR regime. The adaptation time of the ASL strategy can be approximated by ${\sf T_{adap}}(\omega)$ expressed as 
		\begin{equation}\label{eq: adaptation time definition}
			{\sf T_{adap}(\omega)\triangleq}\dfrac{\log(1-\sqrt{1-\omega})}{\log(1-\delta)}
		\end{equation}
		for any combination policy $\pi\in\Pi$.
	\end{theorem} 
	\begin{proof}
		See Appendix \ref{appendix: adaptation time}.
	\end{proof}
	
	\noindent Theorem \ref{theorem: adaptation time} indicates that when the hypotheses are hard to distinguish, the combination policy does not play an important role in the adaptation time of the ASL strategy. Instead, it is the step-size $\delta$ that plays the dominant role in the time of adaptation, which is consistent with the analysis in \cite{Bordignon:2020wx}. This result also differs from the analogous results for non-adaptive social learning found in \cite{jadbabaie2013information,Shahrampour:2016kk}, where the importance of the combination policy in the transient learning performance is highlighted. Theorem \ref{theorem: adaptation time} ensures that choosing a combination policy with better steady-state learning performance, as suggested by optimizing the error exponent in Theorems \ref{theorem: optimal Perron eigenvector}--\ref{theorem: epsilon-optimal Perron eigenvector}, does not negatively impact the transient learning performance in the low SNR regime. Furthermore, ${\sf T_{adap}}(\omega)$ depends only on $\omega$ and $\delta$ in \eqref{eq: adaptation time definition}, so it is applicable for all learning models that admit the parabolic approximation in \eqref{eq: second-order polynomial}. 
	
	Recalling the definition of adaptation time in \eqref{eq: inst error pb ub}, an identical adaptation time means that at any time instant $i$, the instantaneous error probability $p_{k,i}$ corresponding to a larger error exponent $\Phi(\pi)$ has a smaller upper bound. Hence, combining the conclusions from Theorems \ref{theorem: optimal Perron eigenvector}, \ref{theorem: epsilon-optimal Perron eigenvector}, and \ref{theorem: adaptation time} for the learning tasks in the low SNR regime, we can expect that when the learning step-size $\delta$ is small, both the steady-state and the instantaneous error probability of an adaptive social learning process will be reduced by employing a combination policy corresponding to a larger error exponent. This point will be further illustrated in the simulations.
	
	As a final remark, we would like to make some comments on the influence of combination policies in the high SNR regime. It is important to note that there is no uniform definition for a low or high SNR regime within social learning tasks. However, the accuracy of the Gaussian approximation \eqref{eq: second-order polynomial} is a key factor in distinguishing these regimes. It turns out that the error probability in the low SNR regime does not need to be too small, and its evolution can be simulated by an affordable number of Monte Carlo runs. In the high SNR regime, the hypotheses are more distinguishable and accordingly, the error probability decreases too rapidly to be captured by inexpensive simulations. To examine the effect of combination policies on the adaptation time in the high SNR regime, we can refer to the approximation $\sf T_{ASL}(\pi,\omega)$ provided in \eqref{eq: T_ASL}. Although the influence of the combination policy is intertwined with that of the agent's signal structure in $\sf T_{ASL}(\pi,\omega)$, a finer analysis presented in \cite{Bordignon:2020wx} reveals that the step-size $\delta$ still maintains a dominant role in the adaptation time. One exception is the favorable case where the agent's initial belief is already biased towards the true hypothesis. In this case, the adaptation time is essentially determined by the mixing rate of agent's beliefs, which is related to the second largest-magnitude eigenvalue of combination policy $A$.
	
	\section{Computer Simulations}
	\label{sec:simulation}
	In this section, we present simulation results to illustrate our findings. We consider both the learning task with accurate signal model (i.e., Case 2 in Section \ref{subsec: interesting cases}) and that with the noisy shift-in-mean Gaussian model (i.e., Case 3 in Section \ref{subsec: interesting cases}). For these two tasks, we can see from \eqref{eq: mave >0 in Corollary 2} and \eqref{eq: mave>0 in Colollary 3} that the consistency condition in \eqref{eq: consistency condition} holds for all Perron eigenvectors in \eqref{eq: Perron eigenvector}. Moreover, according to \eqref{eq: achieves all upper bounds in Corollary 2} and \eqref{eq: achieves all upper bounds in Corollary 3}, the optimal Perron eigenvectors provided in Corollaries \ref{Corollary  2} and \ref{Corollary  3} are independent of the choice of the global truth. This means that they are optimal in the non-stationary environment with time-varying true states. We will furthermore examine the influence of the combination policy on both the learning and adaptation abilities of the ASL strategy.
	\subsection{Social learning with accurate signal model}\label{subsec: accurate likelihood models}
	
	In this part, we consider an Erd{\"o}s-R{\'e}nyi random network \cite{newman2018networks} of 10 agents where each edge is generated with probability 0.5. The undirected graph of the resulting network is shown in the left panel of Fig. \ref{fig: network topology}. We also assume that all agents have self-loops (not shown in Fig. \ref{fig: network topology}). The agents in the network will perform the ASL protocol \eqref{eq: ASL0} and \eqref{eq: ASL 1} with three hypotheses $\{\theta_1,\theta_2,\theta_3\}$. We consider a family of Laplace likelihood functions with scale parameter 1 for all agents:
	\begin{equation}
		F_h({\xi})\triangleq0.5\exp{-\abs{\xi-0.1h}},\quad\forall\xi\in\mathcal{X}_k
	\end{equation}
	with $h\in\{0,1,2\}$. The local likelihood models of each agent are shown in the right panel of Fig. \ref{fig: network topology}.
	\begin{figure}
		\centering
		\includegraphics[width=\linewidth]{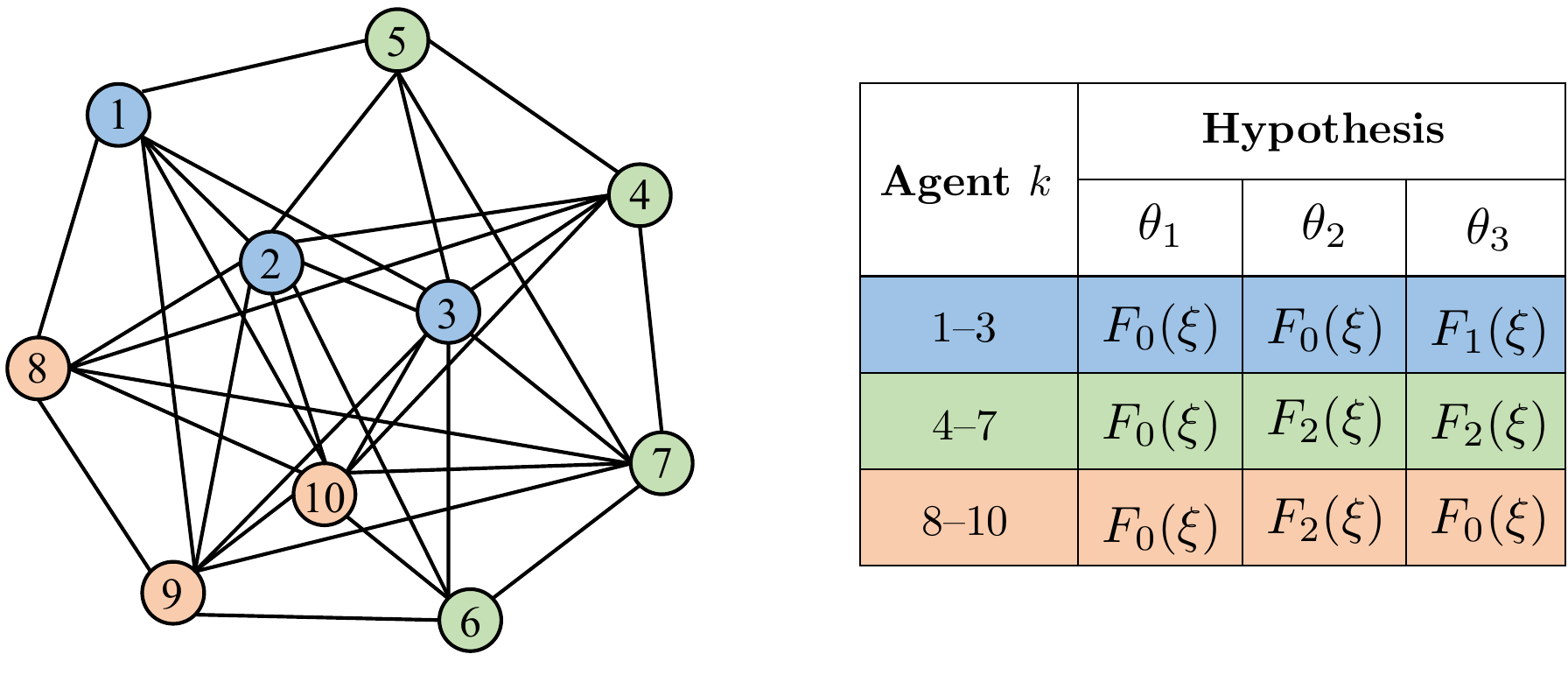}    
		\caption{The undirected graph of the network (left) and the local likelihood models of each agent (right).}
		\label{fig: network topology}
	\end{figure}
	According to Corollary \ref{Corollary  2}, the uniform Perron eigenvector is an optimal solution to the error exponent maximization problem \eqref{eq: objective function}--\eqref{eq: constraint 2}. To compare the performance of the ASL strategy under different combination policies, we employ 5 left-stochastic combination matrices $A_1$--$A_5$ and 5 doubly-stochastic combination matrices $A_6$--$A_{10}$. These matrices are generated by an iterative method based on the given network topology. For each matrix $A_i$, the iterative method starts by generating an initial matrix $A_0$ that conforms to the network topology, i.e., $[A_0]_{\ell k}>0$ if $\ell\in\mathcal{N}_k$ and $[A_0]_{\ell k}=0$ otherwise. A left-stochastic combination matrix $A_i$ is then constructed by normalizing each column of $A_0$. If $A_i$ should be doubly-stochastic, then row and column normalization is performed alternatively until convergence. All the combination matrices $A_1$--$A_{10}$ are checked to have a positive Perron eigenvector for truth learning. 
	
	First, we study the influence of combination policies on the error exponent. To evaluate the steady-state error probability, we consider a stationary environment where the true hypothesis is selected as $\theta_1$. We further assume that all agents hold a uniform initial belief. Table \ref{table:critical parameters} lists the $\theta$-related error exponents $\Phi(\pi,\theta)$ and their approximation $\widehat{\Phi}(\pi,\theta)$ based on \eqref{eq: second-order polynomial} (see expression in \eqref{eq: approximated phi}). It is observed that for all 10 combination matrices, the error exponent in this learning task is determined by the wrong hypothesis $\theta_3$. Moreover, the doubly-stochastic matrices deliver a larger error exponent than all left-stochastic ones. Therefore, we can expect that as $\delta\to 0$, the steady-state error probability vanishes to zero with a larger decaying rate when the doubly-stochastic combination policy is employed. 
	\begin{table}
		\caption{$\Phi(\pi,\theta)$ and $\widehat{\Phi}(\pi,\theta)$ for different combination matrices}
		\centering
		\begin{tabular}{c c c c c c c c c}
			\toprule
			\textbf{Matrix $A$} & $\Phi(\pi,\theta_2)$ & $\widehat{\Phi}(\pi,\theta_2)$ & $\Phi(\pi,\theta_3)$ & $\widehat{\Phi}(\pi,\theta_3)$  \\ [0.6ex] 
			\midrule
			$A_1$ & $0.0578$ & $0.0582$ & $0.0392$ & $0.0394$ \\
			$A_2$ & $0.0464$ & $0.0467$ & $0.0377$ & $0.0379$ \\
			$A_3$ & $0.0536$ & $0.0539$ & $0.0360$ & $0.0362$ \\
			$A_4$ & $0.0567$ & $0.0571$ & $0.0329$ & $0.0331$ \\
			$A_5$ & $0.0543$ & $0.0546$ & $0.0352$ & $0.0354$ \\
			\boxit{2.75in}$A_6$--$A_{10}$ & $0.0656$ & $0.0660$ & $0.0447$ & $0.0450$ \\ 
			\bottomrule
		\end{tabular}
		\label{table:critical parameters}
	\end{table}
	
	In Fig. \ref{fig: steady-state error probability}, the steady-state error probability averaged across agents, is presented for 10 combination polices and different step-sizes. For each step-size, we select the terminal time as 3000 and run 1000000 Monte Carlo simulations to obtain the average results. It can be observed that the doubly-stochastic combination matrices $A_6$--$A_{10}$ lead to a similar steady-state error probability that is smaller than those corresponding to left-stochastic combination matrices  $A_1$--$A_5$. In Fig. \ref{fig: node's steady-state error probability}, we provide the steady-state error probability of 4 agents. Despite the slight differences for different agents, the advantage of doubly-stochastic combination matrices becomes more pronounced as $\delta$ decreases. This is consistent with our conclusion in Corollary \ref{Corollary  2}.
	\begin{figure}
		\centering
			\centering
			\includegraphics[width=.8\linewidth]{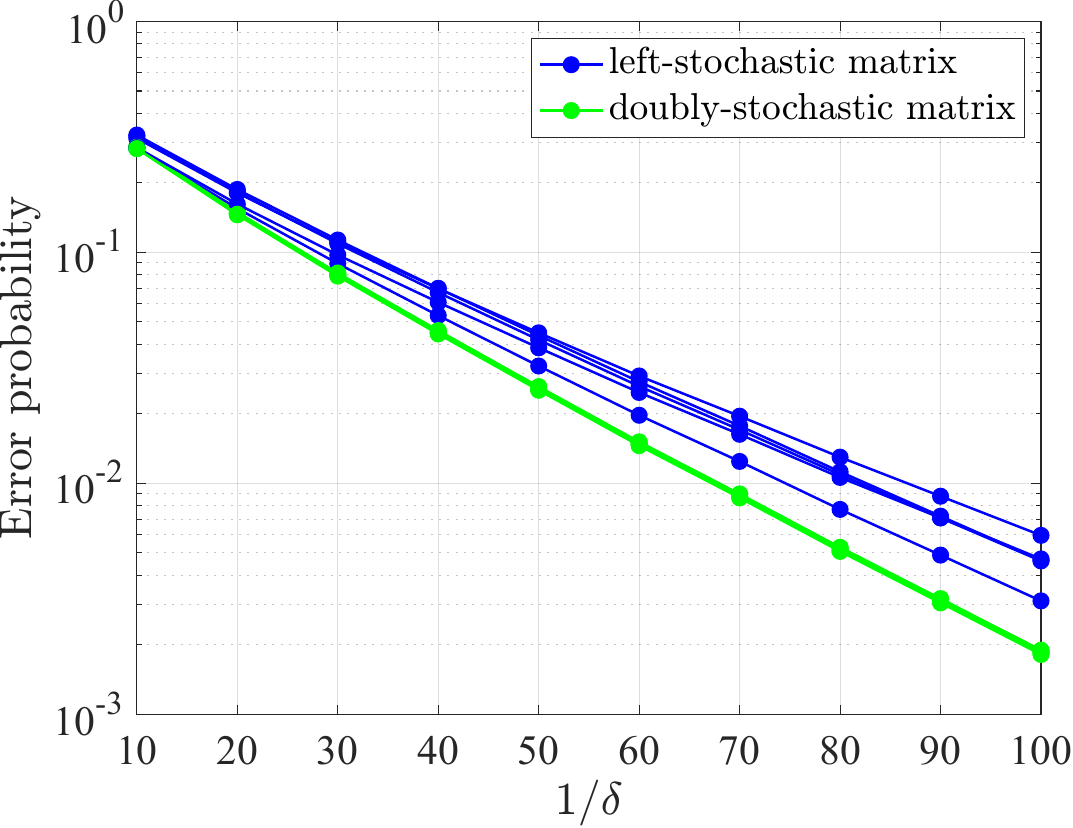}    
			\caption{Average steady-state error probability.}
			\label{fig: steady-state error probability}
	\end{figure}
	\begin{figure}
		\centering
		\includegraphics[width=.8\linewidth]{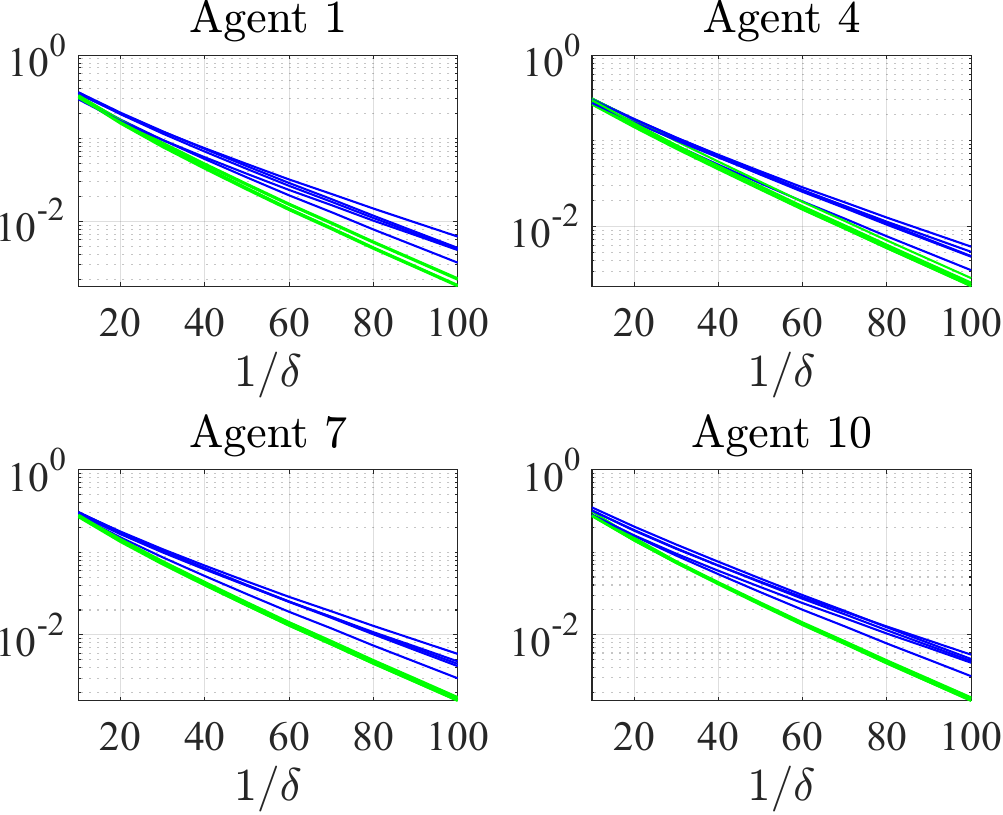}    
		\caption{Steady-state error probability of 4 agents. Blue (green) lines denote left (doubly)-stochastic matrices.}
		\label{fig: node's steady-state error probability}
	\end{figure}

	Next, we study the effect of combination policies on the adaptation time. Here, we consider a non-stationary environment where the true state changes from $\theta_1$ to $\theta_2$ at $i=1000$ and from $\theta_2$ to $\theta_3$ at $i=2000$. Under a small step-size $\delta=0.01$ and a uniform initial belief condition, the transient dynamics of the average instantaneous error probability over $i\in[0,3000]$ is depicted in Fig. \ref{fig: transient average error probability}. An important observation is that, in the non-stationary environment, the adaptation time related to different combination matrices is very close to each other. To derive a quantitative comparison, we calculate the simulated adaptation time when the true state is $\theta_1$. Let $p_{{\sf ave},i}$ and $p_{\sf ave}$ denote the average instantaneous error probability at time instant $i$ and the average steady-state error probability, respectively. According to the definition of adaptation time in \eqref{eq: inst error pb ub}, we record the time instant $i_0$ after which the average instantaneous error probability $p_{{\sf ave},i}$ satisfies
	\begin{equation}\label{eq: simulated adaptation time}
		\log{p_{{\sf ave},i}}\leq(1-\omega)\log{p_{\sf ave}}, \quad i\geq i_0.
	\end{equation} 
	The simulated adaptation time calculated by \eqref{eq: simulated adaptation time} and its approximation ${\sf T_{adap}}(\omega)$ under different values of $\omega$ is presented in Fig. \ref{fig: adaptation time}. We also include the approximation $\sf T_{ASL}(\pi,\omega)$ \eqref{eq: T_ASL} provided in \cite{Bordignon:2020wx}. It is clear that the difference in adaptation time for all considered combination matrices is almost negligible irrespective of $\omega$ in our simulations. Consequently, as observed in Fig. \ref{fig: transient average error probability}, the doubly-stochastic combination policies contribute to a smaller instantaneous error probability during the learning process. This provides a solid foundation for employing a doubly-stochastic combination policy in this learning task. Moreover, we see from Fig. \ref{fig: adaptation time} that both approximations ${\sf T_{adap}}(\omega)$ and ${\sf T_{ASL}}(\pi,\omega)$ provide an upper bound on the simulated adaptation time. Compared with ${\sf T_{ASL}}(\pi,\omega)$, which applies to all learning tasks, ${\sf T_{adap}}(\omega)$ discussed for the low SNR regime defines a better bound and illuminates the minor role of combination policies.	
	\begin{figure}
		\centering
		\includegraphics[width=.8\linewidth]{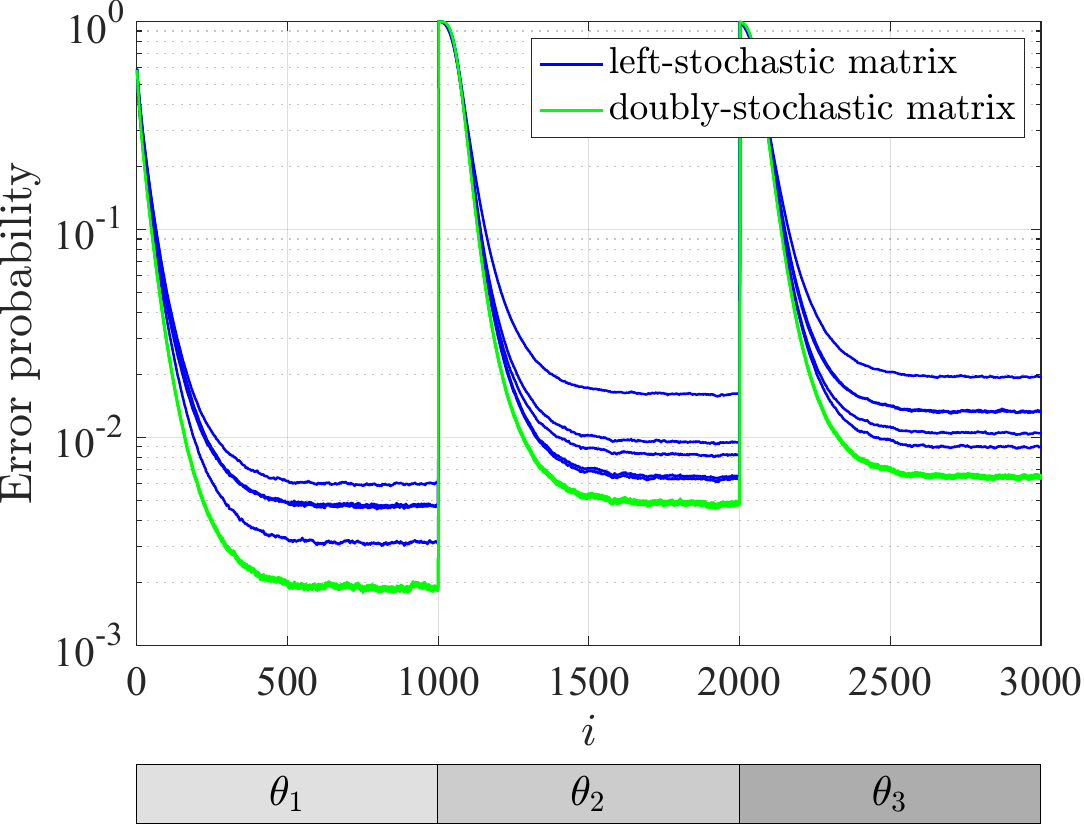}
		\caption{Average instantaneous error probability with $\delta=0.01$. The variation of the true state is depicted in the bottom panel.}
		\label{fig: transient average error probability}
	\end{figure}
	\begin{figure}
		\centering
		\includegraphics[width=.8\linewidth]{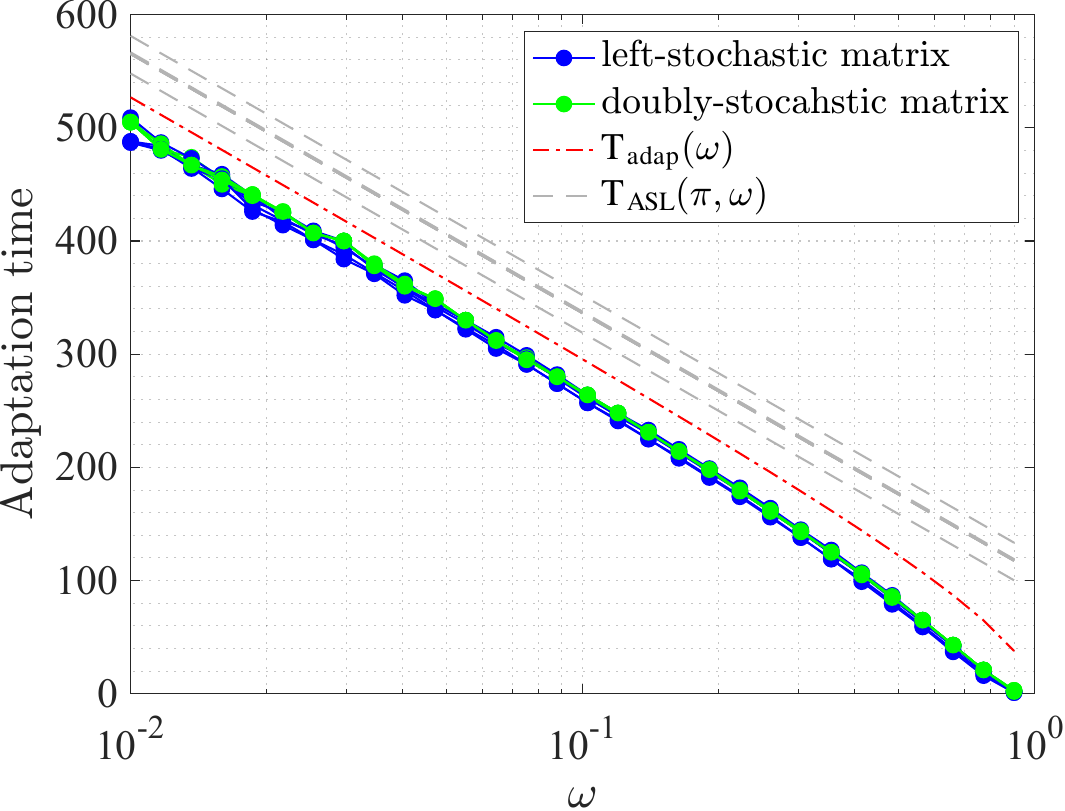}
		\caption{Adaptation time under different combination matrices.}
		\label{fig: adaptation time}
	\end{figure}
	
	\subsection{Social learning with noisy shift-in-mean Gaussian model}
	In this part, we illustrate the learning performance of the ASL strategy (Theorem \ref{theorem: learning performance for the general model}) and the optimal Perron eigenvector $\pi^\star$ given by \eqref{eq: optimal pi for noisy Gaussian model} for the noisy shift-in-mean Gaussian model \eqref{eq: mean-shift Gaussian model}. As discussed in Section  \ref{sec: construction of combination policy}, it is not always possible to construct a combination policy with the predefined Perron eigenvector for a given network topology. To avoid this inconvenience, we consider an undirected network (shown in Fig. \ref{fig: network topology}) by adding more edges to the Erd{\"o}s-R{\'e}nyi random network considered in Section \ref{subsec: accurate likelihood models}. Since we have assumed previously that each agent has a self-loop, the construction rule in \eqref{eq: construct combination policy} can be employed to generate a combination policy with any prescribed Perron eigenvector $\pi$ for this undirected network. We also consider a social learning protocol with three hypotheses. The parameters of the considered noisy shift-in-mean Gaussian model are listed in Table \ref{table:noisy Gaussian model}.
	\begin{table}
		\centering
		\caption{Parameters of the noisy shift-in-mean Gaussian model}
		\resizebox{\linewidth}{!}{
		\begin{tabular}{|c|c|c|c|c|c|c|}
			\hline
			{\textbf{Agent $k$}} & ${\sf m}_k(\theta_1)$& ${\sf m}_k(\theta_2)$& ${\sf m}_k(\theta_3)$& $\sigma_k^2$& $\varepsilon_k$ & $t_k^{\sf nc}(\theta)$ \\ 
			\hline
			1--3& 0& $0.1$& $0.1$& $1$& $1$ &$-0.5,-0.5$\\ 
			\hline
			4--7& $0.2$& $0$& $0.2$& $2$& $0.1$ &$-0.9091,-C$\\
			\hline
			8--10& $0.3$& $0.3$& $0$& $3$& $0.001$ &$-C,-0.9990$\\ 
			\hline
		\end{tabular}}
		\label{table:noisy Gaussian model}
	\end{table}	

	First, we show the learning behavior of the ASL strategy in this task when the true hypothesis is $\theta_1$. We consider the uniform averaging rule $a_{\ell k}={1}/{\abs{\mathcal{N}_k}}, \forall \ell\in\mathcal{N}_k$
	where $\abs{\mathcal{N}_k}$ is the cardinality of set $\mathcal{N}_k$. Let $\pi^{\sf u}$ be the Perron eigenvector corresponding to the uniform averaging rule, we can derive ${\sf m_{ave}}(\pi^{\sf u},\theta_2)=0.0052$ and ${\sf m_{ave}}(\pi^{\sf u},\theta_3)=0.0064$ based on Table \ref{table:noisy Gaussian model}.  With $\delta=0.002$ and a uniform initial belief, the time evolution of log-belief ratios $\bm{\lambda}_{k,i}(\theta)$ in one realization is shown in Fig. \ref{fig: log-belief ratio}. We can see that for this small step-size, the steady-state log-belief ratios of each agent concentrate around the expectation values ${\sf m_{ave}}(\pi^{\sf u},\theta)$. This shows the consistency of learning with the ASL strategy for the noisy shift-in-mean Gaussian model, as predicted by Theorem \ref{theorem: learning performance for the general model}. Furthermore, we evaluate the estimation performance of the two methods proposed in Section \ref{sec: estimation of t_k} on the quantity $t_k^{\sf nc}(\theta)$, whose value can be calculated by using \eqref{eq: noisy t_k star}. The theoretical values of $t_k^{\sf nc}(\theta)$ are listed in Table \ref{table:noisy Gaussian model}. We note that under the given parameter configuration, there are no conflicting agents in the network. For one realization, we plot the estimates of $t_k^{\sf nc}(\theta)$ for each agent within the first $10000$ observations in Fig. \ref{fig: Gaussian MGF approximation}. Under the considered noisy shift-in-mean Gaussian model, the estimate of $t_k^{\sf nc}(\theta)$ in the indirect estimation method admits a closed-form expression \eqref{eq: t_k estimate for shift-in-mean Gaussian models}. It is clear from Fig. \ref{fig: Gaussian MGF approximation} that the indirect estimation method is more efficient than the direct one in our simulations. Since $\bm{x}_{k,i}(\theta)$  is unbounded in Gaussian models, a substantial number of observations will be needed to obtain a good approximation with the direct estimation method. 
	\begin{figure}
		\centering
		\includegraphics[width=.8\linewidth]{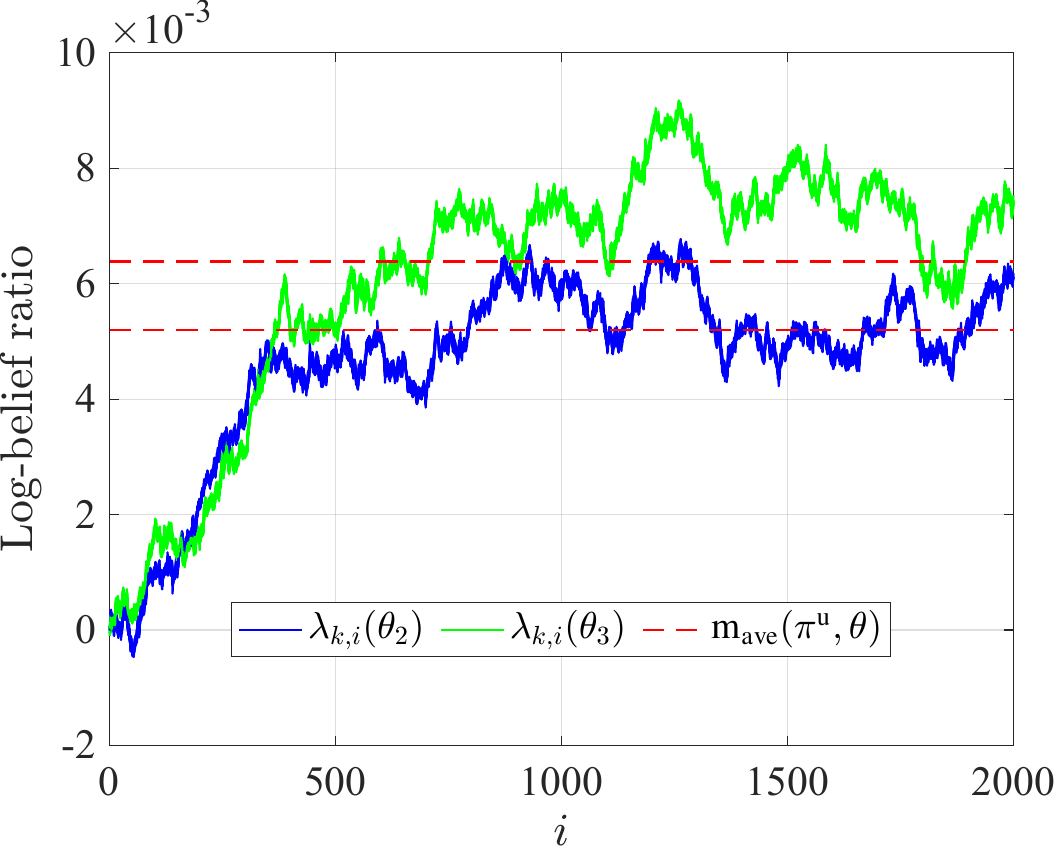}
		\caption{Evolution of log-belief ratios of each agent in one realization.}
		\label{fig: log-belief ratio}
	\end{figure}
	\begin{figure}
		\centering
		\includegraphics[width=.8\linewidth]{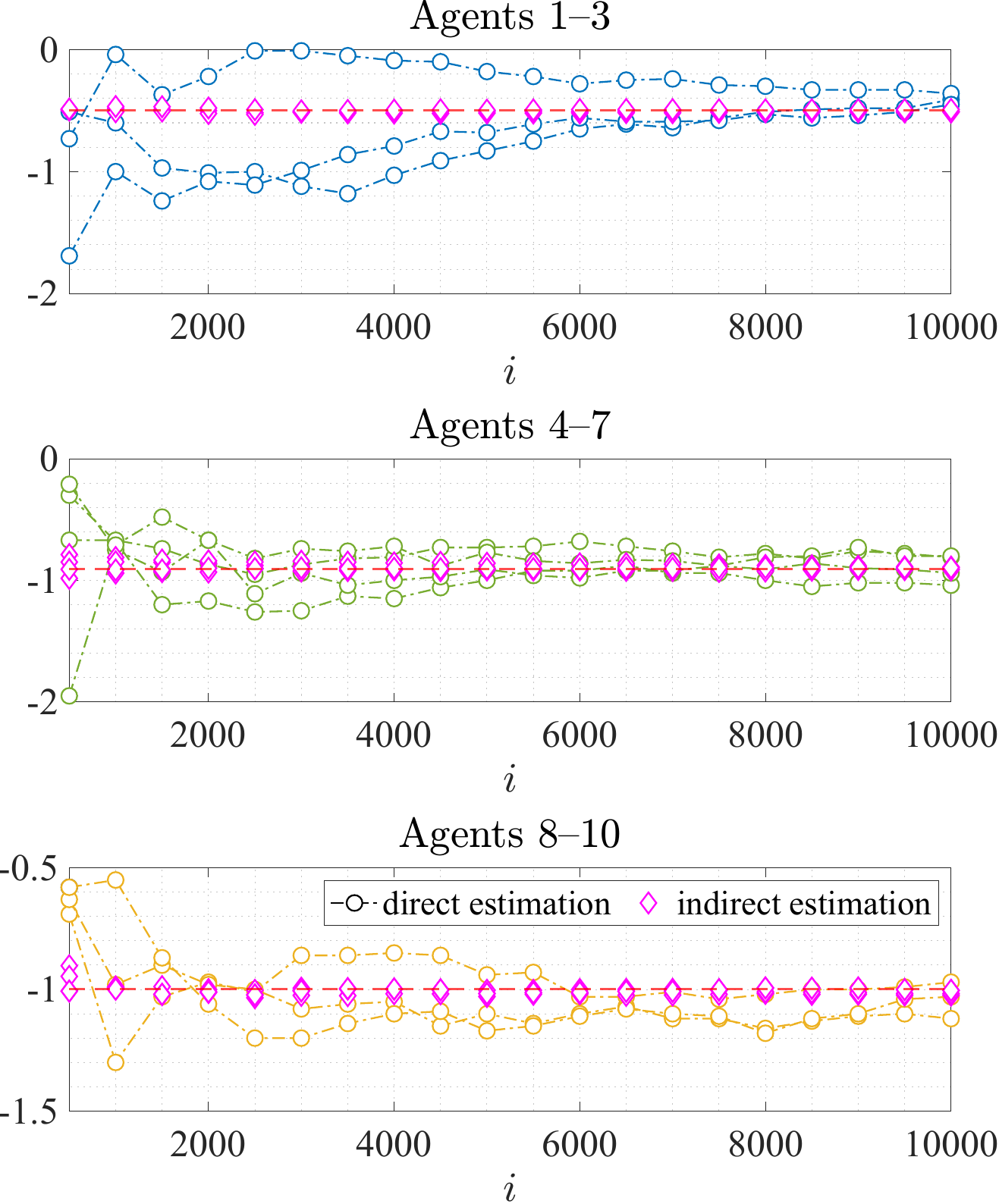}
		\caption{Estimation of $t_k^{\sf nc}(\theta)$ after every 500 observations with the direct and indirect estimation methods. The dotted red lines represent the theoretical values of $t_k^{\sf nc}(\theta)$. The discretization step-size is 0.01 in the direct estimation method. Curves related to the $\theta$-uninformative agents are removed for clarity.} 
		\label{fig: Gaussian MGF approximation}
	\end{figure}

	Next, we consider the effect of combination policies on the error exponent in this noisy learning environment. According to Corollary \ref{Corollary  3}, the maximum error exponent can be achieved by the Perron eigenvector $\pi^\star$ in \eqref{eq: optimal pi for noisy Gaussian model}. Based on the parameter setting in Table \ref{table:noisy Gaussian model}, we obtain $\pi_k^\star=0.0615$ for agents 1--3, $\pi_k^\star=0.1118$ for agents 4--7, and $\pi_k^\star=0.1228$ for agents 8--10.
	A left-stochastic combination matrix with Perron eigenvector $\pi^\star$ is constructed by following rule \eqref{eq: construct combination policy} and is denoted by $A^\star$. For comparison, we introduce another 5 left-stochastic matrices $A_1$--$A_5$ with different Perron eigenvectors and 5 doubly-stochastic matrices $A_6$--$A_{10}$. All 11 combination matrices have a positive Perron eigenvector. With a uniform initial belief, the steady-state error probability of the network and 4 agents for different small step-sizes are presented in Figs. \ref{fig: noisy steady state} and \ref{fig: noisy steady state of each node}. For each step-size, the terminal time is 3000 and the number of Monte Carlo simulations is $1000000$. It is seen that the combination matrix $A^\star$ with Perron eigenvector $\pi^\star$ leads to a larger error exponent than that of all other combination policies in comparison.  This is consistent with Corollary \ref{Corollary  3}.
	\begin{figure}
		\centering
		\includegraphics[width=.8\linewidth]{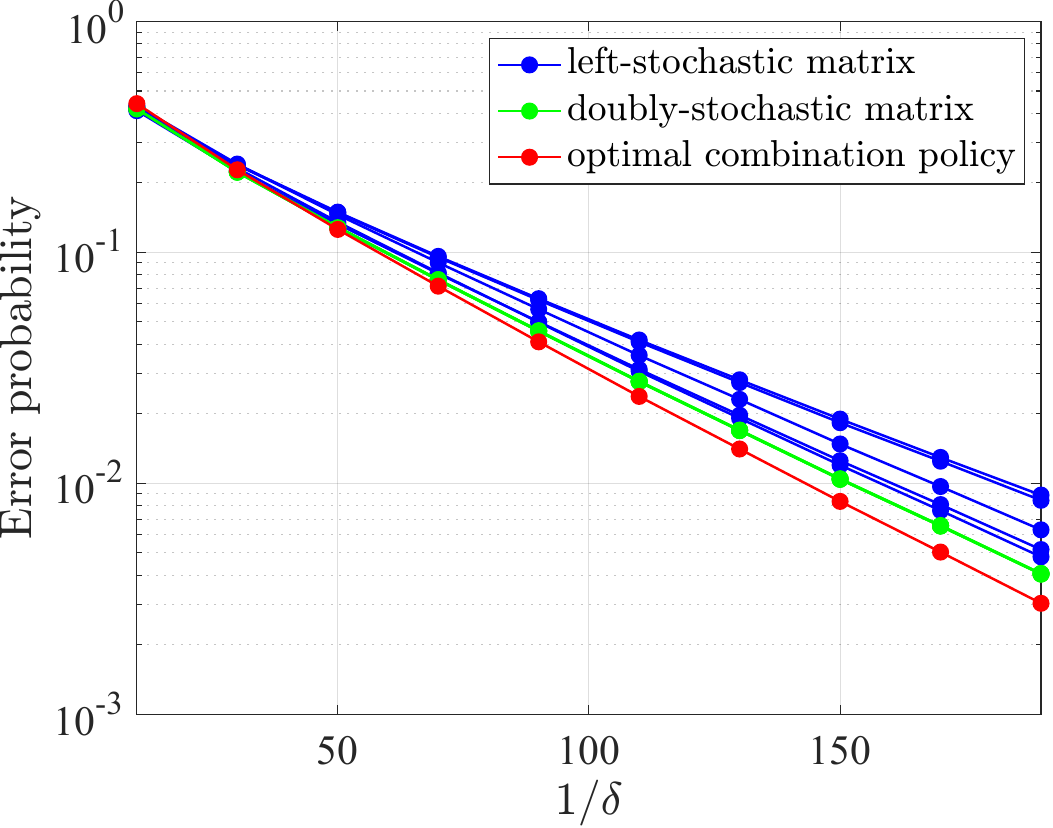}
		\caption{Average steady-state error probability. }
		\label{fig: noisy steady state}
	\end{figure}
	\begin{figure}
		\centering
		\includegraphics[width=.8\linewidth]{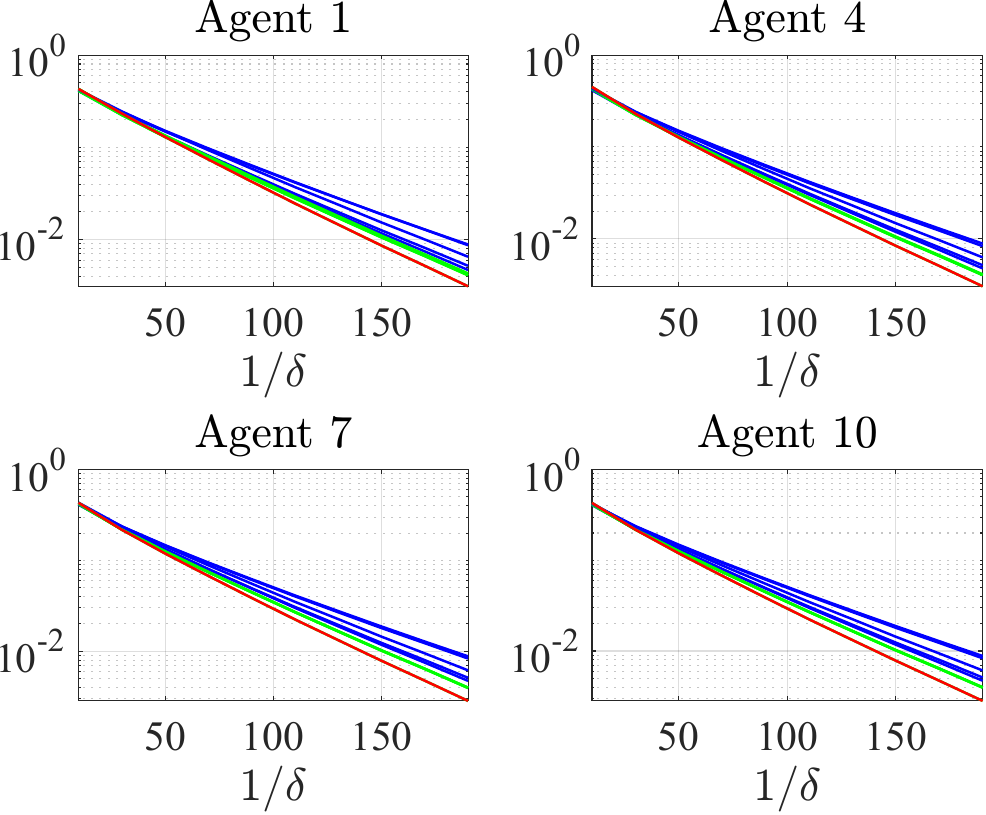}
		\caption{Steady-state error probability of each agent. The curve labels are the same as Fig. \ref{fig: noisy steady state}.}
		\label{fig: noisy steady state of each node}
	\end{figure}

	Then, we compare the adaptation time of the ASL strategy under different combination policies. From \eqref{eq: pdf of noisy log-likelihood ratio}, the parabolic approximation \eqref{eq: second-order polynomial} holds for the noisy shift-in-mean Gaussian model. We expect that ${\sf T_{adap}}(\omega)$ in Theorem \ref{theorem: adaptation time} will be a reasonable approximation of the adaptation time in the slow adaptation regime.  With $\delta=0.01$ and different $\omega$, we calculate the simulated adaptation time for each combination matrix based on \eqref{eq: simulated adaptation time} and then compare them with the approximated values $\sf T_{ASL}(\pi,\omega)$ and ${\sf T_{adap}}(\omega)$ given by \eqref{eq: T_ASL} and \eqref{eq: adaptation time definition} respectively. Corresponding results are presented in Fig. \ref{fig: noisy adaptation time}. It is shown that the effect of combination policies on the adaptation time of the ASL strategy is not significant under the given step-size. The difference between the simulated adaptation time and the theoretical one in \eqref{eq: adaptation time definition} comes from the sub-exponential term $\mathcal{O}(\delta)$ in \eqref{eq: inst error pb ub}. Similar to the observation in Fig. \ref{fig: adaptation time}, ${\sf T_{adap}}(\omega)$ provides a better characterization of the adaptation time than $\sf T_{ASL}(\pi,\omega)$ for this learning task. Fig. \ref{fig: noisy transient dynamics} presents the time evolution of the instantaneous error probability in a non-stationary environment, where the true state changes from $\theta_1$ to $\theta_2$ and from $\theta_2$ to $\theta_3$ at $i=1000$ and $2000$, respectively. Similar to Fig. \ref{fig: transient average error probability}, the curve related to the optimal combination matrix $A^\star$ is lower than all other curves during the whole learning process. Therefore, it is beneficial to employ a combination policy with the best Perron eigenvector for the steady-state learning performance.
	\begin{figure}
		\centering
		\includegraphics[width=.8\linewidth]{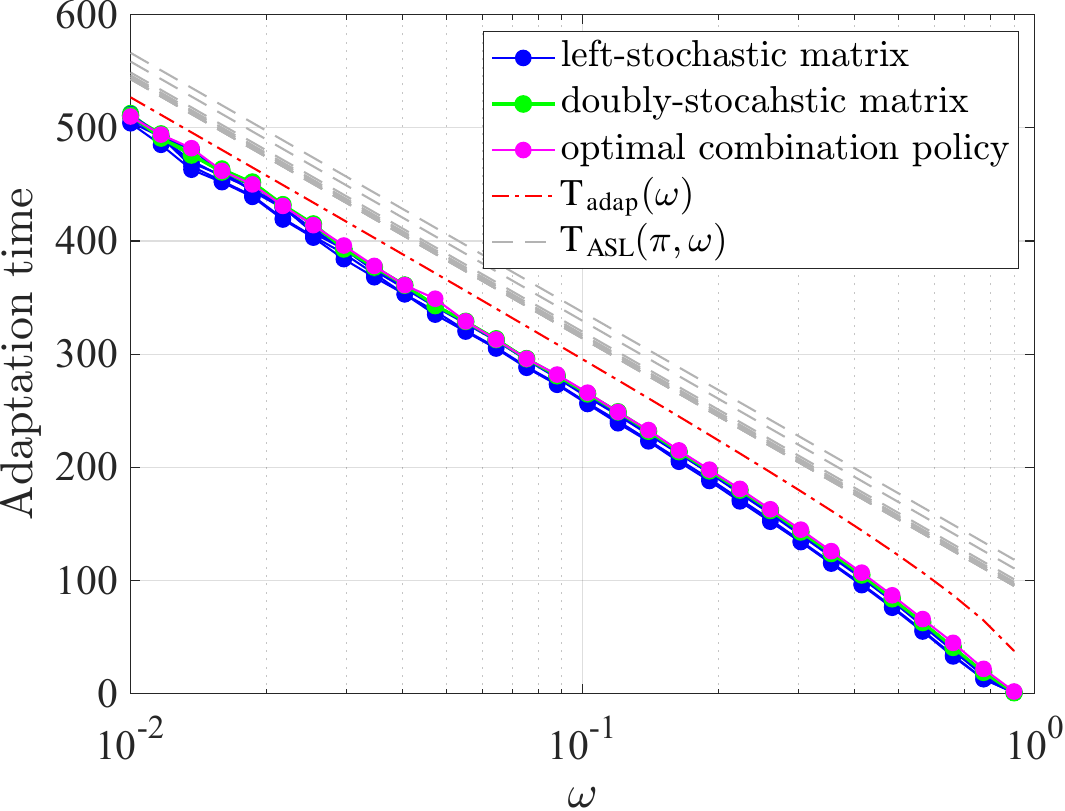}
		\caption{Adaptation time under different combination matrices.}
		\label{fig: noisy adaptation time}
	\end{figure}
	
	Finally, we examine the case where the optimal combination policy changes with the true state of nature. To simulate this scenario, we assume that the noise level in the private signals is dependent on the underlying state. Specifically, we consider $\varepsilon_k=0.1$ (or $0.001$) for agents 1--3, $\varepsilon_k=0.001$ (or 1) for agents 4--7, and $\varepsilon_k=1$ (or 0.1) for agents 8--10 when the true state is $\theta_2$ (or $\theta_3$). Based on \eqref{eq: optimal pi for noisy Gaussian model}, we compute an optimal Perron eigenvector associated with each hypothesis, which is then used to design the combination policy according to \eqref{eq: construct combination policy}. The three resulting optimal combination matrices are denoted by $A_1^\star$, $A_2^\star$ and $A_3^\star$. We study a non-stationary environment similar to Fig. \ref{fig: noisy transient dynamics}, where the true hypothesis changes from $\theta_1$ to $\theta_2$ at $i=1000$ and then to $\theta_3$ at $i=2000$. The evolution of the instantaneous error probabilities associated with $A_1^\star$--$A_3^\star$ is presented in Fig. \ref{fig: noisy transient dynamics-changing}. Since the optimization problem \eqref{eq: objective function}--\eqref{eq: constraint 2} is formulated to maximize the error exponent, the optimal combination policy pertaining the underlying state will deliver a smaller steady-state error probability in the small-$\delta$ regime. On the other hand, we note that our analysis of the adaptation time does not depend on the true hypothesis. Therefore, the conclusion from Theorem \ref{theorem: adaptation time} is still valid, which implies that the optimal combination policy for the steady-state learning performance does not harm the transient behavior. This is in line with our observations in Fig. \ref{fig: noisy transient dynamics-changing}.  
	\begin{figure}
		\centering
		\includegraphics[width=.8\linewidth]{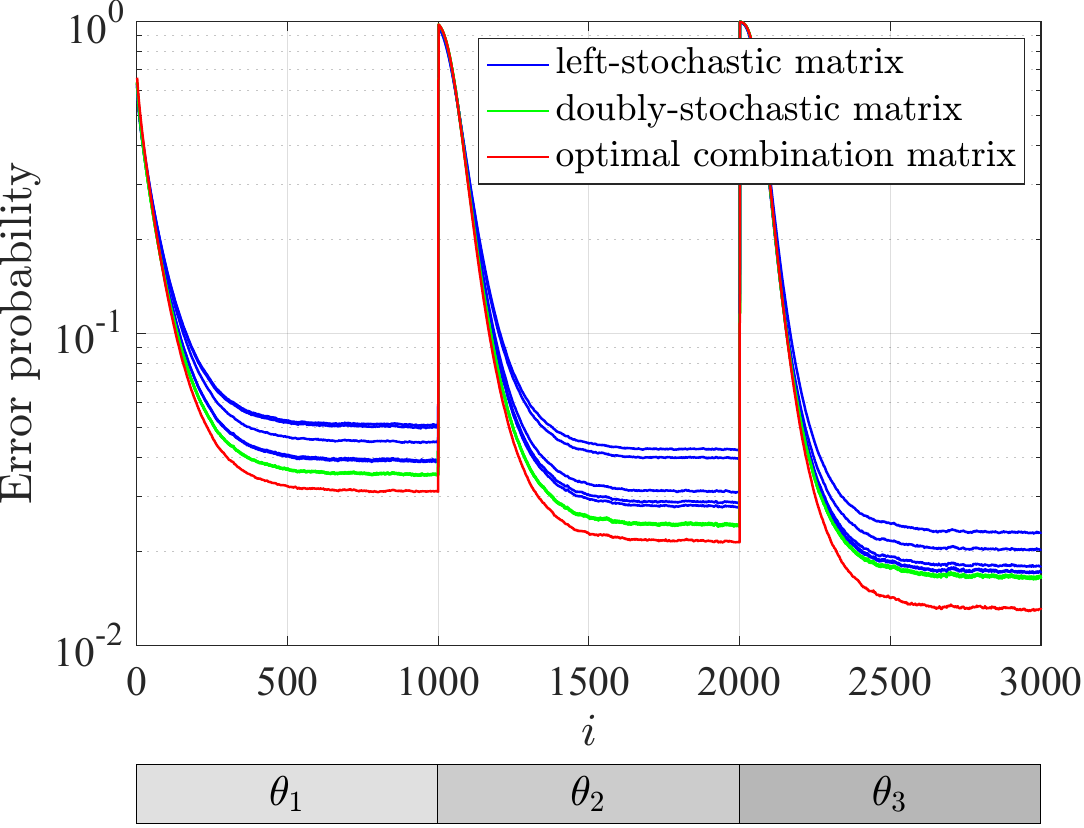}
		\caption{Average instantaneous error probability with $\delta=0.01$. The variation of the true state is depicted in the bottom panel.}
		\label{fig: noisy transient dynamics}
	\end{figure}

	\begin{figure}
	\centering
	\includegraphics[width=.8\linewidth]{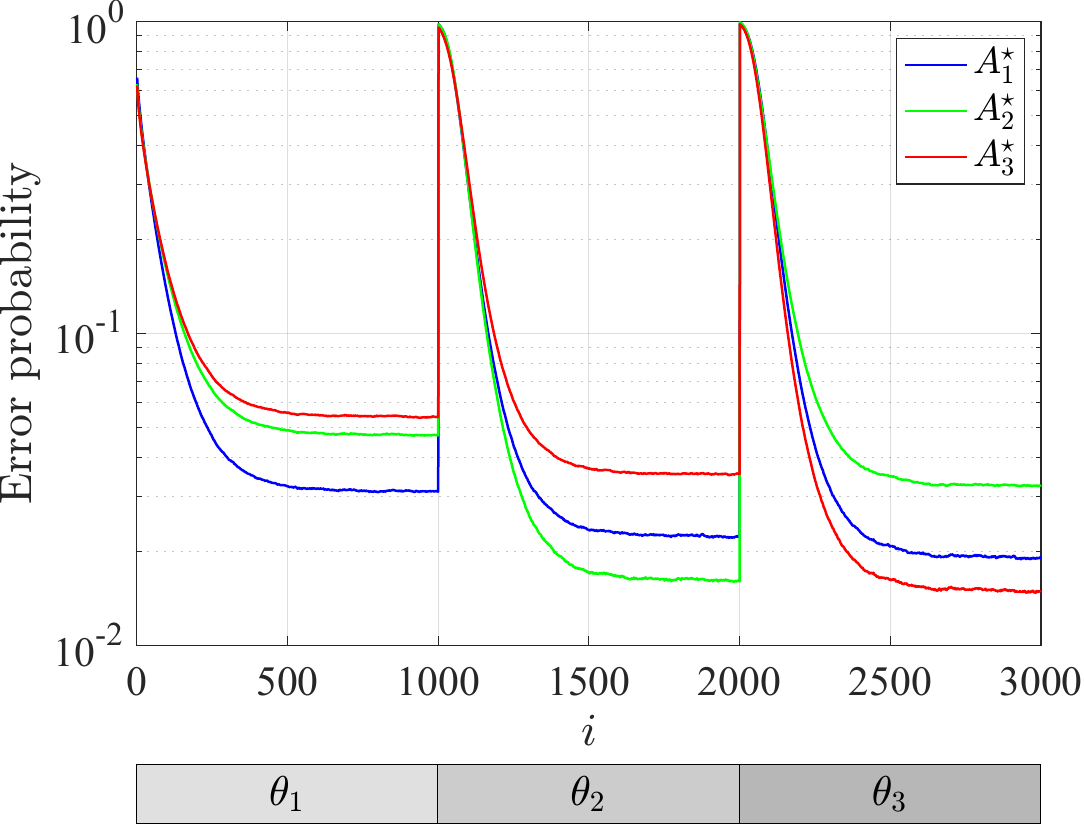}
	\caption{Average instantaneous error probability with $\delta=0.01$. The noise level changes with the variation of the true state depicted in the bottom panel.} 
	\label{fig: noisy transient dynamics-changing}
	\end{figure}

	\section{Concluding Remarks}
	\label{sec:conclusion}
	Combination policies play an important role in the behavior of social learning strategies. In this work, we discussed the effect of combination policies on two key performance metrics: the error exponent (i.e., steady-state learning ability) and the adaptation time (i.e., transient learning behavior) in the slow adaptation regime. For a general signal model, we characterized the performance limits of the error exponent and provided the set of optimal Perron eigenvectors. Moreover, we showed that the difference of the adaptation time under different combination policies is almost negligible if the hypotheses are close to each other. Our findings reveal the important relation between the learning performance of the ASL strategy and the combination policy among agents, which can be useful for the network design of many inference problems.
	
	Several useful extensions and generalizations are possible. One extension refers  to the online learning of the optimal combination policy. From Theorems \ref{theorem: optimal Perron eigenvector} and \ref{theorem: epsilon-optimal Perron eigenvector}, we know that the optimal combination policy depends on the critical value $t_k^{\sf nc}(\theta)$ that is determined by the LMGF of the log-likelihood ratio variable $\bm{x}_{k,i}(\theta)$. Therefore, full knowledge of the agents' signal models are needed to construct the optimal combination policy. This raises the question of how to learn the combination policy in an online manner, where the agents infer the underlying truth and estimate $t_k^{\sf nc}(\theta)$ with their observations in the same time. This online setting is especially important in the non-stationary environment where the optimal combination policy changes with the true state of nature as shown in Fig. \ref{fig: noisy transient dynamics-changing}.
	
	Another useful generalization concerns the compact space of hypotheses. In this work, we consider a finite number of hypotheses. It will be worthwhile to investigate the cases where the hypothesis is characterized by a continuous parameter. A good reference along these lines for the non-adaptive social learning is \cite{uribe2022nonasymptotic}. Using the similar variational interpretation of the adaptation step \eqref{eq: ASL0}, it is straightforward to generalize the ASL algorithm \eqref{eq: ASL} to the continuous space. However, new analytical approaches need to be developed for the performance analysis.
	
	\appendices
	
	\section{Useful properties for the proofs}\label{appendix: lemma}
	For the forthcoming proofs, we need some useful properties of the LMGFs $\Lambda_{k}(t;\theta)$ and functions $\phi_k^{\sf nc}(t;\theta)$ under the general signal model \eqref{eq: general signal model}--\eqref{eq: likelihood of an observation} as listed in Lemma \ref{lemma:properties}. Similar properties for the accurate signal model \eqref{eq: likelihood of an observation}--\eqref{eq: accurate likelihood model} can be found in \cite{Bordignon:2020wx,dembo1998applications,den2008large,Matta:2016kh}. The first and second derivatives of a function $f(t)$ will be denoted by $f^\prime(t)$ and $f^{\prime\prime}(t)$, respectively.
	\begin{lemma}[\textbf{Useful properties}]\label{lemma:properties}
		For each wrong hypothesis $\theta\neq\theta_1$, the following properties associated with the LMGF $\Lambda_{k}(t;\theta)$ and function $\phi_k^{\sf nc}(t;\theta)$ hold for all agents $k\notin\mathcal{N}^{ U}(\theta)$:
		\begin{enumerate}
			\item[\textit{i)}] $\Lambda_k (t;\theta)$ is infinitely differentiable in $\mathbb{R}$. Also, it is strictly convex for all $t\in\mathbb{R}$.
			\item[\textit{ii)}] $\Lambda_k(t;\theta)$ has a zero-point at $t=0$. If $d_k(\theta)=0$, this zero-point is unique, i.e., $\Lambda_k (t;\theta)>0$ for all $t\neq 0$; If $d_k(\theta)<0$, $\Lambda_k(t;\theta)$ has a positive zero-point $t_0$, i.e., $\exists t_0> 0$ such that $\Lambda_k(t_0;\theta)=0$;  If $d_k(\theta)>0$, $\Lambda_k (t;\theta)$ has a negative zero-point $t_0$. In particular, if $f_k=L_k(\cdot|\theta_1)$, we have 
			\begin{equation}\label{eq: bound of t_0}
				t_0=-1.
			\end{equation}
			\item[\textit{iii)}] The function $\frac{\Lambda_k(t;\theta)}{t}$ is continuous for all $t\in \mathbb{R}$, with
			\begin{equation}\label{eq: dphi_dt}
				\underset{t\to 0}{\lim}\frac{\Lambda_k(t;\theta)}{t}=\Lambda^\prime_k(0;\theta)=d_k(\theta).
			\end{equation}
			\item[\textit{iv)}] $\phi_k^{\sf nc}(t;\theta)$ is strictly convex for all $t\in\mathbb{R}$.
		\end{enumerate}
		Furthermore, the following properties in regard to the LMGF $\Lambda_{\sf ave}(t;\pi,\theta)$ and function $\phi(t;\pi,\theta)$ hold for all feasible Perron eigenvectors $\pi\in\Pi$:
		\begin{enumerate}
			\item[\textit{v)}] $\Lambda_{\sf ave}(t;\pi,\theta)$ and $\phi(t;\pi,\theta)$ are strictly convex for all $t\in\mathbb{R}$.
			\item[\textit{vi)}] $\Lambda_{\sf ave}(t;\pi,\theta)$ has a negative zero-point.
		\end{enumerate}
	\end{lemma} 
	\begin{proof}
		The strict convexity of the LMGF $\Lambda_{k}(t;\theta)$ in property i) is the basis of the subsequent properties ii)--vi). We assume that $\bm{x}_{k,i}(\theta)$ is non-degenerate under the general signal model \eqref{eq: general signal model}--\eqref{eq: likelihood of an observation}. Property i) is established by using Cauchy–Schwarz inequality, and the proof can be found in \cite{Matta:2016kh,dembo1998applications,den2008large}. 
		
		Since ${\Lambda}_{k}(0;\theta)=0$ and ${\Lambda}_{k}^\prime(0;\theta)=d_k(\theta)$ for all $k$ and $\theta$, $t=0$ is a zero-point of ${\Lambda}_k(t;\theta)$ for all $k$ and $\theta$.   Furthermore, the first-order condition of the strictly convex function $\Lambda_{k}(t;\theta)$ gives us
		\begin{equation}\label{eq: convexity inequality}
			\Lambda_{k}(t;\theta)> d_k(\theta)t.
		\end{equation}
		Obviously, ${\Lambda}_k(t;\theta)> 0$ for all $t\neq 0$ when $d_k(\theta)=0$. Since the signal and likelihood models share the same support, we have $\lim_{t\to\pm\infty}\Lambda_{k}(t;\theta)=\infty$. If $d_k(\theta)\neq 0$, then there exists a unique non-zero zero-point $t_0\neq 0$ such that $\Lambda_{k}(t_0;\theta)=0$.  If $d_k(\theta)<0$, we have $\Lambda_{k}(t;\theta)\geq d_k(\theta)t>0$ for all $t<0$ according to \eqref{eq: convexity inequality}. Hence, $\Lambda_{k}(t;\theta)$ has a positive zero-point in this case. Likewise,  ${\Lambda}_{k}(t;\theta)$ has a negative zero-point $t_0$ for the case of $d_k(\theta)>0$. If $f_k=L_k(\cdot|\theta_1)$, we obtain $\Lambda_k(-1;\theta)=0$, and thus  $t_0=-1$ is the unique negative zero-point of $\Lambda_{k}(t;\theta)$ as demonstrated in \cite{Bordignon:2020wx}. 
		
		Property iii) can be proved by applying L'H\^opital's rule to the continuous function $\Lambda_k(t;\theta)$ at $t=0$ \cite{Matta:2016kh}. This property further ensures that the function $\phi_k^{\sf nc}(t;\theta)$ defined in \eqref{eq: rate function individually} is well-posed. The second-order derivative of $\phi_k^{\sf nc}(t;\theta)$ satisfies
		\begin{equation}\label{eq: convexity of phi}
			\begin{aligned}
				\left[{\phi_k^{\sf nc}}(t;\theta)\right]^{\prime\prime}\triangleq\;\frac{d}{dt}\frac{\Lambda_{k}(t;\theta)}{t}
				{=}\;\frac{\Lambda^\prime_{k}(t;\theta)t-\Lambda_{k}(t;\theta)}{t^2}>0
			\end{aligned}
		\end{equation}
		due to the first-order condition of the strictly convex function $\Lambda_{k}(t;\theta)$. Property iv) follows \eqref{eq: convexity of phi}. 
		
		Under the independence assumption of local observations over space, the strict convexity of individual LMGF $\Lambda_k(t;\theta)$ for all agents $k\notin\mathcal{N}^U(\theta)$ (property i)) ensures that $\Lambda_{\sf ave}(t;\pi,\theta)$ is also a strictly convex function for $t\in\mathbb{R}$. The strict convexity of $\phi(t;\pi,\theta)$ can be established in the same way as \eqref{eq: convexity of phi}.
		Property vi) is guaranteed by the consistency condition \eqref{eq: consistency condition}, i.e., ${\sf m_{ave}}(\pi,\theta)>0$ for all $\theta\neq\theta_1$. Similar to the proof of property ii), we can prove property vi) using the strict convexity of $\Lambda_{\sf ave}(t;\pi,\theta)$ established in property v). 
	\end{proof}
	
	\section{Proof of Theorem \ref{theorem: bound}}\label{appendix: theorem 3}
	First, we show that for each $\theta$-informative agent $k\in\mathcal{N}^I(\theta)$, $t_k^{\sf nc}(\theta)$ defined in \eqref{eq: t_k for informative agent k} corresponds to a negative zero-point of $\Lambda_{k}(t;\theta)$. According to property iv) in Lemma \ref{lemma:properties}, $\phi_k^{\sf nc}(t;\theta)$ is a strictly convex function for $t\in\mathbb{R}$. Therefore, the infimum of $\phi_k^{\sf nc}(t;\theta)$ is achieved at its unique stationary point:
	\begin{equation}\label{eq: stationary point}
		[\phi_k^{\sf nc}(t;\theta)]^\prime=\frac{\Lambda_{k}(t;\theta)}{t}=0.
	\end{equation}
	From definition \eqref{eq: informative agent}, we have $d_k(\theta)> 0$ for all $k\in\mathcal{N}^I(\theta)$. Property iii) shows that the stationary point in \eqref{eq: stationary point} corresponds to the negative zero-point of $\Lambda_{k}(t;\theta)$, whose existence and uniqueness are guaranteed by property ii) in Lemma \ref{lemma:properties}. Therefore, $t_k^{\sf nc}(\theta)$ is given by
	\begin{equation}\label{eq: definition of t_k for informative agents}
		t_k^{\sf nc}(\theta)=\left\{t\neq 0: {\Lambda}_{k}(t;\theta)=0\right\} <0.
	\end{equation}
	Next, we prove the bound of the error exponent established in Theorem \ref{theorem: bound}. Based on the definition of $t_k^{\sf nc}(\theta)$ in \eqref{eq: t_k in non-cooperative scenario}, we have $t_k^{\sf nc}(\theta)\leq 0$ and
	\begin{equation}\label{eq: value of LMGF is negative}
		\Lambda_k(t;\theta)\leq 0, \quad \forall t\in[t_k^{\sf nc}(\theta),0],
	\end{equation}
	for all agents $k\in\mathcal{N}$. Moreover, $\Phi_k^{\sf nc}(\theta)=0$ holds for each $k\notin\mathcal{N}^I(\theta)$. This yields
	\begin{equation}\label{eq: sum of error exponnet of informative agents}
		\Phi_{\sum}^{\sf nc}(\theta)\triangleq\sum_{k=1}^N\Phi_k^{\sf nc}(\theta)=\sum_{k\in\mathcal{N}^{I}(\theta)}\Phi_k^{\sf nc}(\theta),\quad\forall \theta\neq\theta_1.
	\end{equation}
	For any Perron eigenvector $\pi\in\Pi$, we have ${\sf m_{ave}}(\pi,\theta)>0$, $\forall \theta\neq\theta_1$ from \eqref{eq: constraint 2}. Due to  properties v) and vi) in Lemma \ref{lemma:properties}, the infimum of $\phi(t;\pi,\theta)$ is achieved at its unique stationary point that corresponds to the negative zero-point of $\Lambda_{\sf ave}(t;\pi,\theta)$.
	Let $t_\theta^\star(\pi)$ be the critical value corresponding to the $\theta$-related error exponent $\Phi(\pi,\theta)$ in \eqref{eq: rate function}, i.e., $\Phi(\pi,\theta)=-\phi(t_\theta^\star(\pi);\pi,\theta)$. Then,  $t_\theta^\star(\pi)$ is the unique non-zero solution to $\Lambda_{\sf ave}(t;\pi,\theta)=0$:
	\begin{equation}\label{eq: t_theta(pi)}
		t_{\theta}^\star(\pi)=\{t\neq 0: \Lambda_{\sf ave}(t;\pi,\theta)=0\}<0.
	\end{equation} 
	Therefore, 
	\begin{align}
		\nonumber
		{\Phi}(\pi,\theta)&\triangleq-\inf\limits_{t\in\mathbb{R}} {\phi}(t;\pi,\theta)\overset{\eqref{eq: t_theta(pi)}}{=}-\inf_{t<0}\phi(t;\pi,\theta)\\
		\nonumber
		&\overset{\textnormal{(a)}}{=}-\inf_{t<0}\sum_{k=1}^N{\phi}^{\sf nc}_k(\pi_kt;\theta) \leq \sum_{k=1}^N -\inf_{t<0}{\phi}^{\sf nc}_k(\pi_kt;\theta)\\
		\nonumber
		&\overset{\textnormal{(b)}}{=}\sum_{k\in\mathcal{N}^I(\theta)}-\inf_{t<0}{\phi}^{\sf nc}_k(t;\theta)=\sum_{k\in\mathcal{N}^I(\theta)}{\Phi}^{\sf nc}_{k}(\theta)\\
		\label{eq: rate function comparison}
		&\overset{\eqref{eq: sum of error exponnet of informative agents}}{=}{\Phi}_{\sum}^{\sf nc}(\theta).
	\end{align}
	where (a) is due to the definitions given in \eqref{eq: LMGF_ave}, \eqref{eq: phi}, and \eqref{eq: rate function individually}, and (b) comes from \eqref{eq: uninformative agent} and \eqref{eq: phi_k conflicting agent}. Since the error exponent $\Phi(\pi)$ is determined by the minimum ${\Phi}(\pi,\theta)$ among all wrong hypotheses $\theta\neq\theta_1$, we can derive from \eqref{eq: rate function comparison} that
	\begin{equation}\label{eq: upper bound of error exponent}
		{\Phi}(\pi)=\min_{\theta\neq\theta_1}{\Phi}(\pi,\theta){\leq}\min_{\theta\neq\theta_1}{\Phi}_{\sum}^{\sf nc}(\theta).
	\end{equation}
	This implies that each aggregated quantity ${\Phi}_{\sum}^{\sf nc}(\theta)$ is an upper bound of the error exponent $\Phi(\pi)$ for all feasible combination policies. On the other hand, for a given Perron eigenvector $\pi\in\Pi$, we can define the following quantities:
	\begin{equation}\label{eq: quantity k and t}
		k_{\sf m}^{\theta}=\argmax_{k\in\mathcal{N}}\frac{t_k^{\sf nc}(\theta)}{\pi_k},\quad t_{\sf m}^{\theta}=\max_{k\in\mathcal{N}}\frac{t_k^{\sf nc}(\theta)}{\pi_k}.
	\end{equation}
	This yields $\pi_{k_{\sf m}^\theta}t_{\sf m}^\theta=t_{k_{\sf m}^\theta}^{\sf nc}(\theta)$ and $\Phi_{k_{\sf m}^\theta}^{\sf nc}(\theta)=-\phi_{k_{\sf m}^\theta}^{\sf nc}(\pi_{k_{\sf m}^\theta}t_{\sf m}^\theta;\theta)$.
	Since $t_k^{\sf nc}(\theta)\leq 0,\forall k\in\mathcal{N}$ by definition \eqref{eq: t_k in non-cooperative scenario}, \eqref{eq: quantity k and t}  gives
	\begin{equation}\label{eq: pi_ell times t_m}
		t_k^{\sf nc}(\theta)\leq\pi_k t_{\sf m}^\theta\leq 0,\quad \forall k\in\mathcal{N}.
	\end{equation}
	Then, for each $\theta\neq\theta_1$, we have
	\begin{align}
		\nonumber
		{\Phi}(\pi,\theta)&\triangleq-\inf\limits_{t\in\mathbb{R}} {\phi}(t;\pi,\theta)\geq -{\phi}(t_{\sf m}^\theta;\pi,\theta)\\
		\nonumber
		\overset{\eqref{eq: quantity k and t}}&{=}-{\phi}^{\sf nc}_{k_{\sf m}^\theta}(\pi_{k_{\sf m}^\theta}t_{\sf m}^\theta;\theta)- \sum_{\ell\neq k_{\sf m}^\theta}{\phi}^{\sf nc}_\ell(\pi_\ell t_{\sf m}^\theta;\theta)\\
		\label{eq: rate function comparison 2}
		\overset{\text{(a)}}&{\geq} {\Phi}^{\sf nc}_{k_{\sf m}^\theta}(\theta)\geq\min_{k\in\mathcal{N}}{\Phi}_{k}^{\sf nc}(\theta)
	\end{align}
	where the inequality (a) comes from \eqref{eq: value of LMGF is negative} and \eqref{eq: pi_ell times t_m}:
	\begin{equation}
		{\phi}^{\sf nc}_\ell(\pi_\ell t_{\sf m}^\theta;\theta)\triangleq\int_{0}^{\pi_\ell t_{\sf m}^\theta}\frac{{\Lambda}_\ell(\tau;\theta)}{\tau}d\tau
		\leq0,\quad \forall \ell\neq k_{\sf m}^\theta.
	\end{equation}
	Therefore, the error exponent ${\Phi}(\pi)$ is lower bounded by
	\begin{equation}\label{eq: lower bound of error exponent}
		{\Phi}(\pi)=\min_{\theta\neq\theta_1}{\Phi}(\pi,\theta){\geq}\min_{\theta\neq\theta_1}\min_{k\in\mathcal{N}}{\Phi}_{k}^{\sf nc}(\theta).
	\end{equation}
	The bounds for $\Phi(\pi,\theta)$ and $\Phi(\pi)$ are established using \eqref{eq: rate function comparison}, \eqref{eq: upper bound of error exponent}, \eqref{eq: rate function comparison 2}  and \eqref{eq: lower bound of error exponent}.
	
	\section{Proof of Theorems \ref{theorem: optimal Perron eigenvector} and \ref{theorem: epsilon-optimal Perron eigenvector}}
	\subsection{Proof of Theorem \ref{theorem: optimal Perron eigenvector}}\label{appendix: theorem 3a}
	In this case, we need to prove that $\Pi^\dagger\neq\emptyset$ is a sufficient and  necessary condition for achieving the upper bound $\Phi_{\sum}^{\sf nc}(\theta^\dagger)$, and that the set of optimal Perron eigenvectors is given by $\Pi^\star=\Pi^\dagger$ when $\Phi^\star=\Phi_{\sum}^{\sf nc}(\theta^\dagger)$. First, we show that $\Pi^\dagger\neq\emptyset$ is a sufficient condition for delivering $\Phi^\star=\Phi_{\sum}^{\sf nc}(\theta^\dagger)$.  From Theorem \ref{theorem: bound}, ${\Phi}_{\sum}^{\sf nc}(\theta^\dagger)$ is an upper bound of the error exponent, and thus $\Phi^\star\leq \Phi_{\sum}^{\sf nc}(\theta^\dagger)$. Consider any Perron eigenvector $\pi\in\Pi^\dagger$, by definition of $\Pi_1$ in \eqref{eq: Pi_1-necessary conditions}, we have $\Phi(\pi,\theta^\dagger)\leq \Phi(\pi,\theta)$, $\forall \theta\neq\theta_1$. This means that the error exponent under $\pi$ is determined by hypothesis $\theta^\dagger$, i.e., $\Phi(\pi)=\Phi(\pi,\theta^\dagger)$. According to definition \eqref{eq: uninformative agent}, we have $\Lambda_k(t;\theta^\dagger)\equiv 0$ for all agents $k\in\mathcal{N}^{{U}}(\theta^\dagger)$. This implies that
	the centrality of the $\theta^\dagger$-uninformative agents has no impact on the value of $\Phi(\pi,\theta^\dagger)$. Recalling the definitions of $t_k^{\sf nc}(\theta^\dagger)$ in \eqref{eq: t_k in non-cooperative scenario} and denoting $t_{\sum}^{\sf nc}=\sum_{k=1}^Nt_k^{\sf nc}(\theta^\dagger)$, we have
	\begin{align}\nonumber
		\Phi(\pi,\theta^\dagger)&\triangleq-\inf_{t\in\mathbb{R}} {\phi(t;\pi,\theta^\dagger)}
		=-\inf_{t\in\mathbb{R}}\sum_{k\notin \mathcal{N}^{ U}(\theta^\dagger)} \phi^{\sf nc}_{k}(\pi_k t;\theta^\dagger)\\
		\nonumber
		&\overset{\text{(a)}}{\geq} -\sum_{k\notin \mathcal{N}^{ U}(\theta^\dagger)} \phi^{\sf nc}_{k}(\pi_k\frac{t_{\sum}^{\sf nc}}{\alpha};\theta^\dagger)\\
		\label{eq: optimality of pi}
		\overset{\eqref{eq: Pi_2-case 1}}&{=} -\sum_{k\notin \mathcal{N}^{ U}(\theta^\dagger)}{\phi^{\sf nc}_{k}(t_k^{\sf nc}(\theta^\dagger);\theta^\dagger)}
		=-\phi({t_{\sum}^{\sf nc}};\pi^\dagger,\theta^\dagger)
	\end{align}
	where (a) follows by letting $t={t_{\sum}^{\sf nc}}/{\alpha}$. From \eqref{eq: pi_dagger for the upper bound-case 1}, we obtain $	\Phi(\pi^\dagger,\theta^\dagger)=-\phi({t_{\sum}^{\sf nc}};\pi^\dagger,\theta^\dagger)={\Phi}_{\sum}^{\sf nc}(\theta^\dagger)$. Then, the optimality of $\pi$ for the error exponent maximization problem \eqref{eq: objective function}--\eqref{eq: constraint 2} is established by 
	\begin{equation}\label{eq: sufficient condition}
		\Phi_{\sum}^{\sf nc}(\theta^\dagger)\overset{\eqref{eq: bound of error exponent}}{\geq}\Phi^\star\geq\Phi(\pi)\overset{\eqref{eq: optimal set-case 1}}{=}\Phi(\pi,\theta^\dagger)\overset{\eqref{eq: optimality of pi}}{\geq}\Phi_{\sum}^{\sf nc}(\theta^\dagger).
	\end{equation}
	This implies that if $\Pi^\dagger$ is nonempty, any Perron eigenvector $\pi\in\Pi^\dagger$ will deliver an error exponent $\Phi_{\sum}^{\sf nc}(\theta^\dagger)$. Therefore, the largest error exponent for the optimization problem \eqref{eq: objective function}--\eqref{eq: constraint 2} is $\Pi^\star=\Phi_{\sum}^{\sf nc}(\theta^\dagger)$. 
	
	Next, we show that in this case (i.e., $\Pi^\dagger\neq\emptyset$), the set $\Pi^\star$ of optimal Perron eigenvectors is given by \eqref{eq: optimal set-case 1}. We noted that in \eqref{eq: sufficient condition}, we have shown that any Perron eigenvector $\pi\in\Pi^\dagger$ is an optimal solution to the error exponent maximization problem \eqref{eq: objective function}--\eqref{eq: constraint 2}. Therefore, $\Pi^\dagger$ is contained in the optimal set $\Pi^\star$. That is,
	\begin{equation}\label{eq: sufficient direction}
		\Pi^\dagger\subseteq\Pi^\star.
	\end{equation}
	It remains to establish the condition in the other direction: $\Pi^\star\subseteq\Pi^\dagger$. 
	To this end, we analyze the necessary conditions for an optimal Perron eigenvector. Since $\Pi$ is the set of all feasible solutions to the optimization problem \eqref{eq: objective function}--\eqref{eq: constraint 2}, we have $\Pi^\star\subseteq\Pi$. By the definition of $\theta^\dagger$ in \eqref{eq: theta_dagger}, the upper bound $\Phi^\star=\Phi_{\sum}^{\sf nc}(\theta^\dagger)$ can be attained if and only if the error exponent $\Phi(\pi)$ is determined by hypothesis $\theta^\dagger$. Therefore, the following inequality
	\begin{equation}
		\Phi(\pi,\theta^\dagger)\leq\Phi(\pi,\theta),\quad \forall \theta\neq\theta_1,
	\end{equation}
	needs to be satisfied for any Perron eigenvector $\pi\in\Pi^\star$. Hence, we obtain a necessary condition:
	\begin{equation}\label{eq: necessary 4}
		\Pi^\star\subseteq\Pi_1\triangleq\left\{\pi\in\Pi:\Phi(\pi,\theta^\dagger)\leq\Phi(\pi,\theta), \forall \theta\neq\theta_1\right\}.
	\end{equation}
	We next prove that the following proportional relation among all $\theta^\dagger$-informative agents, as captured by $\pi^\dagger$ in \eqref{eq: candidate Perron eigenvector-case 1}:
	\begin{equation}\label{eq: proportional relation}
		\pi_k t_\ell^{\sf nc}(\theta^\dagger)=\pi_\ell t_k^{\sf nc}(\theta^\dagger),\quad \forall k,\ell\in\mathcal{N}^{{I}}(\theta^\dagger),
	\end{equation}
	must be satisfied for reaching the upper bound $\Phi_{\sum}^{\sf nc}(\theta^\dagger)$. For any Perron eigenvector $\pi$ that violates \eqref{eq: proportional relation}, there must exist two $\theta^\dagger$-informative agents $m$ and $n$ such that
	\begin{equation}
		\pi_m t_n^{\sf nc}(\theta^\dagger)\neq\pi_n t_m^{\sf nc}(\theta^\dagger).
	\end{equation}
	Then, for any $t\neq 0$, the two equations $\pi_m t= t_m^{\sf nc}(\theta^\dagger)$ and $\pi_n t= t_n^{\sf nc}(\theta^\dagger)$ cannot be both satisfied. Under Assumption \ref{asmp: assumption 3}, there is at least one $\theta^\dagger$-informative agent in the network. Recalling the definition of $t_\theta^\star(\pi)$ in  \eqref{eq: t_theta(pi)}, we obtain $t_{\theta^\dagger}^\star(\pi)<0$ for any Perron eigenvector $\pi\in\Pi$ in this case. Hence, at least one of the following two conditions is true: 
	\begin{equation}
		\pi_m t_{\theta^\dagger}^\star(\pi)\neq t_m^{\sf nc}(\theta^\dagger)\text{ or } \pi_n t_{\theta^\dagger}^\star(\pi)\neq t_n^{\sf nc}(\theta^\dagger).
	\end{equation}
	Assume that $\pi_m t_{\theta^\dagger}^\star(\pi)\neq t_m^{\sf nc}(\theta^\dagger)$ is satisfied, then
	\begin{align}
		\nonumber
		\Phi_m^{\sf nc}(\theta^\dagger)\overset{\eqref{eq: t_k for informative agent k}}&{=}- \phi_m^{\sf nc}(t_m^{\sf nc}(\theta^\dagger);\theta^\dagger)\overset{\eqref{eq: error exponent individually}}{=} -\inf_{t\in\mathbb{R}}\phi_m^{\sf nc}(t;\theta^\dagger)\\
		\label{eq: relation inequality}
		\overset{(\textnormal a)}&{>}  -\phi_m^{\sf nc}(\pi_m t_{\theta^\dagger}^\star(\pi);\theta^\dagger),
	\end{align}
	where (a) is due to the strict convexity of $\phi_m^{\sf nc}(t;\theta)$ provided by property iv) in Lemma \ref{lemma:properties}. Since 
	\begin{align}
		\nonumber
		\Phi(\pi,\theta^\dagger)&\triangleq-\phi(t_{\theta^\dagger}^\star(\pi);\pi,\theta^\dagger)\\
		\nonumber
		&=-\sum_{k\neq m} {\phi^{\sf nc}_{k}(\pi_k t_{\theta^\dagger}^\star(\pi);\theta^\dagger)}-{\phi^{\sf nc}_{m}(\pi_m t_{\theta^\dagger}^\star(\pi);\theta^\dagger)}\\
		\nonumber
		\overset{\eqref{eq: error exponent individually}}&{\leq} \sum_{k\neq m}\Phi_k^{\sf nc}(\theta^\dagger)-	\phi_m^{\sf nc}(\pi_m t_{\theta^\dagger}^\star(\pi);\theta^\dagger) \\
		\label{eq: uniqueness of pi_dagger}
		\overset{\eqref{eq: relation inequality}}&{<}\sum_{k\in\mathcal{N}}\Phi_k^{\sf nc}(\theta^\dagger)	\triangleq\Phi_{\sum}^{\sf nc}(\theta^\dagger),
	\end{align}
	we obtain $\Phi(\pi)\leq\Phi(\pi,\theta^\dagger)<\Phi^\star$, which implies that $\pi\notin\Pi^\star$. Therefore, the proportional relation \eqref{eq: proportional relation} among $\theta^\dagger$-informative agents is necessary for the optimality of a Perron eigenvector. Let $\Pi_I$ be the collection of all Perron eigenvectors that satisfy \eqref{eq: proportional relation}:
	\begin{equation}
		\Pi_I=\left\{\pi: \pi_k t_\ell^{\sf nc}(\theta^\dagger)=\pi_\ell t_k^{\sf nc}(\theta^\dagger), \forall k,\ell\in\mathcal{N}^{{I}}(\theta^\dagger)\right\},
	\end{equation}
	then the following condition holds:
	\begin{equation}\label{eq: necessary 2}
		\Pi^\star\subseteq\Pi_I.
	\end{equation}
	In view of the definition of $\pi^\dagger$ in \eqref{eq: candidate Perron eigenvector-case 1}, $\Pi_I$ can be rewritten as  
	\begin{equation}\label{eq: new Pi_I}
		\Pi_I=\left\{\pi: \pi_k =\alpha \pi_k^\dagger,\forall k\in\mathcal{N}^{{I}}(\theta^\dagger),\forall \alpha>0\right\}\triangleq\Pi_2.
	\end{equation}
	Combining the necessary conditions  \eqref{eq: necessary 4} and \eqref{eq: necessary 2} yields
	\begin{equation}\label{eq: necessary condition}
		\Pi^\star\subseteq\left\{\Pi_1\cap\Pi_2\right\}\triangleq\Pi^\dagger.
	\end{equation}
	Together with \eqref{eq: sufficient direction}, \eqref{eq: necessary condition} gives $\Pi^\star=\Pi^\dagger$. This proves that if $\Pi^\dagger\neq\emptyset$, the optimization problem \eqref{eq: objective function}--\eqref{eq: constraint 2} has the solution: $\Phi^\star=\Phi_{\sum}^{\sf nc}(\theta^\dagger)$, and $\Pi^\star=\Pi^\dagger$. Furthermore, if $\Pi^\dagger=\emptyset$, there is no Perron eigenvector $\pi\in\Pi$ that satisfies both necessary conditions \eqref{eq: necessary 4} and \eqref{eq: necessary 2}. Therefore, the upper bound of the error exponent is not achievable for the given learning task. This completes the proof of Theorem \ref{theorem: optimal Perron eigenvector}.
	
	\subsection{Proof of Theorem \ref{theorem: epsilon-optimal Perron eigenvector}}\label{appendix: theorem 3b}
	In this case, we need to prove that if $\Pi_\epsilon^\dagger$ is not empty, any Perron eigenvector $\pi\in\Pi_\epsilon^\dagger$ is $\epsilon$-optimal. In other words, $\Pi_\epsilon^\dagger\neq\emptyset$ is a sufficient condition for deriving an $\epsilon$-optimal Perron eigenvector. This can be established by repeating our analysis in \eqref{eq: optimality of pi} and \eqref{eq: sufficient condition}. 
	
	\section{Proof of Corollary  \ref{Corollary  3}} \label{appendix: corollary 3}
	In this noisy environment, each agent $k$ will receive a noisy private signal at each time instant $i$, which is denoted by $\widehat{\bm{\xi}}_{k,i}$:
	\begin{equation}\label{eq: noisy observations}
		\widehat{\bm{\xi}}_{k,i}=\bm{\xi}_{k,i}+\bm{n}_{k,i}.
	\end{equation}
	With the shift-in-mean Gaussian model \eqref{eq: mean-shift Gaussian model}, the log-likelihood ratio ${\bm{x}}_{k,i}(\theta)$ of the noisy observation $\widehat{\bm{\xi}}_{k,i}$ is given by
	\begin{equation}
		\begin{aligned}
			{\bm{x}}_{k,i}(\theta)&\triangleq \log\frac{L_k(\widehat{\bm{\xi}}_{k,i}|\theta_1)}{L_k(\widehat{\bm{\xi}}_{k,i}|\theta)}\\
			&=\frac{{\sf m}_k(\theta_1)-{\sf m}_k(\theta)}{\sigma_k^2}\Big({\widehat{\bm{\xi}}}_{k,i}-\frac{{\sf m}_k(\theta)+{\sf m}_k(\theta_1)}{2}\Big)\\
			\overset{\eqref{eq: noisy observations}}&{=}\frac{{\sf m}_k(\theta_1)-{\sf m}_k(\theta)}{\sigma_k^2}\Big(\bm{\xi}_{k,i}+\bm{n}_{k,i}-\frac{{\sf m}_k(\theta)+{\sf m}_k(\theta_1)}{2}\Big).
		\end{aligned}
	\end{equation}
	The random variable ${\bm{x}}_{k,i}(\theta)$ follows a Gaussian distribution: 
	\begin{equation}\label{eq: pdf of noisy log-likelihood ratio}
		\begin{aligned}
			{\bm{x}}_{k,i}(\theta)\sim {\sf N}(a_k(\theta),b_k^2(\theta))
		\end{aligned}
	\end{equation}
	with
	\begin{equation}\label{eq: alpha_k and beta_k}
		\begin{cases}
			a_k(\theta)=\frac{({\sf m}_k(\theta_1)-{\sf m}_k(\theta))^2}{2\sigma^2_k},\\
			b_k^2(\theta)=\frac{({\sf m}_k(\theta_1)-{\sf m}_k(\theta))^2}{\sigma_k^2}\left(1+\varepsilon_k\right).
		\end{cases}
	\end{equation}
	It is straightforward to see that $d_k(\theta)=a_k(\theta)\geq 0$. Since the LMGF of a Gaussian random variable $\bm{y}\sim {\sf N}(a,b)$ is given by $at+bt^2/2$, \eqref{eq: pdf of noisy log-likelihood ratio} yields
		\begin{equation}\label{eq: LMGF of shift-in-mean Gaussian model}
			{\Lambda}_k(t;\theta)=a_k(\theta)t+\frac{b_k^2(\theta)}{2}t^2.
		\end{equation}
	From \eqref{eq: alpha_k and beta_k}, we know that $\Lambda_k(t;\theta)\equiv0$ if $a_k(\theta)=0$. This implies $\mathcal{N}^{{C}}(\theta)=\emptyset$ according to definition \eqref{eq: conflicting agent}. Furthermore, the consistency condition \eqref{eq: consistency condition} is not impacted by the Gaussian noises. That is, 
	\begin{equation}\label{eq: mave>0 in Colollary 3}
		{\sf m_{ave}}(\pi,\theta)>0,\quad\forall\theta\neq\theta_1,
	\end{equation}
	for all Perron eigenvector $\pi$. 
	For any wrong hypothesis $\theta\neq\theta_1$, the quantities $t_k^{\sf nc}(\theta)$ and $\Phi_k^{\sf nc}(\theta)$ of each $\theta$-informative agent $k$ now admit an explicit expression:
	\begin{align}\label{eq: noisy t_k star}
		t_k^{\sf nc}(\theta)&=-\frac{2a_k(\theta)}{b_k^2(\theta)}=-\frac{1}{1+\varepsilon_k},\\ \label{eq: noisy Phi ast}
		{\Phi}_k^{\sf nc}(\theta)&=\frac{a_k(\theta)^2}{b_k^2(\theta)}=\frac{\left({\sf m}_k(\theta_1)-{\sf m}_k(\theta)\right)^2}{4\sigma_k^2(1+\varepsilon_k)}.
	\end{align}
	It can be observed from \eqref{eq: noisy t_k star} that $t_k^{\sf nc}(\theta)$  is determined only by the noise level $\varepsilon_k$, and thus it is identical for all hypothesis $\theta\neq\theta_1$. Moreover, according to \eqref{eq: t_k in non-cooperative scenario}, the critical value $t_k^{\sf nc}(\theta)$ associated with a $\theta$-uninformative agent $k$ can be set to any value. Assume that $t_k^{\sf nc}(\theta)=-1/(1+\varepsilon_k)$ (i.e., following the same expression as \eqref{eq: noisy t_k star}) for each $\theta$-uninformative agent $k\in\mathcal{N}^{U}(\theta)$, then from \eqref{eq: candidate Perron eigenvector-case 1}, we obtain
	\begin{align}\label{eq: optimal Gaussian model}
		\pi_k^\dagger&=\frac{\left(1+\varepsilon_k\right)^{-1}}{\sum_{\ell=1}^N\left(1+\varepsilon_\ell\right)^{-1}},\\
\label{eq: optimal error exponent Gaussian model}
		{\Phi}^{\sf nc}(\theta^\dagger)&=\min_{\theta\neq\theta_1}\sum_{k=1}^N\frac{\left({\sf m}_k(\theta_1)-{\sf m}_k(\theta)\right)^2}{4\sigma_k^2(1+\varepsilon_k)}.
	\end{align}
	It is noted that based on above definitions, the critical value $t_k^{\sf nc}(\theta)$ does not depend on the associated hypothesis $\theta$. This means that the Perron eigenvector $\pi^\dagger$ is feasible to reach the upper bound $\Phi_{\sum}^{\sf nc}(\theta)$ corresponding to any wrong hypothesis $\theta$. That is, by repeating the same arguments used in \eqref{eq: pi_dagger for the upper bound-case 1}, we have
	\begin{equation}\label{eq: achieves all upper bounds in Corollary 3}
		\Phi(\pi^\dagger,\theta)=\Phi_{\sum}^{\sf nc}(\theta),\quad \forall \theta\neq\theta_1.
	\end{equation}
	By definition \eqref{eq: theta_dagger}, we obtain $\Phi(\pi^\dagger,\theta^\dagger)<\Phi(\pi^\dagger,\theta)$, $\forall\theta\neq\theta_1$, and thus $\pi^\dagger\in\Pi_1$. Since $\pi^\dagger$ is always an element of set $\Pi_2$, we have $\Pi^\dagger\neq\emptyset$ and $\pi^\dagger\in\Pi^\dagger$. Therefore, $\pi^\star=\pi^\dagger$ is an optimal solution to the the error exponent maximization problem \eqref{eq: objective function}--\eqref{eq: constraint 2}, with $\Phi^\star=\Phi_{\sum}^{\sf nc}(\theta^\dagger)$ given by \eqref{eq:  optimal error exponent Gaussian model}. 
	
	\section{Proof of Theorem \ref{theorem: adaptation time}}\label{appendix: adaptation time}
	Under the uniform initial belief condition (i.e., $\bm{\lambda}_{k,0}(\theta)=0$ for all $k\in\mathcal{N}$ and $\theta\in\Theta$), we can obtain from \eqref{eq: ASL expanded recursion} that
	\begin{align}
		\nonumber
		\bm{\lambda}_{k,i}(\theta)&=\delta\sum_{m=0}^{i-1}\sum_{\ell=1}^N(1-\delta)^m[A^{m+1}]_{\ell k}\bm{x}_{\ell,i-m}(\theta)\\
		&=\widehat{\bm{\lambda}}_{k,i}(\theta) \overset{\text{d}}{=}\widetilde{\bm{\lambda}}_{k,i}(\theta)
	\end{align}
	where $\widehat{\bm{\lambda}}_{k,i}(\theta)$ and $\widetilde{\bm{\lambda}}_{k,i}(\theta)$ are respectively defined in \eqref{eq: lambda_hat} and \eqref{eq: lambda_tilde}, and the equality in distribution is due to \eqref{eq: equality in distribution}. According to \eqref{eq: instantaneous error probability definition}, the instantaneous error probability of agent $k$ at time instant $i$ associated with $\theta\neq \theta_1$ is given by
	\begin{align}
		\nonumber
		\mathbb{P}\left[\bm{\lambda}_{k,i}(\theta)\leq 0\right]
		&=\mathbb{P}\left[\widetilde{\bm{\lambda}}_{k,i}(\theta)\leq 0\right]\\\nonumber
		\overset{\text{(a)}}&{=}\;\mathbb{P}\left[\frac{t_\theta^\star(\pi)}{\delta}\widetilde{\bm{\lambda}}_{k,i}(\theta)\geq0\right]\\
		\label{eq: error probability theta}
		\overset{\text{(b)}}&{\leq}\mathbb{E}\left[\exp(\frac{t_\theta^\star(\pi)}{\delta}\widetilde{\bm{\lambda}}_{k,i}(\theta))\right],
	\end{align}
	where (a) is due to $t_\theta^\star(\pi)<0$ from  \eqref{eq: t_theta(pi)} and (b) follows by applying Markov's inequality. Denoting the LMGF of $\widetilde{\bm{\lambda}}_{k,i}(\theta)$ as 
	\begin{equation}
		\Lambda_{k,i}(t;\theta)\triangleq \log\mathbb{E}\left[e^{t\widetilde{\bm{\lambda}}_{k,i}(\theta)}\right],
	\end{equation}
	the independence assumption of the local observations over time and space gives
	\begin{equation}\label{eq: Lambda_{k,i}}
		\Lambda_{k,i}(t;\theta)=\sum_{m=0}^{i-1}\sum_{\ell=1}^N \Lambda_{\ell}\left(\delta(1-\delta)^m[A^{m+1}]_{\ell k}t;\theta\right).
	\end{equation}
	Let $b_i=[A^{i+1}]_{\ell k}$ and using  Eqs. (85) and (86) from \cite{Matta:2016cb}, we have
	\begin{align}
		\nonumber
		\Lambda_{k,i}(t;\theta)\overset{\eqref{eq: Lambda_{k,i}}}&{=}\sum_{\ell=1}^N\sum_{m=0}^{i-1} \Lambda_{\ell}\left(\delta(1-\delta)^m[A^{m+1}]_{\ell k}t;\theta\right)\\
		\nonumber
		&=\sum_{\ell=1}^N \frac{1}{\delta}\left[\int_{\delta(1-\delta)^i t}^{\delta  t}\frac{\Lambda_\ell(\pi_\ell\tau;\theta)}{\tau}d\tau+\mathcal{O}(\delta)\right]\\
		\label{eq: LMGF2}
		&=\frac{1}{\delta}\left[\int_{\delta(1-\delta)^i t}^{\delta  t}\frac{\Lambda_{\sf ave}(\tau;\pi,\theta)}{\tau}d\tau+\mathcal{O}(\delta)\right].
	\end{align}
	Let $t=t_\theta^\star(\pi)/\delta$, we obtain the following expression:
	\begin{equation}\label{eq: LMGF 3}
		\begin{aligned}
			\Lambda_{k,i}\left(\frac{t_\theta^\star(\pi)}{\delta};\theta\right)=\frac{1}{\delta}\Bigg[\int_{(1-\delta)^{i} t_\theta^\star(\pi)}^{t^\star_\theta(\pi)}\frac{\Lambda_{\sf ave}(\tau;\pi,\theta)}{\tau}d\tau+\mathcal{O}(\delta)\Bigg].
		\end{aligned}
	\end{equation}
	Applying the parabolic approximation \eqref{eq: second-order polynomial} in the low SNR regime, we have
	\begin{equation}
		\frac{\Lambda_{\sf ave}(t;\pi,\theta)}{t}\approx \kappa_1(\pi,\theta)+\frac{\kappa_2(\pi,\theta)}{2}t	
	\end{equation}
	for $t\in[t^\star_\theta(\pi),0]$. Therefore, we can approximate $t^\star_\theta(\pi)$ and $\Phi(\pi,\theta)$ as
	\begin{align}
		\widehat{t^\star_\theta}(\pi)&\triangleq-\frac{2\kappa_1(\pi,\theta)}{\kappa_2(\pi,\theta)}=-\frac{2{\sf m_{ave}}(\pi,\theta)}{{\sf c_{ave}(\pi,\theta)}}, \label{eq: approximated t}\\
		\widehat{\Phi}(\pi,\theta)&\triangleq-\frac{\kappa_1(\pi,\theta)^2}{\kappa_2(\pi,\theta)}=-\frac{{\sf m}_{\sf ave}(\pi,\theta)^2}{{\sf c_{ave}(\pi,\theta)}}. \label{eq: approximated phi}
	\end{align}
	\noindent With the expressions of $\widehat{t^\star_\theta}(\pi)
	$ and $\widehat{\Phi}(\pi,\theta)$, an approximation for \eqref{eq: LMGF 3} is obtained: 
	\begin{align}
		\nonumber
		&\quad\;\Lambda_{k,i}\left(\frac{t_\theta^\star(\pi)}{\delta};\theta\right)\\
		\nonumber
		&\approx  \frac{1}{\delta}\left[\int_{(1-\delta)^{i} \widehat{t_\theta^\star}(\pi)}^{\widehat{t^\star_\theta}(\pi)}\left\{ \kappa_1(\pi,\theta)+\frac{\kappa_2(\pi,\theta)}{2}\tau \right\} d\tau+\mathcal{O}(\delta)\right]\\
		\nonumber
		&= -\frac{1}{\delta}\left\{\left[1-(1-\delta)^{i}\right]^2\widehat{\Phi}(\pi,\theta)+\mathcal{O}(\delta)\right\}\\
		\label{eq: LMGF2 approximation}
		&\approx -\frac{1}{\delta}\Bigg\{\underbrace{\left[1-(1-\delta)^{i}\right]^2}_{\text{denoted as }\eta(\delta,i)}\Phi(\pi,\theta)+\mathcal{O}(\delta)\Bigg\}.
	\end{align}
	Combining \eqref{eq: error probability theta} and \eqref{eq: LMGF2 approximation}, the instantaneous error probability of agent $k$ at time instant $i$ associated with $\theta\neq\theta_1$ is upper bounded by
	\begin{align}
		\nonumber
		\mathbb{P}\left[\bm{\lambda}_{k,i}(\theta)\leq 0\right]
		&\leq  \exp\left[\Lambda_{k,i}\left(\frac{t_\theta^\star(\pi)}{\delta};\theta\right)\right]\\\label{eq: instantaneous error probability approximation}
		&\approx \exp\left\{-\frac{1}{\delta}\big[\eta(\delta,i)\Phi(\pi,\theta)+\mathcal{O}(\delta)\big]\right\}.
	\end{align}
	For a given small $\omega>0$ in \eqref{eq: inst error pb ub}, we let $\eta(\delta,i)\geq1-\omega$ and then obtain
	\begin{equation}\label{eq: i inequality}
		i\geq \dfrac{\log(1-\sqrt{1-\omega})}{\log(1-\delta)}.
	\end{equation}
	The right-hand side of \eqref{eq: i inequality} is exactly the definition of ${\sf T_{adap}}(\omega)$ given in Theorem \ref{theorem: adaptation time}.
	Using Boole's inequality, the instantaneous error probability $p_{k,i}$ satisfies
	\begin{align}
		\nonumber
		p_{k,i} &=\mathbb{P}\Bigg[\bigcup_{\theta\neq\theta_1}{\bm{\lambda}}_{k,i}(\theta)
		\leq 0\Bigg]\\
		\nonumber
		&\leq\sum_{\theta\neq\theta_1}\mathbb{P}\left[{\bm{\lambda}}_{k,i}(\theta)\leq 0\right]\\
		\nonumber
		\overset{\eqref{eq: instantaneous error probability approximation}}&{\approx} \sum_{\theta\neq\theta_1}\exp\left\{-\frac{1}{\delta}\left[\eta(\delta,i)\Phi(\pi,\theta)+\mathcal{O}(\delta)\right]\right\}\\
		\nonumber
		\overset{\eqref{eq: i inequality}}&{\leq} \sum_{\theta\neq\theta_1}{e^{-\frac{1}{\delta}\big[(1-\omega)\Phi(\pi,\theta)+\mathcal{O}(\delta)\big]}}\\
		&\doteq e^{-\frac{1}{\delta}\big[(1-\omega)\Phi(\pi)+\mathcal{O}(\delta)\big]}
	\end{align}
	for $i\geq {\sf T_{adap}}(\omega)$, which shows that the condition in \eqref{eq: inst error pb ub} is satisfied. Hence, ${\sf T_{adap}}(\omega)$ provides a reasonable estimate for the adaptation time of the ASL strategy in the low SNR regime. 
	
\bibliographystyle{IEEEtran}
\bibliography{Ref}

\begin{thebibliography}{10}
\providecommand{\url}[1]{#1}
\csname url@samestyle\endcsname
\providecommand{\newblock}{\relax}
\providecommand{\bibinfo}[2]{#2}
\providecommand{\BIBentrySTDinterwordspacing}{\spaceskip=0pt\relax}
\providecommand{\BIBentryALTinterwordstretchfactor}{4}
\providecommand{\BIBentryALTinterwordspacing}{\spaceskip=\fontdimen2\font plus
\BIBentryALTinterwordstretchfactor\fontdimen3\font minus
  \fontdimen4\font\relax}
\providecommand{\BIBforeignlanguage}[2]{{%
\expandafter\ifx\csname l@#1\endcsname\relax
\typeout{** WARNING: IEEEtran.bst: No hyphenation pattern has been}%
\typeout{** loaded for the language `#1'. Using the pattern for}%
\typeout{** the default language instead.}%
\else
\language=\csname l@#1\endcsname
\fi
#2}}
\providecommand{\BIBdecl}{\relax}
\BIBdecl

\bibitem{Ping2022:Optimal}
P.~Hu, V.~Bordignon, S.~Vlaski, and A.~H. Sayed, ``Optimal combination policies
  for adaptive social learning,'' in \emph{Proc. IEEE ICASSP}, Singapore, May
  2022, pp. 1--5.

\bibitem{Molavi:2018kq}
P.~Molavi, A.~Tahbaz-Salehi, and A.~Jadbabaie, ``A theory of non-{B}ayesian
  social learning,'' \emph{Econometrica}, vol.~86, no.~2, pp. 445--490, 2018.

\bibitem{Jadbabaie:2012ii}
A.~Jadbabaie, P.~Molavi, A.~Sandroni, and A.~Tahbaz-Salehi, ``Non-{B}ayesian
  social learning,'' \emph{Games and Economic Behavior}, vol.~76, no.~1, pp.
  210--225, 2012.

\bibitem{Zhao:2013wa}
X.~Zhao and A.~H. Sayed, ``Learning over social networks via diffusion
  adaptation,'' in \emph{Proc. Asilomar Conference on Signals, Systems and
  Computers}, Pacific Grove, USA, 2012, pp. 709--713.

\bibitem{Salami:2017jf}
H.~Salami, B.~Ying, and A.~H. Sayed, ``Social learning over weakly connected
  graphs,'' \emph{IEEE Transactions on Signal and Information Processing over
  Networks}, vol.~3, no.~2, pp. 222--238, 2017.

\bibitem{Lalitha:2018ej}
A.~Lalitha, T.~Javidi, and A.~D. Sarwate, ``Social learning and distributed
  hypothesis testing,'' \emph{IEEE Transactions on Information Theory},
  vol.~64, no.~9, pp. 6161--6179, 2018.

\bibitem{Nedic:2017bg}
A.~Nedic, A.~Olshevsky, and C.~A. Uribe, ``Fast convergence rates for
  distributed non-{B}ayesian learning,'' \emph{IEEE Transactions on Automatic
  Control}, vol.~62, no.~11, pp. 5538--5553, 2017.

\bibitem{Hare2021General}
J.~Z. Hare, C.~A. Uribe, L.~Kaplan, and A.~Jadbabaie, ``A general framework for
  distributed inference with uncertain models,'' \emph{IEEE Transactions on
  Signal and Information Processing over Networks}, vol.~7, pp. 392--405, 2021.

\bibitem{Matta2020Interplay}
V.~Matta, V.~Bordignon, A.~Santos, and A.~H. Sayed, ``Interplay between
  topology and social learning over weak graphs,'' \emph{IEEE Open Journal of
  Signal Processing}, vol.~1, pp. 99--119, 2020.

\bibitem{shi2020distributed}
C.-X. Shi and G.-H. Yang, ``Distributed learning over networks: Effect of using
  historical observations,'' \emph{IEEE Transactions on Automatic Control},
  vol.~65, no.~12, pp. 5503--5509, 2020.

\bibitem{Mitra:2020fi}
A.~Mitra, J.~A. Richards, and S.~Sundaram, ``A new approach to distributed
  hypothesis testing and non-{B}ayesian learning: Improved learning rate and
  byzantine-resilience,'' \emph{IEEE Transactions on Automatic Control},
  vol.~66, no.~9, pp. 4084--4100, 2021.

\bibitem{Bordignon:2020wx}
V.~Bordignon, V.~Matta, and A.~H. Sayed, ``Adaptive social learning,''
  \emph{IEEE Transactions on Information Theory}, vol.~67, no.~9, pp.
  6053--6081, 2021.

\bibitem{jadbabaie2013information}
A.~Jadbabaie, P.~Molavi, and A.~Tahbaz-Salehi, ``Information heterogeneity and
  the speed of learning in social networks,'' \emph{Columbia Business School
  Research Paper}, no. 13-28, 2013.

\bibitem{Shahrampour:2016kk}
S.~Shahrampour, A.~Rakhlin, and A.~Jadbabaie, ``Distributed detection:
  Finite-time analysis and impact of network topology,'' \emph{IEEE
  Transactions on Automatic Control}, vol.~61, no.~11, pp. 3256--3268, 2016.

\bibitem{Viswanathan1997:Distirbuted}
R.~Viswanathan and P.~Varshney, ``Distributed detection with multiple sensors:
  {P}art {I}--{F}undamentals,'' \emph{Proceedings of the IEEE}, vol.~85, no.~1,
  pp. 54--63, 1997.

\bibitem{Blum1997:Distributed}
R.~Blum, S.~Kassam, and H.~Poor, ``Distributed detection with multiple sensors:
  Part {II}--{A}dvanced topics,'' \emph{Proceedings of the IEEE}, vol.~85,
  no.~1, pp. 64--79, 1997.

\bibitem{Chamberland2003:Decentralized}
J.-F. Chamberland and V.~Veeravalli, ``Decentralized detection in sensor
  networks,'' \emph{IEEE Transactions on Signal Processing}, vol.~51, no.~2,
  pp. 407--416, 2003.

\bibitem{matta2018estimation}
V.~Matta and A.~H. Sayed, ``Estimation and detection over adaptive networks,''
  in \emph{Cooperative and Graph Signal Processing}, P.~M. Djurić and
  C.~Richard, Eds.\hskip 1em plus 0.5em minus 0.4em\relax Amsterdam, The
  Netherlands: Elsevier, 2018, pp. 69--106.

\bibitem{Bajovic2011:Distributed}
D.~Bajovic, D.~Jakovetic, J.~Xavier, B.~Sinopoli, and J.~M.~F. Moura,
  ``Distributed detection via {G}aussian running consensus: Large deviations
  asymptotic analysis,'' \emph{IEEE Transactions on Signal Processing},
  vol.~59, no.~9, pp. 4381--4396, 2011.

\bibitem{Bajovic2012:Large}
D.~Bajovic, D.~Jakovetic, J.~M.~F. Moura, J.~Xavier, and B.~Sinopoli, ``Large
  deviations performance of consensus+innovations distributed detection with
  non-{G}aussian observations,'' \emph{IEEE Transactions on Signal Processing},
  vol.~60, no.~11, pp. 5987--6002, 2012.

\bibitem{Bajovic2016:Distributed}
D.~Bajovic, J.~M.~F. Moura, J.~Xavier, and B.~Sinopoli, ``Distributed inference
  over directed networks: Performance limits and optimal design,'' \emph{IEEE
  Transactions on Signal Processing}, vol.~64, no.~13, pp. 3308--3323, 2016.

\bibitem{Matta:2016kh}
V.~Matta, P.~Braca, S.~Marano, and A.~H. Sayed, ``Diffusion-based adaptive
  distributed detection: Steady-state performance in the slow adaptation
  regime,'' \emph{IEEE Transactions on Information Theory}, vol.~62, no.~8, pp.
  4710--4732, 2016.

\bibitem{Matta:2016cb}
------, ``Distributed detection over adaptive networks: Refined asymptotics and
  the role of connectivity,'' \emph{IEEE Transactions on Signal and Information
  Processing over Networks}, vol.~2, no.~4, pp. 442--460, 2016.

\bibitem{Marano2021:Decision}
S.~Marano and A.~H. Sayed, ``Decision learning and adaptation over multi-task
  networks,'' \emph{IEEE Transactions on Signal Processing}, vol.~69, pp.
  2873--2887, 2021.

\bibitem{sayed2014adaptation}
A.~H. Sayed, ``Adaptation, learning, and optimization over networks,''
  \emph{Foundations and Trends in Machine Learning}, vol.~7, no. 4-5, pp.
  311--801, 2014.

\bibitem{newman2018networks}
M.~Newman, \emph{Networks}.\hskip 1em plus 0.5em minus 0.4em\relax Oxford
  University Press, 2018.

\bibitem{lewis2011network}
T.~G. Lewis, \emph{Network Science: Theory and Applications}.\hskip 1em plus
  0.5em minus 0.4em\relax John Wiley \& Sons, 2011.

\bibitem{sayed2014adaptive}
A.~H. Sayed, ``Adaptive networks,'' \emph{Proceedings of the IEEE}, vol. 102,
  no.~4, pp. 460--497, 2014.

\bibitem{cover1999elements}
T.~M. Cover and J.~A. Thomas, \emph{Elements of Information Theory}.\hskip 1em
  plus 0.5em minus 0.4em\relax Hoboken, NJ, USA: Wiley, 1991.

\bibitem{koliander2022fusion}
G.~Koliander, Y.~El-Laham, P.~M. Djuri{\'c}, and F.~Hlawatsch, ``Fusion of
  probability density functions,'' \emph{Proceedings of the IEEE}, vol. 110,
  no.~4, pp. 404--453, 2022.

\bibitem{shumovskaia2023discovering}
V.~Shumovskaia, M.~Kayaalp, M.~Cemri, and A.~H. Sayed, ``Discovering
  influencers in opinion formation over social graphs,'' \emph{IEEE Open
  Journal of Signal Processing}, pp. 1--20, 2023.

\bibitem{bordignon2021learning}
V.~Bordignon, S.~Vlaski, V.~Matta, and A.~H. Sayed, ``Learning from
  heterogeneous data based on social interactions over graphs,'' \emph{IEEE
  Transactions on Information Theory}, vol.~69, no.~5, pp. 3347--3371, 2023.

\bibitem{dembo1998applications}
A.~Dembo and O.~Zeitouni, \emph{Large Deviations Techniques and
  Applications}.\hskip 1em plus 0.5em minus 0.4em\relax Springer, 1998.

\bibitem{den2008large}
F.~Den~Hollander, \emph{Large Deviations}.\hskip 1em plus 0.5em minus
  0.4em\relax American Mathematical Society, 2008.

\bibitem{rohwer2015convergence}
C.~M. Rohwer, F.~Angeletti, and H.~Touchette, ``Convergence of large-deviation
  estimators,'' \emph{Physical Review E}, vol.~92, no.~5, p. 052104, 2015.

\bibitem{nielsen2013information}
F.~Nielsen, ``An information-geometric characterization of {C}hernoff
  information,'' \emph{IEEE Signal Processing Letters}, vol.~20, no.~3, pp.
  269--272, 2013.

\bibitem{sayed2022inference}
A.~H. Sayed, \emph{Inference and Learning from Data}.\hskip 1em plus 0.5em
  minus 0.4em\relax Cambridge University Press, 2022.

\bibitem{chu2005inverse}
M.~Chu and G.~Golub, \emph{{I}nverse {E}igenvalue {P}roblems: {T}heory,
  {A}lgorithms, and {A}pplications}.\hskip 1em plus 0.5em minus 0.4em\relax
  Oxford University Press, June 2005.

\bibitem{borges1995some}
C.~F. Borges, R.~Frezza, and W.~B. Gragg, ``Some inverse eigenproblems for
  {J}acobi and arrow matrices,'' \emph{Numerical Linear Algebra with
  Applications}, vol.~2, no.~3, pp. 195--203, 1995.

\bibitem{chu1994symmetric}
M.~T. Chu and M.~A. Erbrecht, ``Symmetric {T}oeplitz matrices with two
  prescribed eigenpairs,'' \emph{SIAM Journal on Matrix Analysis and
  Applications}, vol.~15, no.~2, pp. 623--635, 1994.

\bibitem{hu2008systematic}
S.-L.~J. Hu and H.~Li, ``A systematic linear space approach to solving
  partially described inverse eigenvalue problems,'' \emph{Inverse Problems},
  vol.~24, no.~3, p. 035014, 2008.

\bibitem{vlaski2021regularized}
S.~Vlaski, L.~Vandenberghe, and A.~H. Sayed, ``Regularized diffusion adaptation
  via conjugate smoothing,'' \emph{IEEE Transactions on Automatic Control},
  vol.~67, no.~5, pp. 2343--2358, 2022.

\bibitem{braca2009asymptotic}
P.~Braca, S.~Marano, V.~Matta, and P.~Willett, ``Asymptotic optimality of
  running consensus in testing binary hypotheses,'' \emph{IEEE Transactions on
  Signal Processing}, vol.~58, no.~2, pp. 814--825, 2009.

\bibitem{kassam2012signal}
S.~A. Kassam, \emph{Signal Detection in Non-{G}aussian Noise}.\hskip 1em plus
  0.5em minus 0.4em\relax Springer, 1987.

\bibitem{uribe2022nonasymptotic}
C.~A. Uribe, A.~Olshevsky, and A.~Nedi{\'c}, ``Nonasymptotic concentration
  rates in cooperative learning—{P}art {II}: Inference on compact hypothesis
  sets,'' \emph{IEEE Transactions on Control of Network Systems}, vol.~9,
  no.~3, pp. 1141--1153, 2022.

\end{thebibliography}

\end{document}